\documentclass[11pt, onecolumn]{IEEEtran}

\usepackage{macro}

\title{Generalized List Decoding}
\author{
\IEEEauthorblockN{
Yihan Zhang\IEEEauthorrefmark{1},
Amitalok J. Budkuley\IEEEauthorrefmark{1},
Sidharth Jaggi\IEEEauthorrefmark{1}
}\\
\IEEEauthorblockA{
Department of Information Engineering, The Chinese University of Hong Kong\\
\IEEEauthorrefmark{1}\href{mailto:zy417@ie.cuhk.edu.hk}{zy417@ie.cuhk.edu.hk},
\IEEEauthorrefmark{1}\href{mailto:amitalok86@gmail.com}{amitalok86@gmail.com},
\IEEEauthorrefmark{1}\href{mailto:jaggi@ie.cuhk.edu.hk}{jaggi@ie.cuhk.edu.hk}
}}

\begin{document}
\maketitle
\footnotetext[1]{
}

\begin{abstract}
This paper concerns itself with the question of list decoding for \emph{general adversarial channels}, e.g., bit-flip ($\XOR$) channels, erasure channels, $\AND$ ($Z$-) channels, $\OR$ ($\reflectbox{Z}$-) channels, real adder channels, noisy typewriter channels, etc. 
We precisely \emph{characterize} when exponential-sized  (or positive \emph{rate}) \emph{$(L-1)$-list decodable} codes (where the \emph{list size} $L$ is a universal constant) exist for such channels. Our criterion asserts that:
\begin{quote}
For any given general adversarial channel, it is possible to construct positive rate $(L-1)$-list decodable codes \emph{if and only if} the set of \emph{completely positive tensors} of order-$L$ with admissible marginals is not entirely contained in the order-$L$ \emph{confusability set} associated to the channel.
\end{quote}
The sufficiency is shown via random code construction (combined with expurgation or time-sharing). The necessity is shown by 
\begin{enumerate}
	\item extracting equicoupled  subcodes (generalization of equidistant code) from \emph{any} large code sequence using hypergraph Ramsey's theorem, and
	\item significantly extending the classic \emph{Plotkin bound} in coding theory to list decoding  for general channels using duality between the completely positive tensor cone and the \emph{copositive} tensor cone.
\end{enumerate}
In the proof, we also obtain a new fact regarding asymmetry of joint distributions, which be may of independent interest.

Other results include
\begin{enumerate}
\item List decoding capacity with asymptotically large $L$ for general adversarial channels;
\item A \emph{tight} list size bound for \emph{most} \emph{constant composition} codes (generalization of constant weight codes);
\item Rederivation and demystification of Blinovsky's \cite{blinovsky-1986-ls-lb-binary} characterization of the list decoding \emph{Plotkin points} (threshold at which large codes are impossible);
\item Evaluation of general bounds (\cite{wang-budkuley-bogdanov-jaggi-2019-omniscient-avc}) for \emph{unique decoding} in the error correction code setting.
\end{enumerate}
\end{abstract}

\newpage 
\tableofcontents
\newpage

\section{Warmup}
\label{sec:warmup}
In favour of   introducing general notions, motivating general problems  and stating our general theorems,  we first go through  concrete numerical examples that are special cases of our results. 

Suppose Alice can transmit a length-$n$ bit string (\emph{codeword}) to Bob and an adversary James can flip $np$ ($0\le p\le 1$) of these bits. Consider first the classic coding theory question.
\begin{enumerate}
	\item\label{itm:question_ecc} \textbf{Error correction.} For what values of $p$, can one construct a \emph{code} (collection  of codewords) of \emph{positive rate} (i.e., size at least $2^{R n}$ for some constant $R>0$) such that Bob can uniquely decode? The classic Plotkin bound tells us that this is impossible for $p>1/4$,\footnote{Actually for $p=1/4$ this is still impossible} and the classic Gilbert--Varshamov (GV) bound tells us that this is possible for $p<1/4$.
	\item\label{itm:question_ldc} \textbf{List decoding.} For what values of $p$, can one construct a code of positive rate such that is \emph{3-list decodable}  (i.e., regardless of which $np$ bits James flips,  Bob can always decode the received word to a \emph{list} of at most 3 codewords, one of which is the codeword transmitted by Alice)?\footnote{Note that a 1-list decodable code is exactly a uniquely decodable code (or more commonly called an \emph{error correction code}).} Due to work by Blinovsky, it is known that this is possible if and only if $p\le 5/16$.\footnote{In fact Blinovsky identified the threshold $p$ up to which positive rate $(L-1)$-list decodable codes exist for \emph{any}  integer $L\ge2$. This, in particular, recovers the Plotkin bound.}
\end{enumerate}
In this work, we are able to rederive all the above thresholds, but are also able to derive the corresponding thresholds for a vast variety of \emph{general adversarial channels}, such as, bit-flip  channels, erasure channels, $\AND$ ($Z$-) channels, $\OR$ ($\reflectbox{Z}$-) channels,  adder channels, noisy typewriter channels, etc. 

In this section, let us revisit the answers to questions \ref{itm:question_ecc} and \ref{itm:question_ldc} in the technical language we develop in this paper. 
\begin{enumerate}
	\item \textbf{Error correction.} Consider any pair of codewords $\vx_1,\vx_2$ that are resilient to $np$ bit-flips. They must therefore be at a Hamming distance larger than $2np$. Said differently, the \emph{joint type} (i.e., the $2\times 2$ matrix whose $(x_1,x_2)$-th entry is the fraction of locations $i$ of $(\vx_1,\vx_2)$ such that $\vx_1(i) = x_1$ and $\vx_2(i) = x_2$) $\tau_{\vx_1,\vx_2} = \begin{bmatrix}t(0,0)&t(0,1)\\t(1,0)&t(1,1)\end{bmatrix}$ of these two codewords  must satisfy 	the condition that 
	\begin{enumerate}
		\item[\textbf{C1}]\label{itm:conf_cond_ex_one} $t(0,1)+t(1,0)\ge2p$. 
	\end{enumerate}
	\begin{enumerate}
		\item In \cite{blinovsky-1986-ls-lb-binary,polyanskiy-2016-ld-ub,alon-bukh-polyanskiy-2018-ld-zero-rate}\footnote{Their and our work showed that it is also possible to find a positive rate subcode such that every \emph{$L$-tuple} of codewords has joint type close to some $P_{\bfx_1,\cdots,\bfx_L}$. This, as we shall see momentarily, is useful for list decoding.} and \cite{wang-budkuley-bogdanov-jaggi-2019-omniscient-avc}, it was shown that:  if a code $\cC$ of size $2^{R n}$ exists, then there must exist a \emph{positive rate} subcode $\cC'\subset\cC$ such that for every pair of codewords $\vx_1,\vx_2$ in $\cC'$, their joint type is approximately the same (as, say, $P_{\bfx_1,\bfx_2}$). 
		\item In \cite{wang-budkuley-bogdanov-jaggi-2019-omniscient-avc}, it was shown that: it is possible to construct positive rate codes with joint types (close to) $P_{\bfx_1,\bfx_2}$ if and only if $P_{\bfx_1,\bfx_2}$ is a \emph{completely positive ($\cp$)} distribution, i.e., joint distributions that can be written as a convex combination of products of independent and identical distributions,
		\[P_{\bfx_1,\bfx_2} = \sum_{i = 1}^r \lambda_iP_{\bfx_i}P_{\bfx_i}^\top,\]
		for some positive integer $k$, convex combination  coefficients $\curbrkt{\lambda_i}_{1\le i\le k}$ and probability vectors $\curbrkt{P_{\bfx_i}}_{1\le i\le k}$. For example, 
		\begin{equation}
		\label{eqn:ex_cp}
		\lambda\begin{bmatrix}1/2&0\\0&1/2\end{bmatrix} + (1-\lambda)\begin{bmatrix}1/4&1/4\\1/4&1/4\end{bmatrix}
		\end{equation}
		is $\cp$ for $\lambda\in[0,1]$ since it can be written as $\frac{\lambda}{2}\begin{bmatrix}1 & 0\end{bmatrix}\begin{bmatrix}0\\1\end{bmatrix} + \frac{\lambda}{2}\begin{bmatrix}0 & 1\end{bmatrix}\begin{bmatrix}0\\1\end{bmatrix} + (1-\lambda)\begin{bmatrix}1/2&1/2\end{bmatrix}\begin{bmatrix}1/2\\1/2\end{bmatrix}$.
		One can check that for $\lambda<0$, matrix \eqref{eqn:ex_cp} is not $\cp$.
		For  condition \textbf{C1} to be satisfied by some $\cp$ distribution, it must be the case that  $2p\le2\cdot(1-\lambda)\cdot(1/4)$ for some $\lambda\in[0,1]$. This is impossible if $p>1/4$.
		As a consequence,   the classic Plotkin bound is recovered in this convex geometry language, since the non-$\cp$ matrices of the form \eqref{eqn:ex_cp} with \emph{negative} $\lambda$ correspond to codes with minimum pairwise fractional distance $\frac{1+\abs{\lambda}}{2}$ (hence correspond to $p=\frac{1+\abs{\lambda}}{4}>1/4$), which, by the Plotkin bound, cannot have positive rate. 
	\end{enumerate}
	\item \textbf{List decoding.} Now let us move to the list decoding question in hands.  For a code to be 3-list decodable, it must be the case that for any quadruple $\vx_1,\vx_2,\vx_3,\vx_4$, there is no $\vy$ such that the Hamming distance from $\vx_i$ to $\vy$ is at most $np$ for \emph{every} $i \in\curbrkt{1,2,3,4}$.  In this case, the appropriate object is therefore a $2\times 2\times 2\times 2$ tensor (or a joint distribution of $(\bfx_1,\bfx_2,\bfx_3,\bfx_4)$) $P_{\bfx_1,\bfx_2,\bfx_3,\bfx_4}$ 
	such that  
	\begin{enumerate}
		\item[\textbf{C2}] \label{itm:conf_cond_ex_two} any its \emph{extension} $P_{\bfx_1,\bfx_2,\bfx_3,\bfx_4,\bfy}$ (i.e., a coupling of $(\bfx_1,\bfx_2,\bfx_3,\bfx_4)$ and $\bfy$, or a $2\times2\times2\times2\times2$ tensor  such that $P_{\bfx_1,\bfx_2,\bfx_3,\bfx_4} = P_{\bfx_1,\bfx_2,\bfx_3,\bfx_4,0}+P_{\bfx_1,\bfx_2,\bfx_3,\bfx_4,1}$) satisfies the condition that $P_{\bfx_i,\bfy}(0,1)+P_{\bfx_i,\bfy}(1,0)>p$ for at least one $i\in\curbrkt{1,2,3,4}$.
	\end{enumerate}
	\begin{enumerate}
		\item Again, by \cite{blinovsky-1986-ls-lb-binary,polyanskiy-2016-ld-ub,alon-bukh-polyanskiy-2018-ld-zero-rate} and our work, we can restrict our attention to codes in which every $L$-tuple of codewords has joint type close to some $P_{\bfx_1,\cdots,\bfx_L}$, since we can find  such a subcode which is sufficiently large in \emph{any} positive rate code.
		\item Generalizing \cite{wang-budkuley-bogdanov-jaggi-2019-omniscient-avc}, we show that codes with order-$4$ joint types (close to) $P_{\bfx_1,\bfx_2,\bfx_3,\bfx_4}$ if and only if $P_{\bfx_1,\bfx_2,\bfx_3,\bfx_4}$ is a \emph{completely positive tensor} of order-4, i.e., joint distributions that can be written as a convex combination of products of independent and identical distributions,
		\[P_{\bfx_1,\cdots,\bfx_L} = \sum_{i = 1}^k\lambda_iP_{\bfx_i}^{\otimes 4}.\]
		One can check that distributions of the form
		\[\lambda\,\diag(1/2)+(1+\lambda)\inputdistrunif^{\otimes4} = \frac{\lambda}{2}\begin{bmatrix}1\\0\end{bmatrix}^{\otimes4} + \frac{\lambda}{2}\begin{bmatrix}0\\1\end{bmatrix}^{\otimes4}+(1-\lambda)\inputdistrunif^{\otimes4}\]
		is $\cp$ if and only if $\lambda\in[0,1]$. On the other hand, for condition \textbf{C2} to be satisfied by some tensor like this, it turns out, as shown by Blinovsky \cite{blinovsky-1986-ls-lb-binary} and us, that $p$ has to be no larger than $5/16$.
	\end{enumerate}
\end{enumerate} 
Of course, bit-flips are just one of the simplest models of corruption that may occur in real-world communication/storage systems. Perhaps, under certain circumstances, in the system, we are allowed to transmit length-$n$ codewords taking values from $\curbrkt{0,1,2,3,4,5}$, but each legitimate codeword $\vx$ has to satisfy the following constraints inherently associated to the system
	\begin{align*}
	\begin{cases}
	\begin{array}{lllllll}
		&\tau_\vx(1)&&+3\tau(3)&&&\le1.2\\
		&&\tau_\vx(2) &-\tau_\vx(3)&&&\ge0.05\\
		\tau_\vx(0)&&&&-\tau_\vx(4)&-0.2\tau_\vx(5)&\le0.7
	\end{array},
	\end{cases}
	\end{align*}
	where $\tau_\vx(x)$ denotes the fraction of $x$ in $\vx$. 
An adversary is allowed to change symbols in the transmitted codeword only from small values to large values, the cost he pays by changing every $i$ to $j$ ($0\le i<j\le 5$) is $j-i$ dollars, and he has a budget of $2.3n$
 dollars in total. The fundamental type of questions we are able to answer in this paper is: is it possible for us to design exponentially large codes so that no matter which codeword is transmitted and how a legitimate adversary corrupts  it, the decoder is always able to output a list of at most 10 codewords which contains  the correct one?

The answer can be stated in a similar manner. This is possible if and only if there is a $\cp$ tensor of order 11 and dimension 6 which does \emph{not} lie inside the \emph{confusability set} determined by the channel. In particular, the confusability set is the set of joint distributions which \emph{fail} to meet the conditions similar to \textbf{C1} or \textbf{C2} that are determined by the channel. 

Our results tell us that if one only aim to search for exponentially large $(L-1)$-list decodable codes (instead of optimizing its size) for a given general adversarial channel, it is \emph{sufficient} (and obviously necessary) to restrict our attention to codes that are \emph{chunk-wise random-like}. Such codes correspond to some $\cp$ distribution $\sum_{i=1}^k\lambda_iP_{\bfx_i}$. If a random code of positive rate in which the  $\lambda_in$ ($1\le i\le k$) components in the $i$-th chunk of each codeword is sampled from distribution $P_{\bfx_i}$ does not work with high probability (w.h.p.), then we can never find positive rate codes of any other form that work for this channel.

By setting the \emph{list size} $L-1=1$, results in \cite{wang-budkuley-bogdanov-jaggi-2019-omniscient-avc} are recovered by our work.

\section{Introduction}
\label{sec:intro}
While the main contribution of this work is to strictly generalize notions that have been primarily studied for ``Hamming metric'' channels, before we precisely define general channels, let us reprise what is known for Hamming metric channels in this section.

\subsection{Error correction codes and Plotkin bound}
The theory of error correction codes is about protecting data from errors.  In classical coding theory, a code, say $\cC$, is just a collection of binary \emph{codewords} (which are usually just binary length-$n$ sequences, where $n$ is called the \emph{blocklength}). The most well-studied error model is \emph{bit-flip}. When a certain codeword is transmitted, an adversary can arbitrarily flip at most $np$ ($0<p<1/2$) bits. It is easy to see that two codewords are not \emph{confusable} if and only if their Hamming distance (number of locations where they differ, denoted $\disth{\cdot}{\cdot}$) is at least $2np+1$. Let
\[\dmin(\cC)=\min_{\vx\ne\vx'\in\cC}\disth{\vx}{\vx'}\] 
denote the minimum pairwise distance of codewords in $\cC$. The goal is to \emph{pack} as many codewords as possible  in Hamming space $\bF_2^n$ while ensuring that the minimum distance is at least $2np+1$. By a simple volume argument (Gilbert--Varshamov (GV) bound \cite{gilbert-1952-gv, varshamov-1957-gv}), it is known that \emph{exponentially many} such vectors  can be packed when $p<1/4$. The fundamental quantity that coding theorists are seeking when faced with any communication model is the largest \emph{achievable} rate, i.e., \emph{capacity}. The \emph{rate} of a code is its normalized cardinality, $R(\cC) = \frac{\log\card{\cC}}{n}$. The capacity $C$ measures asymptotically, as the blocklength grows, the largest fraction of bits (out of $n$) that can be {reliably} transmitted despite $np$ adversarial bit-flips. $C$ is formally defined as 
\[C\coloneqq\limsup_{n\to\infty}\max_{\cC\subset\bF_2^n\colon \dmin(\cC)>2np}R(\cC).\footnote{Allowing vanishing probability of decoding error does not change the problem.}\]
For the aforementioned bit-flip model, as said, the problem of finding the capacity can be also cast as determining the \emph{sphere packing density}. It is notoriously difficult and is still open to date. However, we do know that $p=1/4$ is the threshold below which exponential packing exists (as suggested by the Gilbert--Varshamov (GV) bound) and above which it is impossible. The latter fact is the famous Plotkin bound. Formally,
\begin{theorem}[Plotkin bound \cite{plotkin-1960-plotkinbound}]
\label{lem:plotkin}
If $p=1/4+\eps$, then any code $\cC$ of distance larger than $2np$ has cardinality at most $1+\frac{1}{4\eps}$ (and hence zero rate).
\end{theorem}
We will  call the value of $p$ at which the capacity hits zero the \emph{Plotkin point}.
Note that the Plotkin bound actually tells us that, above the Plotkin point, any code/packing not only has size $2^{o(n)}$ (hence rate zero), but should be at most a \emph{constant} (independent of the blocklength $n$). Coupled with the achievability result given by the GV bound, the phase transition threshold for exponential-sized packing is thereby identified precisely. 

\subsection{List decoding and list decoding Plotkin bound}
We now introduce another important notion: \emph{list decoding}. List decodability still requires codewords to be separated out, but in a more relaxed sense. It  requires that only a few codewords can be captured by a ball of some radius, no matter where it is put. 
\begin{definition}[List decodability \cite{elias1957ld, wozencraft1958ld}]
\label{def:ld}
A code $\cC$ is $(p,L-1)$-list decodable (or $(p,<L)$-list decodable) if for all $\vy\in\bF_2^n$, $\card{\cC\cap\bham\paren{\vy,np}}<L$, where $\bham(\vy,np)$ denotes a Hamming ball centered at $\vy$ of radius $np$. 
\end{definition}
Of course we want the \emph{list size} $L$ to be as small as possible. In particular, the problem is trivial when $L=\card{\cC}$. (The decoder ignores the channel output and outputs the full code.)
When $L=2$, it becomes precisely packing. As the admissible $L$ grows, the problem is expected to become easier. 

List decoding is an important and well-studied subject in coding theory. It is a natural mathematical question to pose for understanding high-dimensional geometry in discrete spaces. It also serves as a useful primitive that shows power within and beyond the scope of coding theory. For instance, in many communication problems (e.g., \cite{ahlswede1973avcfeedback,chen-et-al}), a proof technique is to let the decoder first perform list decoding and get a short list (usually $\poly(n)$ suffices) of candidate messages, then use other information to disambiguate the list and get the truely transmitted message. List decoding also finds application in complexity theory, cryptography, etc \cite{guruswami2006list-dec-cplx-pr}. For instance, it is  used for amplifying hardness and constructing extractors, pseudorandom generators and other pseudorandom objects \cite{dmoz-nearly-opt-pr-from-hard}. 
The idea of relaxing the problem by asking the solver to just output a list (ideally as small as possible) of solutions that is guaranteed to contain the correct one, instead of insisting on  a unique answer, is also adopted in many other fields of computer science \cite{diakonikolas2018list, raghavendra2019list, karmalkar2019list}. In the context of high-dimensional geometry in finite fields, list decoding is equivalent to multiple packing just like error correction codes are equivalent to sphere packing. Multiple packing is a natural generalization of the famous sphere packing problem in which, instead of insisting on disjoint balls, overlap is allowed but with bounded multiplicity. 
\begin{definition}[Multiple packing]
A subset $\cC\subset\bF_2^n$ is a $(p,L-1)$-multiple packing if when we put balls of radii $np$ around each vector in $\cC$,  no point in the space simultaneously lies in  the intersection of at least $L$ balls. 
\end{definition}

See Fig. \ref{fig:packing_vs_multiple_packing} for examples of packing and multiple packing in Hamming space. 
\begin{figure}
    \centering
    \begin{subfigure}{0.45\textwidth}
    	\centering
    	\includegraphics[scale = 2]{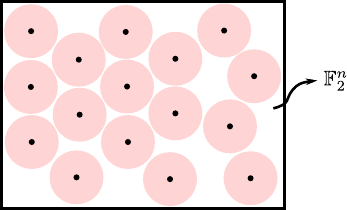}
    	\caption{An $(L-1)$-packing for $L=2$, i.e., disjoint packing.}
    	\label{fig:packing}
    \end{subfigure}\quad 
    \begin{subfigure}{0.45\textwidth}
    	\centering
    	\includegraphics[scale = 2]{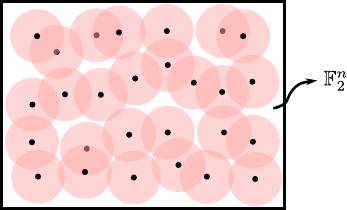}
    	\caption{An $(L-1)$-packing for $L=3$, i.e., packing with multiplicity 2.}
    	\label{fig:3_packing}
    \end{subfigure}
    \caption{Packing (uniquely decodable codes) vs. multiple packing (list decodable codes). The geometry depicted in the above figures may be misleading compared with the truth in binary Hamming space.}
    \label{fig:packing_vs_multiple_packing}
\end{figure}

Surprisingly, list decoding capacity is known if we allow $L$ to be  asymptotically large. In some sense, list decoding makes us information-theoretic since in many (but not all) cases the list decoding capacity coincides with the corresponding Shannon channel capacity for which the noise is random with the same ``power" (e.g., in the bit-flip/erasure case, the random noise is independently and identically distributed (i.i.d.) according to a Bernoulli distribution per component with mean $p$). 
\begin{theorem}[List decoding capacity (folklore)]\label{thm:ld_cap_bitflip}
Given any $\delta>0$, there exists an infinite sequence of $\paren{p,\cO(1/\delta)}$-list decodable codes $\cC$ of rate $1-H(p)-\delta$. Indeed, a random code (each codeword sampled uniformly at random from $\bF_2^n$) of rate $1-H(p)-\delta$ is $\paren{p,\cO(1/\delta)}$-list decodable w.h.p.

On the other hand, any infinite sequence of codes of rate $1-H(p)+\delta$ is $\paren{p,2^{\Omega(n\delta)}}$-list decodable. 

We call $1-H(p)$ the \emph{$p$-list decoding capacity} (without specifying a specific $L$). In particular, the Plotkin point for $p$-list decoding when $L$ is sufficiently large is $1/2$.
\end{theorem}

Though the fundamental limit for the relaxed problem for large constant $L$ is essentially understood, $(p,L-1)$-list decodability for small $L$ (e.g., absolute constant, say $3,8,100$, etc.; or sublinear in $1/\delta$, say $(1/\delta)^{1/2}$, $(1/\delta)^{1/3}\log(1/\delta)$, $\log\log(1/\delta)$) is way far from being understood. Indeed, it is believed (at least for absolute constant $L$) to be equivalently hard as the sphere packing problem. Formally, the question of understanding the  role of $L$ can be cast as follows. Note first that when $L=2$, the (unknown) capacity lies somewhere between the Gilbert--Varshamov bound and Linear Programming bound (\cite{delsarte-1973-lp, macwilliams-1963-identity, mrrw1, mrrw2, navon-samorodnitsky-2009-lp-fourier}). When $L=\cO(1/\delta)$, the list decoding capacity $1-H(p)$ is much larger than the unique decoding capacity. As we increase $L$, the $(p,L-1)$-list decoding capacity should be gradually lifted and the Plotkin point should somehow move rightwards from $1/4$ to $1/2$. The final goal is to completely understand the dynamics of this evolution. 

\begin{remark}
In this paper, we explicitly distinguish the list decoding capacity for large $L$ and for small $L$. When we say that $L$ is asymptotically large, we refer to $L=\Omega(1/\delta)$ which suffices to approach the $p$-list decoding capacity within gap $\delta$. When we say that $L$ is small without further specification, we refer to absolute constant $L$. The $p$-list decoding capacity for large $L$ is fully characterized as in Theorem \ref{eqn:list_dec_cap}, denoted $C$, yet the $(p,L-1)$-list decoding capacity for small $L$ is widely open and is denoted by $C_{L-1}$.
\end{remark}

Again, for any absolute constant $L$, the $(p,L-1)$-list decoding capacity is poorly understood. We only have non-matching lower and upper bounds. To the best of our knowledge, the current record holder is still the ones by Blinovsky from the 80s \cite{blinovsky-1986-ls-lb-binary, blinovsky-2005-ls-lb-qary, blinovsky-2008-ls-lb-qary-supplementary}, except for sporadic values of $L$ in some  regimes of $p$. Specifically, for $L=3$, Ashikhmin--Barg--Litsyn \cite{ashikhmin-barg-litsyn-2000-list-size-2} can uniformly improve Blinovsky's upper bound for all values of $p$. For \emph{even} $L$'s that are at least $4$, Polyanskiy \cite{polyanskiy-2016-ld-ub} can partially beat Blinovsky's bounds in the low rate regime. 

Though the speed of convergence in $L$ is not exactly known, Blinovsky's bounds \emph{do} resolve the dynamics of Plotkin point evolution! Let $P_{L-1}$ denote the Plotkin point for $(p,L-1)$-list decoding. Let $L=2k$ or $2k+1$ ($k\ge1$). Then Blinovsky's results imply that $P_{L-1}$ is precisely given  by the following formula
\[P_{L-1} = \sum_{i = 1}^k\frac{\binom{2(i - 1)}{i - 1}}{i} 2^{-2i}.\]
Later, Alon--Bukh--Polyanskiy \cite{alon-bukh-polyanskiy-2018-ld-zero-rate} recover this result with a simpler looking formula
\[P_{L-1} = \frac{1}{2} - 2^{-2k - 1}\binom{2k}{k},\]
For instance, $P_{1} = P_2 = 1/4$, $P_3 = P_4 = 5/16$, etc. As can be noted, the Plotkin point moves \emph{periodically}! The fact that the above two formulas are always evaluated to the same value is implicit in \cite{alon-bukh-polyanskiy-2018-ld-zero-rate} and formally justified in Appendix \ref{app:blinovsky_abp}.

\section{Our contributions}
\label{sec:contributions}
Our motivation comes from a well-known connection between list decodability and reliability of  communication over adversarial channels. A binary code is $(p,L-1)$-list decodable if and only if it has zero error  when used over the following \emph{adversarial bit-flip channel} (Fig. \ref{fig:channel_bitflip}).
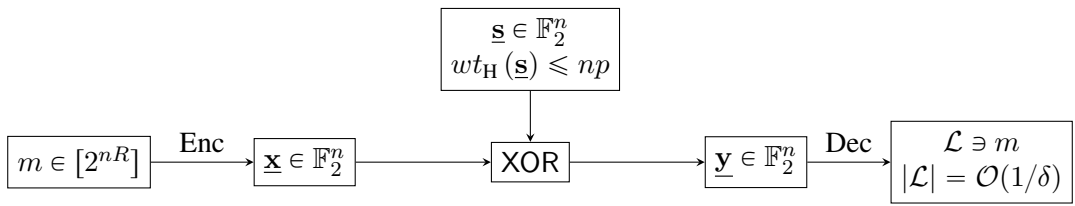
\begin{figure}[h]
	\centering
	\begin{tikzpicture}
	\node (m)[draw, align=center] {$m\in[2^{nR}]$};
	\node (x)[draw, align=center, right of=m, xshift = 2cm] {$\vbfx\in\bF_2^n$};
	\draw [->,>=stealth] (m) -- node[above]{Enc}   (x);
	
	\node (+)[draw, align=center, right of=x, xshift = 2cm]  {$\XOR$};
	\draw [->,>=stealth] (x) -- (+);
	
	\node (y)[draw, align=center, right of=+, xshift = 2cm] {$\vbfy\in\bF_2^n$};
	\draw [->,>=stealth] (+) -- (y);
	\node (s)[draw, align=center, above of=+, yshift = 0.5cm] {$\vbfs\in\bF_2^n$\\$\wth{\vbfs}\le np$};
	\draw [->,>=stealth] (s) -- (+);
	
	\node (M)[draw, align=center, right of=y, xshift = 2cm] {$\cL\ni m$\\ $|\cL|=\cO(1/\delta)$ };
	\draw [->,>=stealth] (y) -- node[above]{Dec}   (M);
	\end{tikzpicture}
	\caption{Adversarial bit-flip channels.}
	\label{fig:channel_bitflip}
\end{figure}

The above system depicts a one-way point-to-point communication in which the encoder (Alice) randomly picks a message $m$ from $2^{nR}$ of them and encodes it into a $n$-bit string, the adversary (James) stares at this codeword and maliciously flips at most $np$ bits of it, the decoder (Bob) receives the corrupted word $\vy$ and is required to output a short list of messages which is guaranteed to contain $m$ with probability 1. 

In the above model, the adversary is power constrained in the sense that he only has a budget of $np$ bit-flips. But the encoder is not constrained -- she can encode the message into any vector in $\bF_2^n$. In some scenarios, codewords are also weight constrained. It makes sense to pose the same question (understanding the list decoding capacity) for input constrained channels. Indeed, this was also studied \cite{guruswami2013combinatorial} and the list decoding capacity is $H(w) - H(p)$ when each codeword has weight at most $nw$. Note that it vanishes at $p = w$. That is, the Plotkin point for weight constrained adversarial bit-flip channels is $w$.

Motivated by this connection, we significantly generalize the bit-flip model and define list decodability for \emph{general adversarial channels}. We consider a large family of channels in which the encoder  is allowed to encode the message  into a length-$n$ sequence $\vx$ over \emph{any}  alphabet $\cX$ of constant size, the adversary is allowed to design an adversarial noise pattern $\vs$ over \emph{any} alphabet $\cS$ and the channel can be any \emph{deterministic component-wise} function taking a pair of strings from $\cX^n\times\cS^n$, outputting a sequence $\vy$ over \emph{any} alphabet $\cY$ of the same length. The system designer can incorporate a large family of constraints on $\vx$ and $\vs$ in terms of their \emph{types} (i.e., empirical distributions). 
The above family of adversarial channels includes but is not limited to
\begin{enumerate}
    \item The standard adversarial bit-flip channels and adversarial erasure channels;
    \item $Z$-channels in which the adversary can only flip 1 to 0 but not the other way around;
    \item Adder channels in which the output is the sum of inputs over the reals rather than modulo the input alphabet size;
    \item Channels equipped with Lee distance instead of Hamming metric.
\end{enumerate}
Indeed, our framework covers most popular error models and more that potentially have not been studied in the literature. 

However, since we require the channel transition function to act on each component of the input codeword independently,  a well-studied family of channels is excluded: the \emph{adversarial deletion channels}. In this model, the adversary can \emph{delete} at most $np$ entries of the transmitted codeword and the decoder receives a vector of smaller length (but at least $(1-p)n$) without knowing the original locations of the symbols he got.\footnote{We want to emphasize the difference between deletions and erasures. When symbols in the codeword are deleted, the rest of the symbols are concatenated and the receiver has no idea which symbols were deleted. When symbols are erased, they are replaced by  erasure symbols $\erasure$ at the same locations and the receiver seeing them knows exactly which symbols were erased. Hence the erasure case is much simpler than the deletion case.} Determining the Plotkin point for this channel is a long standing open problem. It is known \cite{bukh-guruswami-hastad-2016-adv-deletion-plotkin} that for binary channels, it  lies between $\sqrt{2} - 1\approx0.414$ and $0.5$; for $q$-ary channels, between $1-\frac{2}{k+\sqrt{k}}$ and $1-\frac{1}{k}$. The capacity of this channel is even less known.

For technical simplicity, we also assume that the channel transition function is  \emph{deterministic}, i.e., the output symbol is a deterministic function of the codeword symbol $x$ and the error symbol $s$.\footnote{The general case in which the channel law is given by a conditional distribution $W_{\bfy|\bfx,\bfs}$ (with not necessarily only singleton atoms) is more technical and is left as one of our future directions. }

However,  without loss of generality one can assume that  none of the encoder, decoder and adversary has private randomness to randomize their strategy. This is because that there are reductions showing that, given randomized encoder/decoder, we can construct a deterministic coding scheme with essentially the same rate. Similarly, given a randomized adversarial error function, we can turn it into a deterministic one which is equivalently malicious in terms of rate. Therefore, for the encoder, it suffices to only consider deterministic codes, i.e., each message is mapped to a unique codeword with probability 1. For the adversary, we can assume the error pattern is a deterministic  function  of the transmitted codeword. {Note that the error function does \emph{not} have to be component-wise independent. The $i$-th component $\vs(i)$ of the noise pattern $\vs$ can depend on \emph{every} entry of $\vx$, not only on the corresponding $\vx(i)$.} Moreover, the decoder's decision of the estimate  message given the received word can also be assumed to be  deterministic. That is, we can require that the decoder  outputs the correct message with \emph{zero} error probability. Hence, the problem is purely combinatorial and all desirable events should happen with probability one.

In this work, we precisely \emph{characterize} the \emph{Plotkin point} for list decoding over any channel from the above large family of general adversarial channels. That is, we provide a criterion (sufficient and necessary condition) under which positive $(L-1)$-list decoding rate is possible for such channels.

In the context of high-dimensional geometry over finite spaces, the result can be also cast as pinning down the location of phase transition threshold for $(L-1)$-multiple packing using general shapes (not necessarily Hamming balls) corresponding to the defining constraints for codewords and errors of the channel, above which exponential-sized multiple packing exists and below which impossible. 

This criterion can be summarized in one sentence:
\begin{quote}
    exponential-sized $(L-1)$-list decodable codes for \emph{general adversarial channels} (or $(L-1)$-multiple packings using \emph{general shapes}) exist if and only if the \emph{completely positive tensor cone} of order-$L$ is not entirely contained in the $(L-1)$-list decoding \emph{confusability set} of the channel.
\end{quote}

Jargon in the above informal statement will become understandable once we formalize the problem setup and present rigorous claims. 
The proof consists of sufficiency part and necessity part. At a very high level, the sufficiency part follows from a random coding argument and its generalization inspired by time-sharing argument frequently used in Network Information Theory. The necessity part builds upon and significantly generalizes the classical Plotkin bound, which goes by first extracting an \emph{equicoupled} subcode using Ramsey theory and then applying a \emph{double counting} trick. 

Other results include the following.
\begin{enumerate}
	\item We pin down the list decoding capacity of any given general adversarial channel for asymptotically large  $L$.  This generalizes the classic list decoding capacity in the bit-flip case. The lower bound is achieved by a purely random code. The upper bound follows from volume packing. 
	\item We determine the \emph{exact} order (in terms of $\delta$) of the list sizes for a large fraction (exponentially close to one) of constant composition codes (all codewords have the same type) achieving the list decoding capacity of a given general adversarial channel within gap $\delta$. It turns out that if we pick a constant composition code from the  set of all such codes, with high probability, it is exactly $\Theta(1/\delta)$-list decodable.
	\item We give a lower bound on the $(L-1)$-list decoding capacity of a given general adversarial channel. It coincides with the generalized Gilbert--Varshamov bound obtained by \cite{wang-budkuley-bogdanov-jaggi-2019-omniscient-avc} when $L-1$ is set to be $1$. Our bound is given by a random code construction assisted by expurgation, generalizing a classic construction for $(p,L-1)$-list decoding in the bit-flip case \cite{guruswami-lncs}. Note that this construction differs from \cite{wang-budkuley-bogdanov-jaggi-2019-omniscient-avc}'s construction for unique decoding using greedy packing.
	\item  In the special case where $L=2$, i.e., the unique decoding setting, we evaluate the Gilbert--Varshamov-type bound and an achievable rate expression of cloud codes (codes constructed from $\cp$ distributions) obtained by \cite{wang-budkuley-bogdanov-jaggi-2019-omniscient-avc} under the bit-flip model. In particular, we show that the Gilbert--Varshamov-type bound for general adversarial channels matches the classic GV bound in the theory of error correction codes. We also provide an explicit convex program for evaluating achievable rates of  codes arising from $\cp$ distributions. 
	\item By evaluating our general criterion under the bit-flip model, we numerically recover Blinovsky's \cite{blinovsky-1986-ls-lb-binary} characterization of the Plotkin point for $(p,L-1)$-list decoding. This boils down to checking the feasibility of an explicit a linear program with structured coefficient matrix. Though the LP has size exponential in $L$, its feasibility can be  checked in constant time since our results are tailored for constant $L$ with no dependence on the blocklength $n$ (which typically approaches infinity for many of our results to hold).
	\item By utilizing facts  discovered in this paper, we rigorously recover Blinovsky's \cite{blinovsky-1986-ls-lb-binary} characterization of the Plotkin point for $(p,L-1)$-list decoding. Our proof avoids the harder calculations and demystify the formula by Blinovsky\footnote{In fact, he provided upper and lower bounds for $(p,L-1)$-list decoding capacity which happen to vanish at the same value of $p$.}. In particular, our lower bound on the Plotkin point explains, in the low rate regime, the fact that \emph{average-radius}\footnote{$(p,L-1)$-average-radius list decodability requires that the \emph{average}  distance (instead of maximum distance required by the classic notion of $(p,L-1)$-list decodability) between any $L$-tuple of codewords and their centroid is larger than $np$. Average-radius list decodability is a more stringent requirement since it implies classic list-decodability. However, it is easier to analyze since the problem is \emph{linearized}. Indeed it shows power in a long line of work understanding  the bit-flip model \cite{guruswami2013combinatorial,wootters-2013-ld-rand-lin-low-rate,rudra-wootters-2014-puncturing,rudra-wootters-2015-rand-op,rudra-wootters-2018-avg-rad-lr-rand-lin}.} list decoding is equivalent to the classic notion of list decoding. We believe that this fact is first observed and rigorously justified by Blinovsky. It was later rediscovered many times and became the basic starting point of many papers, especially those regarding list decoding random $q$-ary linear codes. Our upper bound  relates the Plotkin point $P_{L-1}$ to the expected translation distance of a one-dimensional unbiased random walk after $L$ steps. In summary, using connections between codes and random variables, we are able to   re-interpret of the formulas given by Blinvosky \cite{wang-budkuley-bogdanov-jaggi-2019-omniscient-avc} and Alon--Bukh--Polyanskiy \cite{alon-bukh-polyanskiy-2018-ld-zero-rate} and provide a new intuitive formula which matches known formulas.
\end{enumerate}

\section{Overview of techniques}
\label{sec:techniques}
Our paper is highly correlated to a sister paper \cite{wang-budkuley-bogdanov-jaggi-2019-omniscient-avc} which a subset of the authors are involved in. That paper provides generalized Plotkin bound for \emph{unique} decoding over general adversarial channels. The authors showed that exponential-sized \emph{uniquely} decodable codes or hard packings exist if and only if the set of completely positive \emph{matrices} is not entirely contained in the \emph{confusability set} associated to the given channel. This answers the question we posed in the beginning of the paper for $L=2$ case. We generalize their results to \emph{any universal constant} $L$. Almost all results in \cite{wang-budkuley-bogdanov-jaggi-2019-omniscient-avc} can be recovered by setting $L=2$ in our paper.

We review the techniques used in this paper and  highlight the similarities and differences between  \cite{wang-budkuley-bogdanov-jaggi-2019-omniscient-avc}\footnote{Though the work by Wang--Budkuley--Bogdanov--Jaggi \cite{wang-budkuley-bogdanov-jaggi-2019-omniscient-avc} has been accepted to ISIT 2019, the conference version is limited to 5 pages and contains essentially no proof. At the time this paper is written,  we do not have a publicly available \emph{full} version of    \cite{wang-budkuley-bogdanov-jaggi-2019-omniscient-avc} and the following comparison is w.r.t. the current status of a draft of \cite{wang-budkuley-bogdanov-jaggi-2019-omniscient-avc} that the authors kindly shared with us. } and our work. 
\begin{enumerate}
    \item The general adversarial channel models that both papers are concerned with belong to a larger family of channels known as \emph{Arbitrarily Varying Channels (AVC)} in Information Theory community. We want to emphasize that a bulk of the literature of AVCs deals with \emph{oblivious} channels in which the adversary has to pick his noise pattern maliciously \emph{before} the codewords is chosen from the codebook by the encoder. This makes the problem significantly easier and the capacity of such channels are precisely known. The channels that \cite{wang-budkuley-bogdanov-jaggi-2019-omniscient-avc} and we are considering are such that the adversary gets to design the error pattern with the knowledge of the transmitted codeword. This problem is way more difficult and the capacity is, again, widely open even for simple models such as the bit-flip channels. Indeed, the subclass of AVCs that \cite{wang-budkuley-bogdanov-jaggi-2019-omniscient-avc} and we defined is motivated by the bit-flip channels and its various variants, e.g., weight constrained channels, $q$-ary channels, etc.
    \item The connection between codes and random variables or distributions are classical in Theoretical Computer Science. The idea of realizing binary error correction codes using $\curbrkt{-1,1}$-valued random variables or functions supported on the Boolean hypercube $\curbrkt{-1,1}^n$ is spread out in the literature explicitly or in  disguise. Such tricks show power since it allows people to borrow tools from other fields of Theoretical Computer Science, e.g., the theory of expander graphs, randomness extractors, small-bias distributions, discrete Fourier analysis, etc. (\cite{sipser1996expandercodes, ben-aroya-doron-ta-shma-2018-explicit-erasure-ld,ta-shma-2017-explicit-gv, bhowmick-lovett-2014-list-dec-rm})  to understand, construct and analyze codes. 
    \item With respect to (w.r.t.) codes for  general adversarial channels, the specific idea of collecting admissible types of good codes and studying the set of corresponding distributions was used in \cite{wang-budkuley-bogdanov-jaggi-2019-omniscient-avc}.  In particular, they defined similar notions of self-couplings and confusability sets which are submanifolds of \emph{matrices}. Such objects only take care of \emph{pairwise} interaction of codewords, which are insufficient for understanding list decoding. We generalize their notions to  \emph{tensors} which captures the (empirical) joint distributions of \emph{lists} of codewords. Some properties in \cite{wang-budkuley-bogdanov-jaggi-2019-omniscient-avc} continue to hold when objects in matrix versions are extended to tensor versions. Other properties fail to hold, as we will see in the rest of the paper. We also encounter issues which merely do not exist in the unique decoding setting. As is well-known, tensors are much more delicate \cite{hillar2013tensor-np-hard} to handle than matrices. 
    \item To prove \emph{upper} bounds on capacity, it is also an old idea to extract structured subcodes from \emph{any} infinite sequence of good codes. Depending on the applications, the types of \emph{structures} and techniques for extracting such structures may vary. To the best of our knowledge, in coding theory, the use of Ramsey theory for obtaining symmetric subcodes dates back to as least as early as Blinovsky \cite{blinovsky-1986-ls-lb-binary}. His techniques are applied in a similar manner in followup work by Polyanskiy \cite{polyanskiy-2016-ld-ub} and Alon--Bukh--Polyanskiy \cite{alon-bukh-polyanskiy-2018-ld-zero-rate}. \cite{wang-budkuley-bogdanov-jaggi-2019-omniscient-avc} generalizes this idea and manages to extract subcodes from arbitrary codes  for \emph{general} adversarial channels. Since they work with unique decoding, \emph{pairwise} equicoupledness suffices. In our setup, we would like a sequence of subcodes which are \emph{$L$-wise equicoupled} in the sense that the (empirical) joint distribution  of \emph{any} $L$-tuple of codewords from the extracted subcode is approximately the same and close to some $\wh P_{\bfx_1,\cdots,\bfx_L}$. This resembles but generalizes Polyanskiy's \cite{polyanskiy-2016-ld-ub} techniques. One of the downsides of invoking Ramsey theory is that the reduction usually causes terrible detriment to the rate of the code, since the smallest size for a combinatorial object to contain abundant structures is generally poorly understood in combinatorics. However, we are fine to tolerate such a rate loss since we only care about the \emph{positivity} of list decoding capacity.
    \item To show lower bounds on capacity, we use random coding  argument aided by \emph{expurgation}. In the prior work \cite{wang-budkuley-bogdanov-jaggi-2019-omniscient-avc}, the achievability result is obtained by greedy packing. This is reminiscent of a classical technique in Coding Theory for proving  existence of good codes of certain size. Since, in the unique decoding (hard packing) setting, goodness of a code  relies merely on pairwise statistics, the size of a greedy packing can be lower bounded using a standard volume counting argument. Indeed, this idea can be implemented in the general setting by counting the volume of the ``forbidden region" of any codeword \cite{wang-budkuley-bogdanov-jaggi-2019-omniscient-avc}. However, in list decoding setting, the notion of \emph{confusability} is defined for \emph{tuples} of codewords and does \emph{not} translate to non-intersection of  forbidden regions of  codewords. It is also not clear how to pack codewords in a greedy manner while ensuring non-existence of local dense clusters. Instead, our code construction is more information-theoretic. We apply ideas of random coding with expurgation which is commonly used in the study of error exponent in Information Theory. A random code may be mildly locally clustered, but this only occurs at rare locations in the space of all length-$n$ sequences over the input alphabet. Indeed, we are able to show that, with high probability, a random code carefully massaged by shoveling off a small number of codewords  attains a  GV-type bound for  general channels. 
    \item The most difficult part of our work is the converse. 
    \begin{enumerate}
        \item First assume that the distribution $\wh P_{\bfx_1,\cdots,\bfx_L}$ associated to the subcode obtained by Ramsey reduction is \emph{symmetric}. To show that no large code exists for general adversarial channels when $\wh P_{\bfx_1,\cdots,\bfx_L}$ is not completely positive, we show contradicting upper and lower bounds, if the code size exceeds certain constant (not even depending on the codeword length!), on the empirical distribution  taken inner product with a copositive witness of non-complete positivity of $\wh P_{\bfx_1,\cdots,\bfx_L}$ and averaged over all $L$-tuples in the symmetric equicoupled subcode. We review this \emph{double counting trick} (for unique and list decoding under special settings that appeared in prior work) in Section \ref{sec:prior_work}. The $L=2$ case is proved in \cite{wang-budkuley-bogdanov-jaggi-2019-omniscient-avc}. The existence of witness of non-complete positivity is guaranteed by duality of certain matrix cones. We generalize calculations in \cite{wang-budkuley-bogdanov-jaggi-2019-omniscient-avc} to joint distributions of $>2$ random variables. Similar notions of complete positivity and CoPositivity for tensors exist in the literature and duality continues to hold. 
        \item If $\wh P_{\bfx_1,\cdots,\bfx_L}$ is asymmetric, we use a completely different argument. We reduce the problem, in a nontrivial way, to the $L=2$ case which is known to be true \cite{wang-budkuley-bogdanov-jaggi-2019-omniscient-avc}. The $L=2$ case itself is proved \cite{wang-budkuley-bogdanov-jaggi-2019-omniscient-avc} by viewing the task of constructing a long sequence of random variables with prescribed asymmetric marginals as a zero sum game and using discrete Fourier analysis to provide  conflicting bounds on the value of the game, if the sequence is longer than certain constant (again independent of the blocklength). 
    \end{enumerate}
\end{enumerate}

\section{Prior work}
\label{sec:prior_work}
Among various ideas, our results are built upon prior work which applies a \emph{double counting trick} to obtain upper bounds on code sizes. We first review this technique which can be found in the proof of classical Plotkin bound and its generalizations. 

\subsection{Plotkin \cite{plotkin-1960-plotkinbound}.}
One way to prove Theorem \ref{lem:plotkin} is by lower and upper bounding the expected pairwise distance of any given code $\cC$ with minimum distance larger than $2np$ ($p=1/4+\eps$) 
\begin{equation}
    \exptover{(\vx,\vx')\sim\cC\times\cC}{\disth{\vx}{\vx'}},
    \label{eqn:pairwise_dist_plotkin}
\end{equation}
where $\vx,\vx'$ are uniformly and independently picked from $\cC$.
First note that pairs $\vx = \vx'$ do not contribute to the expectation.
On the one hand, the expectation is clearly at least 
\[\cardC^{-2}\cardC(\cardC - 1)\dmin\ge\cardC^{-1}(\cardC - 1)2np\ge\cardC^{-1}(\cardC - 1)2n(1/4+\eps).\]
On the other hand, if we stack codewords into a $2^{nR}\times n$ matrix and let $S_j$ denote the number of 1's in the $j$-th column, then from the column's perspective, the above expectation is at most 
\[\frac{1}{\cardC^2}\sum_{j = 1}^n2S_i(\cardC - S_i).\]
The coefficient 2 is because we need to count $(\vx,\vx')$ and $(\vx',\vx)$ separately. This bound is at most $n/2$ by concavity of the summands. 
Comparing the upper and lower bounds we have that $\cardC\le1+\frac{1}{4\eps}$, as claimed in Theorem \ref{lem:plotkin}.

\subsection{Blinovsky \cite{blinovsky-1986-ls-lb-binary}.}
The above double counting argument can be generalized to the setting of list decoding. For the $(p,L-1)$-list decoding setup we introduced in Definition \ref{def:ld},  the earliest work we are aware of following this idea is the one by Blinovsky \cite{blinovsky-1986-ls-lb-binary}. 

Unlike Theorem \ref{lem:plotkin}, Blinovsky did not only show that any $(p,L-1)$-list decodable code has to be small as long as $p>P_{L-1}$. He actually gave an upper bound (and is still essentially the best as far as we know) on $(p,L-1)$-list decoding capacity for \emph{any} $L$. We sketch his idea below but omit the complicated calculations. 

First note that proving upper bounds on $C_{L-1}$ for fixed $p$ is equivalent to proving upper bounds on $p$ for fixed rate $R$.  We define the following three quantities
\begin{align}
    \rld = &\min_{\cL\in\binom{\cC}{L}} \min_{\vy\in\bF_2^n}\max_{\vx\in\cL}{\disth{\vy}{\vx}},\label{eqn:rld}\\
    \ravg = &\min_{\cL\in\binom{\cC}{L}} \min_{\vy\in\bF_2^n}\mathop{\bE}_{\vx\sim\cL}\sqrbrkt{\disth{\vy}{\vx}},\label{eqn:ravg}\\
    \rdc=&\mathop{\bE}_{\cL\sim\binom{\cC}{L}}\min_{\vy\in\bF_2^n}\mathop{\bE}_{\vx\sim\cL}\sqrbrkt{\disth{\vy}{\vx}}.\label{eqn:rdc}
\end{align}
All expectations are over uniform selection from corresponding sets. Namely,
\[\exptover{\cL\sim\binom{\cC}{L}}{\cdot} = \frac{1}{\binom{\cardC}{L}}\sum_{\cL\in\binom{\cC}{L}}\sqrbrkt{\cdot},\quad\exptover{\vx\sim\cL}{\cdot} = \frac{1}{L}\sum_{\vx\in\cL}\sqrbrkt{\cdot}.\]
Let us parse what these quantities are measuring.
\begin{enumerate}
    \item $\rld$ is known as the \emph{list decoding radius} of a given code $\cC$. The minimax expression associated to a set $\cL$ of vectors
    \[\rcheb\coloneqq\min_{\vy\in\bF_2^n}\max_{\vx\in\cL}{\disth{\vy}{\vx}}\]
    is known as the \emph{Chebyshev radius} of $\cL$. It is the radius of the smallest circumscribed ball of $\cL$. And
    \[p^*(R)\coloneqq\limsup_{n\to\infty}\max_{\cC\subset\bF_2^n\colon\cardC\ge2^{nR}}\rld(\cC)\]
    is precisely the largest allowable $p$  for $(p,L-1)$-list decodable code of a fixed rate $R$. 
    \item $\ravg$ is known as the \emph{average list decoding radius} and the min-average expression
    \[\min_{\vy\in\bF_2^n}\mathop{\bE}_{\vx\sim\cL}\sqrbrkt{\disth{\vy}{\vx}}\]
    is the \emph{average radius} of a list. It is not hard to see that the average radius center of $\cL$ is the component-wise majority of vectors in $\cL$, i.e., the minimizer $\vy^*$ has $\maj\paren{\vx(i)\colon \vx\in\cL}$ as its $i$-th component. Define \emph{plurality} as
    \begin{align*}
        \begin{array}{rlll}
            \plur\colon & \bF_2^L & \to & [0,1] \\
             & (x_1,\cdots,x_L) & \mapsto & \frac{1}{L}\card{\curbrkt{i\in[L]\colon x_i = \maj(x_1,\cdots,x_L)}},
        \end{array}
    \end{align*}
    which is the fraction of the most frequent symbol. Then the average radius of $\cL$ can be explicitly written as
    \begin{align*}
        \min_{\vy\in\bF_2^n}\mathop{\bE}_{\vx\sim\cL}\sqrbrkt{\disth{\vy}{\vx}} = & \sum_{j = 1}^n\paren{1-\plur\paren{\vx(i)\colon \vx\in\cL}}.
    \end{align*}
    \item $\rdc$ is a further variant of $\rld$ -- the ultimate quantity we are looking for. This is the object that Blinovsky was really dealing with. Note that this is in the same spirit as the quantity \eqref{eqn:pairwise_dist_plotkin} considered in the double counting argument in the proof of the classical Plotkin bound. Blinovsky used $\rdc$ as a proxy to finally bound 
    \[\exptover{\cL\sim\binom{\cC}{L}}{\rcheb(\cL)}.\]
\end{enumerate}
By extracting a constant weight subcode and applying the double counting trick (and using convexity of a certain function), Blinovsky showed that
\begin{lemma}
\label{lem_blinovsky_lem}
Let $\lambda\in[0,1/2]$ and fix $R=1-H(\lambda)$. Then \[\rdc\le\sum_{i=1}^{\ceil{L/2}}\frac{\binom{2i-2}{i-1}}{i}(\lambda(1-\lambda))^i.\]
\end{lemma}
Apparently, by definition, we have
\[\rld\ge\ravg,\quad\rdc\ge\ravg.\]
So Lemma \ref{lem_blinovsky_lem} automatically holds for $\ravg$. However, a priori the relation between $\rld$ and $\rdc$ is unclear. Surprisingly, Blinovsky showed that it is ``okay" to replace the first and third optimization with averaging, in the sense that
\begin{lemma}
\label{lem:blinovsky_avg}
For any infinite sequence of codes $\cC_n$, there exists an infinite sequence of subcodes $\cC_n'\subseteq\cC_n$ such that $\rld(\cC')=\ravg(\cC')+o(n)$.
\end{lemma}
The proof involves an \emph{equidistant} subcode extraction step using Ramsey theory. Lemma \ref{lem:blinovsky_avg} implies that the same bound in Lemma \ref{lem_blinovsky_lem} holds for $\rld$ as well!

\subsection{Cohen--Litsyn--Z\'emor \cite{cohen-etal-it1994-gen-dist}}
Similar ideas were used to provide upper bounds on erasure list decoding capacity. A binary code is said to be \emph{$(p,L-1)$-erasure list decodable} if for any $\cT\in\binom{[n]}{n(1-p)}$ and any $\vy\in\bF_2^{(1-p)n}$, $\card{\curbrkt{\vx\in\cC\colon \vx|_\cT = \vy}}\le L-1$, where $\vx|_\cT$ denotes the restriction of $\vx$ to $\cT$, i.e., a vector of length $\card{\cT}$ only consisting of components from $\vx$ indexed by elements in $\cT$. The erasure list decoding radius $\rlderas$ and the $(p,L-1)$-erasure list decoding capacity $C_{L-1,\mathrm{eras}}$ are defined in the same manner.  Cohen--Litsyn--Z\'emor \cite{cohen-etal-it1994-gen-dist} showed that
\begin{theorem}[\cite{cohen-etal-it1994-gen-dist}]
\label{thm:erasure_ld}
$C_{L,\mathrm{eras}}\le1-H(\lambda)$, where $\lambda$ is the unique root of the equation $\lambda^{L+1} + (1-\lambda)^{L+1} = 1-p$ in $[0,1/2]$. 
\end{theorem}
The idea is essentially again double counting. Here, it turns out that the right object to be counted is the \emph{erasure radius} of a list $\cL$,
\[\reras\coloneqq\card{\curbrkt{i\in[n]\colon \vx(i)\text{ are the same }\forall\vx\in\cL}}.\]
Extracting a subcode living on a sphere (followed by shifting out the center to get a constant weight code $\cC'$) and conducting similar calculations on
\[\exptover{\cL\sim\binom{\cC'}{L}}{\reras(\cL)},\]
allow the authors  to conclude Theorem \ref{thm:erasure_ld}. 
\begin{remark}
The original paper \cite{cohen-etal-it1994-gen-dist} was stated for  \emph{generalized distance} which is an equivalent object and can be mapped to erasure list decoding radius  via a well-known connection. The above version was presented in Guruswami's PhD thesis \cite{guruswami-lncs}. 
\end{remark}

\subsection{Wang--Budkuley--Bogdanov--Jaggi \cite{wang-budkuley-bogdanov-jaggi-2019-omniscient-avc}}
As mentioned, our work is a continuation of the prior work \cite{wang-budkuley-bogdanov-jaggi-2019-omniscient-avc} which a subset of authors were involved in. We refer the readers to the corresponding paragraphs in Sec. \ref{sec:warmup} and Sec. \ref{sec:contributions} for review of their work and comparison with ours.

\section{Organization of the paper}
\label{sec:organization}
In Sec. \ref{sec:warmup}  we have seen numeric examples that illustrate our results.
In Sec. \ref{sec:intro} we properly motivated the problem and introduced relevant background in coding theory. 
Our contributions in this paper were listed in details in Sec. \ref{sec:contributions}.
In Sec. \ref{sec:techniques} we reviewed various techniques  used in this paper and highlighted our innovations.
Prior works that our results build on and push forward were surveyed in Sec. \ref{sec:prior_work}.

The rest of the paper is organized as follows. 
We fix our notational conventions in Sec. \ref{sec:notation} and provide necessary preliminaries, especially \emph{the method of types} in information theory, in Sec. \ref{sec:prelim}. 
We develop basic notions that will be used throughout the paper in Sec. \ref{sec:basic_def}. In particular, \emph{general adversarial channels} and objects associated to them will be introduced in this section.
In Sec. \ref{sec:ld_cap} we prove the list decoding capacity theorem for general adversarial channels when $L$ is  asymptotically large.
Furthermore, we obtain \emph{tight} list size bounds for \emph{most} capacity-achieving constant composition codes. 
In Sec. \ref{sec:achievability} and Sec. \ref{sec:converse} we show sufficiency and necessity, respectively, of the criterion we obtain for the existence of exponential-sized $(L-1)$-list decodable codes (where $L$ is a arbitrary universal constant) for general adversarial channels.
In Sec. \ref{sec:rethinking_converse} we make two remarks on the converse, which is  technically the most challenging piece of our work.
In Sec. \ref{sec:sanity_checks} we verify the correctness of our characterization in Sec. \ref{sec:achievability} and Sec. \ref{sec:converse} by running it on the problem specialized to a typical coding theory model which has been understood in prior works \cite{blinovsky-1986-ls-lb-binary,alon-bukh-polyanskiy-2018-ld-zero-rate}.
In Sec. \ref{sec:blinovsky_revisited}, utilizing tools developed and facts proved in this paper, we rigorously rederive Blinovsky's \cite{blinovsky-1986-ls-lb-binary} results. We obtain more intuitive expressions and demystify his calculations.
In Sec. \ref{sec:gv_vs_cloud} we evaluate bounds on unique decoding capacity ($L=2$) in \cite{wang-budkuley-bogdanov-jaggi-2019-omniscient-avc} under  a typical coding theory model. 
We conclude the paper and list several open questions and future directions in Sec. \ref{sec:rk_and_open}.
Some calculations and background knowledge are deferred to Appendices \ref{app:cp_cop}, \ref{app:bound_hypergraph_ramsey_number}, \ref{app:rw_dist} and \ref{app:blinovsky_abp}.

\section{Notation}
\label{sec:notation}

\noindent\textbf{Conventions.}
Sets are denoted by capital letters in calligraphic typeface, e.g., $\cC,\cI$, etc. 
Random variables are denoted by lower case letters in boldface or capital letters in plain typeface, e.g., $\bfm,\bfx,\bfs,U,W$, etc. Their realizations are denoted by corresponding lower case letters in plain typeface, e.g., $m,x,s,u,w$, etc. Vectors (stochastic or deterministic) of length $n$, where $n$ is the blocklength, are denoted by lower case letters with an underline, e.g., $\vbfx,\vbfs,\vx,\vs$, etc. The $i$-th entry of a vector $\vx\in\cX^n$ is denoted by $\vx(i)$  since we can alternatively think $\vx$ as a function from $[n]$ to $\cX$. Same for random vector $\vbfx(i)$. Matrices are denoted by capital letters in boldface, e.g., $\bfP,\mathbf{\Sigma}$, etc. Similarly, the $(i,j)$-th entry of a matrix $\bfG\in\cX^{n\times\kappa}$ is denoted by $\bfG(i,j)$. Letter $\bfI$ is reserved for identity matrix. We sometimes write $\bfI_n$ to explicitly specify that it is an $n\times n$ square identity matrix. Tensors are denoted by capital letters in plain typeface, e.g., $T,P$, etc.

\noindent\textbf{Functions.}
We use the standard Bachmann--Landau (Big-Oh) notation for asymptotics of functions in positive integers. 

For $x\in\bR$, let $[x]^+\coloneqq\max\curbrkt{x,0}$.

For two real valued functions $f,g$ on the same domain $\Omega$, let $fg$ and $f/g$ denote the functions obtained by multiplying and taking the ratio of the images of $f$ and $g$ point-wise, respectively. That is, for $\omega\in\Omega$,
\[(fg)(\omega)  = f(\omega)g(\omega),\quad (f/g)(\omega) = f(\omega)/g(\omega).\]
In particular, for types or distributions, we can write
$
\tau_{\bfx,\bfy} = \tau_{\bfx}\tau_{\bfy|\bfx}, \tau_{\bfy|\bfx} = \tau_{\bfx,\bfy}/\tau_{\bfx}
$,
or 
$
P_{\bfx,\bfy} = P_{\bfx}P_{\bfy|\bfx}, P_{\bfy|\bfx} = P_{\bfx,\bfy}/P_{\bfx}
$ and so on.

For two real-valued functions $f(n),g(n)$ in positive integers, we say that $f(n)$ \emph{asymptotically equals} $g(n)$, denoted $f(n)\asymp g(n)$, if 
\[\lim_{n\to\infty}\frac{f(n)}{g(n)} = 1.\]
For instance, $2^{n+\log n}\asymp2^{n+\log n}+2^n$, $2^{n+\log n}\not\asymp2^n$.
 We write $f(n)\doteq g(n)$ (read $f(n)$ dot equals $g(n)$) if the coefficients of the dominant terms in the exponents of $f(n)$ and $g(n)$ match,
\[\lim_{n\to\infty}\frac{\log f(n)}{\log g(n)} = 1.\]
For instance, $2^{3n}\doteq2^{3n+n^{1/4}}$, $2^{2^n}\not\doteq2^{2^{n+\log n}}$. Note that $f(n)\asymp g(n)$ implies $f(n)\doteq g(n)$, but the converse is not true.

For any $q\in\bR_{>0}$, we write $\log_q(\cdot)$ for the logarithm to the base $q$. In particular, let $\log(\cdot)$ and $\ln(\cdot)$ denote logarithms to the base two and $e$, respectively.

\noindent\textbf{Sets.}
For any two sets $\cA$ and $\cB$ with additive and multiplicative structures, let $\cA+\cB$ and $\cA\cdot\cB$ denote the Minkowski sum and Minkowski product of them which are defined as
\[\cA+\cB\coloneqq\curbrkt{a+b\colon a\in\cA,b\in\cB},\quad\cA\cdot\cB\coloneqq\curbrkt{a\cdot b\colon a\in\cA,b\in\cB},\]
respectively.
If $\cA=\{x\}$ is a singleton set, we write $x+\cB$ and $x\cB$ for $\{x\}+\cB$ and $\{x\}\cdot\cB$.

For any finite set $\cX$ and any integer $0\le k\le |\cX|$, we use $\binom{\cX}{k}$ to denote the collection of all subsets of $\cX$ of size $k$.
\[\binom{\cX}{k}\coloneqq\curbrkt{\cY\subseteq\cX\colon\card{\cY}=k}.\]

For $M\in\bZ_{>0}$, we let $[M]$ denote the set of first $M$ positive integers $\{1,2,\cdots,M\}$. 

For any $\cA\subseteq\Omega$, the indicator function of $\cA$ is defined as, for any   $x\in\Omega$, 
\[\ind{\cA}(x)=\begin{cases}1,&x\in \cA\\0,&x\notin \cA\end{cases}.\]
At times, we will slightly abuse notation by saying that $\ind{\mathsf{A}}$ is $1$ when event $\mathsf{A}$ happens and zero otherwise. Note that $\one_\cA(\cdot)=\indicator{\cdot\in\cA}$.


\noindent\textbf{Geometry.}
For any $\vx\in\bF_q^n$, let $\wth{\vx}$ denote the Hamming weight of $\vx$, i.e., the number of nonzero entries of $\vx$.
\[\wth{\vx}\coloneqq\card{\curbrkt{i\in[n]\colon \vx(i)\ne0}}.\]
For any $\vx,\vy\in\bF_q^n$, let $\disth{\vx}{\vy}$ denote the Hamming distance between $\vx$ and $\vy$, i.e., the number of locations where they differ.
\[\disth{\vx}{\vy}\coloneqq\wth{\vx-\vy}=\card{\curbrkt{i\in[n]\colon \vx(i)\ne\vy(i)}}.\]
Balls and spheres in $\bF_q^n$ centered around some point $\vx\in\bF_q^n$ of certain radius $r\in\curbrkt{0,1,\cdots,n}$ w.r.t. the Hamming metric are defined as follows.
\[\bham^n(\vx,r)\coloneqq\curbrkt{\vy\in\bF_q^n\colon \disth{\vx}{\vy}\le r},\quad\sham^n(\vx,r)\coloneqq\curbrkt{\vy\in\bF_q^n\colon \disth{\vx}{\vy}=r}.\]
We will drop the subscript and superscript for the associated metric and dimension  when they are clear from the context.


\noindent\textbf{Probability.}
For a finite set $\cX$, $\Delta(\cX)$ denotes the probability simplex on $\cX$, i.e., the set of all probability distributions supported on $\cX$,
\[\Delta(\cX)\coloneqq\curbrkt{P_\bfx\in[0,1]^{\card{\cX}}\colon\sum_{x\in\cX}P_\bfx(x) = 1 }.\]
Similarly, $\Delta\paren{\cX\times\cY}$ denotes the probability simplex on $\cX\times\cY$,
\[\Delta\paren{\cX\times\cY}\coloneqq\curbrkt{P_{\bfx,\bfy}\in[0,1]^{\cardX\times\cardY}\colon\sum_{x\in\cX}\sum_{y\in\cY} P_{\bfx,\bfy}(x,y) = 1 }.\]
Let $\Delta(\cY|\cX)$ denote the set of all conditional distributions,
\[\Delta(\cY|\cX)\coloneqq\curbrkt{P_{\bfy|\bfx}\in\bR^{\cardX\times\cardY}\colon P_{\bfy|\bfx}(\cdot|x)\in\Delta(\cY),\;\forall x\in\cX}.\]
The general notion for multiple spaces is defined in the same manner.

The probability mass function (p.m.f.) of a discrete random variable $\bfx$ or a random vector $\vbfx$ is denoted by $P_{\bfx}$ or $P_{\vbfx}$. 
Here we use the following shorthand notation to denote the probability that $\bfx$ or $\vbfx$ distributed according to $P_\bfx$ or $P_\vbfx$ takes a particular value.
\[P_\bfx(x) \coloneqq \probover{\bfx\sim P_\bfx}{\bfx = x},\quad P_\vbfx(\vx) = \probover{\vbfx\sim P_\vbfx}{\vbfx = \vx},\]
for some $x\in\cX$ or $\vx\in\cX^n$.
If every entry of $\vbfx$ is independently and identically distributed (i.i.d.) according to $P_{\bfx}$, then we write $\vbfx\sim P_{\bfx}^{\tn}$, where $P_\bfx^\tn$ is a product distribution defined as
\[P_{\vbfx}(\vx)=P_{\bfx}^{\tn}(\vx)\coloneqq\prod_{i=1}^nP_{\bfx}(\vx(i)).\]

Let $\unif(\Omega)$ denote the uniform distribution over some probability space $\Omega$. 

For a joint distribution $P_{\bfx,\bfy}\in\Delta(\cX\times\cY)$, let $\sqrbrkt{P_{\bfx,\bfy}}_\bfx\in\Delta(\cX)$ denote the \emph{marginalization} onto the  variable $\bfx$, i.e., for $x\in\cX$,
\[\sqrbrkt{P_{\bfx,\bfy}}_\bfx(x) = \sum_{y\in\cY}P_{\bfx,\bfy}(x,y).\]
Sometimes we simply write it as $P_\bfx$  when notation is not overloaded.

\noindent\textbf{Algebra.}
Let $\|\cdot\|_p$ denote the standard $\ell^p$-norm. Specifically, for any $\vx\in\bR^n$,
\[\|\vx\|_p\coloneqq\paren{\sum_{i=1}^n\abs{\vx(i)}^p}^{1/p}.\]
For brevity, we also write $\|\cdot\|$ for the $\ell^2$-norm.

An order-$k$ dimension-$(n_1,\cdots,n_k)$ tensor $T$ is a multidimensional array. It can be thought as a function on the product space $[n_1]\times\cdots\times[n_k]$ which identifies the value of each of its entries.
\begin{align*}
    \begin{array}{rlll}
        T\colon & [n_1]\times\cdots\times[n_k] &\to &\bR \\
         & (i_1,\cdots,i_k) &\mapsto & T(i_1,\cdots,i_k),
    \end{array}
\end{align*}
where, as usual, we use $T(i_1,\cdots,i_k)$ to denote its $(i_1,\cdots,i_k)$-th entry.

We list below various sets/spaces of matrices and tensors that we are going to use in this paper. Without specification, all matrices and tensors are over the real number field. 
\begin{itemize}
    \item The space of $n\times m$ matrices:
    \[\mat_{n\times m}\coloneqq\curbrkt{\bfM\in\bR^{n\times m}}\cong\bR^{n\cdot m}.\]
    When $n=m$,  we write $\mat_{n}$ for the space of square matrices of dimension $n$.
    \item The space of order-$k$ dimension-$(n_1,\cdots,n_k)$ tensors:
    \[\ten_{n_1,\cdots,n_k}^{\otimes k}\coloneqq\curbrkt{T\in\bR^{n_1\times\cdots\times n_k}}\cong\bR^{n_1\cdots n_k}.\]
    If every dimension of $T$ is the same, i.e., $n_1 =\cdots = n_k=n$, then we write $\ten_{n}^{\otimes k}$ for the space of equilateral tensors of order $k$ and dimension $n$. 
    \item 
    Definitions of sets of \emph{symmetric} ($\sym$), \emph{non-negative} ($\nn$), \emph{doubly non-negative} ($\dnn$), \emph{positive semidefinite} ($\psd$), \emph{completely positive} ($\cp$), \emph{copositive} ($\cop$), etc. of matrices and tensors are deferred to the corresponding sections.
\end{itemize}
Note that $\mat_{n,m} = \ten_{n,m}^{\otimes 2}$. When the order of the tensors is $k=2$, namely matrices, we drop the superscript $\otimes 2$. 

For a tensor $T\in\ten_{n_1,\cdots,n_k}^{\otimes k}$, we use $\normf{T}$ to denote the \emph{Frobenius norm} of $T$, which is the $\ell^2$ norm when $T$ is vectorized into a length-$n_1\cdots n_k$ vector.
\[\normf{T}\coloneqq\paren{\sum_{(i_1,\cdots,i_k)\in[n_1]\times\cdots\times[n_k]}T(i_1,\cdots,i_k)^2}^{1/2}.\]
We use $\normsav{T}$ to denote the \emph{sum-absolute-value norm} of $T$ which is the $\ell^1$ norm after vectorization. 
\[\normsav{T}\coloneqq\sum_{(i_1,\cdots,i_k)\in[n_1]\times\cdots\times[n_k]}\abs{T(i_1,\cdots,i_k)}.\]
Similarly, define
\[\normmav{T}\coloneqq\max_{(i_1,\cdots,i_k)\in[n_1]\times\cdots\times[n_k]}\abs{ T(i_1,\cdots,i_k) }\]
to be the \emph{max-absolute-value} norm of $T$, which is the $\ell^\infty$ norm when viewed as a vector.

Note that the Frobenius norm, sum-absolute-value norm and max-absolute-value are different from the matrix/tensor 2-norm, 1-norm and $\infty$-norm. However, they do coincide with the corresponding vector norm when the order of the tensor is one.

We endow the matrix/tensor space with an inner product. For tensors $T_1$ and $T_2$ both of order $k$ and dimension $(n_1,\cdots,n_k)$, 
\[\inprod{T_1}{T_2}=\sum_{(i_1,\cdots,i_k)\in[n_1]\times\cdots\times[n_k]}T_1(i_1,\cdots,i_k)T_2(i_1,\cdots,i_k).\]
When $T_1,T_2$ are matrices, the above definition agrees with the \emph{Frobenius inner product}, which is alternatively defined as $\tr\paren{T_1^\top T_2}$. When $T_1,T_2$ are vectors, this inner product becomes the standard inner product associated to $\bR^n$ as a Hilbert space, which is denoted by the same notation without confusion. 

Let $S_n$ denote the \emph{symmetric group} of degree $n$ consisting of $n!$ permutations on $[n]$. Permutations are typically denoted by Greek letters.

\noindent\textbf{Information theory.}
We use $H(\cdot)$ to interchangeably denote  the binary entropy function or the Shannon entropy; the exact meaning  will usually be clear from the context.
In particular, for any $p\in[0,1]$, $H(p)$ denotes the binary entropy 
\[H(p)=p\log\frac{1}{p}+(1-p)\log\frac{1}{1-p}.\]
For a distribution $P\in\Delta(\cX)$ on a finite alphabet $\cX$ or a random variable $\bfx\sim P$ distributed according to $P$, the Shannon entropy of $P$ or $\bfx$ is defined similarly as
\[H(P) = H(\bfx)\coloneqq\sum_{x\in\cX}P_\bfx(x)\log\frac{1}{P_\bfx(x)}.\]

For two distributions $P,Q\in\Delta(\cX)$ on the same alphabet $\cX$, the \emph{Kullback--Leibler (KL) divergence} between them is defined as
\[D(P\|Q) \coloneqq \sum_{x\in\cX}P(x)\log\frac{P(x)}{Q(x)}.\]

If $\bfx,\bfy$ are jointly distributed according to $P_{\bfx,\bfy}\in\Delta(\cX\times\cY)$, then 
\begin{itemize}
    \item Their \emph{joint entropy} is defined as
    \[H(\bfx,\bfy)=H(P_{\bfx,\bfy})\coloneqq\sum_{x\in\cX}\sum_{y\in\cY}P_{\bfx,\bfy}(x,y)\log\frac{1}{P_{\bfx,\bfy}(x,y)};\]
    \item  Their \emph{mutual information} is defined as 
    \begin{align*}
        I(\bfx;\bfy)\coloneqq&D\paren{P_{\bfx,\bfy}\|P_{\bfx}P_{\bfy}}\\
        =&\sum_{x\in\cX}\sum_{y\in\cY}P_{\bfx,\bfy}(x,y)\log\frac{P_{\bfx,\bfy}(x,y)}{P_\bfx(x)P_\bfy(y)}\\
        =&\sum_{y\in\cY}P_\bfy(y)\sum_{x\in\cX}P_{\bfx|\bfy}(x|y)\log\frac{P_{\bfx|\bfy}(x|y)}{P_\bfx(x)}.
    \end{align*}
\end{itemize}

If the conditional distribution of $\bfy$ given $\bfx$ is $P_{\bfy|\bfx}\in\Delta(\cY|\cX)$, then the \emph{conditional entropy} of $\bfy$ given $\bfx$ is defined as
\begin{align*}
    H(\bfy|\bfx)\coloneqq&\sum_{x\in\cX}P_{\bfx}(x)H(\bfy|\bfx = x)\\
    =&\sum_{x\in\cX}\sum_{y\in\cX}P_{\bfx,\bfy}(x,y)\log\frac{P_{\bfx}(x)}{P_{\bfx,\bfy}(x,y)}.
\end{align*}

It is  easy to check that different  definitions above for the same quantities are consisted with each other. 


\section{Preliminaries}
\label{sec:prelim}

\begin{lemma}[Stirling's approximation]
For any $n\in\bZ_{>0}$, 
\[{n!}\asymp{\sqrt{2\pi n}\paren{\frac{n}{e}}^n}.\]
\label{lem:stirling}
\end{lemma}

\begin{corollary}[Asymptotics of multinomials]
\label{cor:multinomial}
For any positive integers $n\ge q$ and any $q$-partition $(n_1,\cdots,n_q)$ of $n$ ($n_1+\cdots+n_q=n$, $n_i\ge0$ for every $i$), 
    \[\binom{n}{n_1,\cdots,n_q}\doteq2^{nH(P)},\]
    where  $P\in\Delta([q])$ is an empirical distribution such that for $i\in[q]$, 
    \[P(i) = \frac{n_i}{n}.\]
More precisely, we have
\[\binom{n}{n_1,\cdots,n_q} \asymp \nu(n)^{-1}2^{nH(P)} ,\]
where $\nu(n)$ is a polynomial defined as
\[\nu(n)\coloneqq \paren{{2\pi n}}^{\frac{q - 1}{2}}\paren{\prod_{i = 1}^{q}P(i)}^{\frac{1}{2}}.\]
\end{corollary}

\begin{fact}[Approximation of binomials]
For any positive integers $n\ge k$, 
\begin{align}
    \paren{\frac{n}{k}}^k\le\binom{n}{k}&\le\paren{\frac{en}{k}}^k,\label{eqn:binom_bd_tight}\\
    (n-k)^k\le(n-k+1)^k\le\binom{n}{k}&\le n^k.\label{eqn:binom_bd_loose}
\end{align}
\end{fact}

Without loss of generality, write $\cX=\curbrkt{x_1,\cdots,x_{\card{\cX}}}$. For $\vx\in\cX^n$ and $x\in\cX$, let
\[N_x(\vx)\coloneqq\card{\curbrkt{i\in[n]\colon \vx(i)=x}},\]
which counts the number of occurrences of a symbol $x$ in a vector $\vx$.

\begin{definition}[Types]
For a length-$n$ vector $\vx$ over a finite alphabet $\cX$, the type $\tau_{\vx}$ of $\vx$ is a length-$\card{\cX}$ (empirical) probability vector (or the histogram of $\vx$), i.e., $\tau_{\vx}\in[0,1]^{\card{\cX}}$ has entries
\[\tau_{\vx}(x)\coloneqq\frac{N_x(\vx)}{n}\]
for any $x\in\cX$.
\end{definition}

\begin{definition}[Joint types and conditional types]
The \emph{joint type} $\tau_{\vx,\vy}\in[0,1]^{\card{\cX}\times\card{\cY}}$ of two vectors $\vx\in\cX^n$ and $\vy\in\cY^n$ is defined as
\[\tau_{\vx,\vy}(x,y)=\frac{N_{x,y}(\vx,\vy)}{n}\]
for $x\in\cX$ and $y\in\cY$, where
\[N_{x,y}\paren{\vx,\vy}\coloneqq\card{\curbrkt{i\in[n]\colon \vx(i) = x,\;\vy(i) = y}}. \]

The conditional type $\tau_{\vy|\vx}\in[0,1]^{\card{\cX}\times\card{\cY}}$ of a vector $\vy\in\cY^n$ given another vector $\vx\in\cX^n$ is defined as 
\[\tau_{\vy|\vx}(y|x) = \frac{N_{x,y}\paren{\vx,\vy}}{N_x\paren{\vx}}.\]
\end{definition}

\begin{remark}[Types vs. distributions]
Types are empirical distributions of length-$n$ vectors. They can only take rational values, in particular, $a/n$ for $a\in\curbrkt{0,1,\cdots,n}$. For a fixed $n$ and finite alphabets, there are only $\poly(n)$ many types. However, there are uncountably infinitely many distributions on any finite alphabets and they form a probability simplex. 
\end{remark}

\begin{definition}[Set of types]
We use $\cP^{(n)}(\cX)$ to denote the set of all possible types of length-$n$ vectors over $\cX$.
\[\cP^{(n)}(\cX) = \curbrkt{ \tau_{\vx}\colon \vx\in\cX^n }.\]
Similarly, define
\begin{align*}
	\cP^{(n)}(\cX,\cY) = &\curbrkt{ \tau_{\vx,\vy}\colon \vx\in\cX^n,\;\vy\in\cY^n },\\
	\cP^{(n)}(\cY|\vx) = &\curbrkt{ \tau_{\vy|\vx}\colon \vy\in\cY^n },\\
	\cP^{(n)}(\cY|\cX) =& \curbrkt{\tau_{\vy|\vx}\colon \vx\in\cX^n,\;\vy\in\cY^n}
\end{align*}
to be
\begin{enumerate}
	\item the set of all joint types;
	\item the set of all conditional types of $\vy$ given a particular $\vx$;
	\item the set of all conditional types of $\vy$ given some $\vx$,
\end{enumerate}
respectively.
\end{definition}

\begin{lemma}[Types are dense in distributions]
\label{lem:types_dense_in_distr}
The union of sets of  types of all possible blocklengths is dense in the set of distributions, i.e.,
\[\bigcup_{n=1}^\infty\cP^{(n)}(\cX) \]
is dense in $\Delta(\cX)$.
This holds true for joint types and conditional types as well.
\end{lemma}

\begin{lemma}[Number of types]
When alphabet sizes are constants, the number of types of length-$n$ vectors is polynomial in $n$. To be precise, the number of types of length-$n$ vectors over $\cX$ is 
\begin{equation}
    \card{\cP^{(n)}(\cX)} = \binom{n+\cardX - 1}{\cardX - 1}.
    \label{eqn:number_type}
\end{equation}
 For a vector $\vx\in\cX^n$ of type $\tau_\vx$, the number of conditional types of length-$n$ vectors  over $\cY$ given $\vx$ is 
\begin{equation}
    \label{eqn:number_cond_type_given}
    \card{\cP^{(n)}(\cY|\vx)} = \prod_{x\in\cX}\binom{\tau_\vx(x)n + \cardY - 1}{\cardY - 1}.
\end{equation}
The number of conditional types of $\cY$-valued vectors given some $\cX$-valued vector is 
\begin{equation}
	\label{eqn:number_cond_type}
	\cP^{(n)}(\cY|\cX) = \sum_{\tau_\bfx\in\cP^{(n)}(\cX)}\prod_{x\in\cX}\binom{\tau_\bfx(x)n + \cardY - 1}{\cardY - 1}.
\end{equation}

The following elementary bounds from \cite{csiszar-korner-book2011} are sufficient for our purposes in this paper. 
\begin{align*}
	\card{ \cP^{(n)}(\cX) }\le & (n+1)^{\cardX},\\
	\card{ \cP^{(n)}(\cY|\vx) }\le\card{ \cP^{(n)}(\cY|\cX) }\le & (n+1)^{\cardX\cdot\cardY}.
\end{align*}
\end{lemma}

\begin{definition}[Type classes]
\label{def:def_type_classes}
Define type class $\cT_\vbfx(\tau_\bfx)$ w.r.t. a type $\tau_\bfx\in\cP^{(n)}(\cX)$ as
\[\cT_\vbfx(\tau_\bfx)\coloneqq\curbrkt{\vx\in\cX^n\colon \tau_\vx=\tau_\bfx}.\]
Joint type classes and conditional type classes can be defined in a similar manner. The joint type class $\cT_{\vbfx,\vbfy}\paren{\tau_{\bfx,\bfy}}$ w.r.t. a joint type $\tau_{\bfx,\bfy}\in\cP^{(n)}(\cX\times\cY)$ is defined as
\[\cT_{\vbfx,\vbfy}\paren{\tau_{\bfx,\bfy}}\coloneqq\curbrkt{\paren{\vx,\vy}\in\cX^n\times\cY^n\colon \tau_{\vx,\vy} = \tau_{\bfx,\bfy} }.\]
The conditional type class $\cT_{\vbfy|\vx}\paren{\tau_{\bfy|\vx}}$ w.r.t. a conditional type $\tau_{\bfy|\vx}\in\cP^{(n)}(\cY|\vx)$ given a vector $\vx\in\cX^n$ is defined as
\[\cT_{\vbfy|\vx}\paren{\tau_{\bfy|\vx}}\coloneqq\curbrkt{ \vy\in\cY^n\colon \tau_{\vy|\vx} = \tau_{\bfy|\vx} }.\]
The conditional type class $\cT_{\vbfy|\vbfx}\paren{\tau_{\bfy|\bfx}}$ w.r.t. a conditional type ${\tau_{\bfy|\bfx}}\in\cP^{(n)}(\cY|\cX)$ given some vector of type $\tau_\bfx\in\cP^{(n)}(\cX)$ is defined as
\begin{align}
    \cT_{\vbfy|\vbfx}\paren{\tau_{\bfy|\bfx}} \coloneqq& \bigcup_{\tau_\bfx\in\cP^{(n)}(\cX)}\bigcup_{\vx'\in\cT_\vbfx(\tau_\bfx)}\cT_{\bfy|\vx'}\paren{\tau_{\bfy|\vx'}}  \label{eqn:def_cond_type_class_one}\\
    =&\curbrkt{ \vy\in\cY^n\colon\exists\vx'\in\cX^n,\; \tau_{\vy|\vx} = \tau_{\bfy|\vx'} },\label{eqn:def_cond_type_class_two}
\end{align}
where in Eqn. \eqref{eqn:def_cond_type_class_one} $\vx'\in\cT_{\vbfx}\paren{\tau_\bfx}$ can be chosen arbitrarily and $\tau_{\bfy|\vx'} = \tau_{\bfy|\bfx}$ in both Eqn. \eqref{eqn:def_cond_type_class_one} and \eqref{eqn:def_cond_type_class_two}.
\end{definition}

\begin{remark}
We will also write $\tau_\bfx,\tau_{\bfx,\bfy},\tau_{\bfy|\vx}, \tau_{\bfy|\bfx}$ etc. for generic types that are taken from the corresponding sets of types even if they do not come from  instantiated vectors. For instance, $\tau_\bfx$ is a type in $\cP^{(n)}(\cX)$ corresponding to any $\vx\in\cT_\vbfx(\tau_\bfx)$. The particular choice of $\vx$ is not important and will not be specified. This is to explicitly distinguish between types and distributions.
\end{remark}

\begin{lemma}[Size of type classes]
\label{lem:size_type_classes}
\begin{enumerate}
    \item For any type $\tau_\bfx\in\cP^{(n)}(\cX)$, 
    \[\card{\cT_{\vbfx}(\tau_\bfx)}\doteq2^{nH(P_\bfx)}.\]
    \item For any vector $\vx\in\cX^n$ and any conditional type $\tau_{\bfy|\vx}\in\cP^{(n)}(\cY|\vx)$, 
    \[\card{\cT_{\vbfy|\vx}\paren{\tau_{\bfy|\vx}}}\doteq2^{nH(\bfy|\bfx)},\]
    where the conditional entropy is evaluated w.r.t. the joint type $\tau_\vx \tau_{\bfy|\vx}$. 
    \item For any conditional type $\tau_{\bfy|\bfx}\in\cP^{(n)}(\cY|\cX)$, 
    \[\card{ \cT_{\vbfy|\vbfx}\paren{\tau_{\bfy|\bfx}} }\doteq2^{n\max_{\tau_\bfx\in\cP^{(n)}(\cX)}H(\bfy|\bfx) },\]
    where the conditional entropy is evaluated w.r.t. the joint type $\tau_\bfx\tau_{\bfy|\bfx}$.
\end{enumerate}
\end{lemma}
\begin{proof}
\begin{enumerate}
    \item The number of sequences $\vx\in\cX^n$ of type $\tau_\bfx$ is precisely \[\binom{n}{\tau_\bfx(1),\cdots,\tau_\bfx(\card{\cX})}\] 
    and the claim follows from Lemma \ref{lem:stirling}.
    \item Given $\vx\in\cX^n$, the number of sequences $\vy\in\cY^n$ of conditional type $\tau_{\bfy|\vx}$ is precisely 
    \[\prod_{x\in\cX}\binom{\tau_\vx(x)}{\tau_{\bfy|\vx}(1|x),\cdots,\tau_{\bfy|\vx}(\card{\cY}|x)},\]
    and the lemma follows from \ref{lem:stirling}.
    \item Note that
    \[\cT_{\vbfy|\vx^*}\paren{\tau_{\bfy|\vx^*}}\le\card{ \cT_{\vbfy|\vbfx}\paren{\tau_{\bfy|\bfx}} } \le\card{\cP^{(n)}(\cX)}\cT_{\vbfy|\vx^*}\paren{\tau_{\bfy|\vx^*}},\]
    where $\vx^*\in\cT_\vbfx\paren{\tau_\bfx^*}$ is chosen arbitrarily and\footnote{In the $\argmax{}$, $\vx\in\cT_\vbfx\paren{\tau_\bfx}$ is arbitrary as well.}
    \[\tau_\bfx^* = \argmax{\tau_\bfx\in\cP^{(n)}(\cX)}\card{\cT_{\vbfy|\vx}\paren{\tau_{\vbfy|\vx}}}.\]
    The claim follows from Eqn. \eqref{eqn:number_type} and the previous claim.
\end{enumerate}
\end{proof}

\begin{lemma}
\label{lem:prob_vec_type}
If $\vbfx$ is generated using the product distribution $P_\bfx^{\otimes n}$, then for any $\vx\in\cT_\vbfx(P_\bfx)$, 
\[\prob{ \vbfx = \vx }=2^{-nH(P_\bfx)}.\]
Moreover,
\[\prob{ \vbfx\in\cT_\vbfx(P_\bfx) }\asymp\nu(n)^{-1}.\]
\end{lemma}
\begin{proof}
Both claims follow from elementary calculations. For the first one,
\begin{align}
    \prob{ \vbfx = \vx }=&\prod_{x\in\cX}P_\bfx(x)^{N_{x}(\vx)}\notag\\
    =&2^{\sum_{x\in\cX}N_{x}(\vx)\log P_\bfx(x)}\notag\\
    =&2^{n\sum_{x\in\cX}P_\bfx(x)\log P_\bfx(x)}\label{eqn:type_equals_distr}\\
    =&2^{-nH(P_\bfx)},\notag
\end{align}
where Eqn. \eqref{eqn:type_equals_distr} is because $\tau_\vx = P_\bfx$ and hence $N_x(\vx)/n = P_\bfx(x)$ for any $x\in\cX$.

For the second one,
\begin{align}
	\prob{ \vbfx\in\cT_\vbfx(P_\bfx) } = &\prob{ \tau_\vbfx = P_\bfx }\notag\\
	=&\binom{n}{P_\bfx(1),\cdots,P_\bfx(\cardX)}\prod_{x\in\cX}P_\bfx(x)^{nP_\bfx(x)}\notag\\
	\asymp& \nu(n)^{-1}2^{nH(P)}2^{-nH(P)}\label{eqn:multinomial_approx}\\
	=&\nu(n)^{-1},\notag
\end{align}
where Eqn. \eqref{eqn:multinomial_approx} is by Corollary \ref{cor:multinomial}.
\end{proof}

\begin{lemma}[Markov]
\label{lem:markov}
For any non-negative random variable $X$ and any positive number $x\in\bR_{>0}$,
\[\prob{X\ge x }\le\frac{\expt{X}}{x}.\]
\end{lemma}

\begin{lemma}[Chernoff]
\label{lem:chernoff}
Let $X_1,\cdots,X_n$ be independent (not necessarily identically distributed) $\curbrkt{0,1}$-valued random variables. Let 
\[X\coloneqq\sum_{i=1}^nX_i.\]
 Then
\begin{align*}
	\prob{X\ge(1+\eps)\expt{X}}\le& e^{-\frac{\eps^2}{3}\expt{X}},\\
	\prob{X\le(1-\eps)\expt{X}}\le& e^{-\frac{\eps^2}{2}\expt{X}},\\
	\prob{X\notin(1\pm\eps)\expt{X}}\le& 2e^{-\frac{\eps^2}{3}\expt{X}}.
\end{align*}
\end{lemma}

\begin{lemma}[Sanov]
Let $\cQ\subset\Delta\paren{\cX}$ be a subset of distributions such that it is equal to the closure of its interior. Let $\vbfx\sim P_\bfx^{\otimes n}$ be a random vector whose components are i.i.d. w.r.t. $P_\bfx$. Clearly $\vbfx$ is expected to have type $\expt{\tau_{\vbfx}} = P_\bfx$. Sanov's theorem determines the first-order exponent of the probability that the vector looks like coming from some distribution $Q\in\cQ$ empirically.
\[\prob{\tau_{\vbfx}\in\cQ}\doteq 2^{-n\min_{Q\in\cQ}D\paren{Q\|P_\bfx}}.\]
\label{thm:sanov}
\end{lemma}
\begin{remark}
One can view Sanov's theorem as a particular form of Chernoff bound. Since $\vbfx(i)$'s are independent, it gives the \emph{correct} exponent of $\prob{ \tau_\vbfx\in\cQ }$ up to lower order term rather than merely a  bound.
\end{remark}

\begin{lemma}[Anti-concentration]
\label{lem:anti_conc}
Let $X$ be a non-negative random variable. Then
\[\prob{X=0}\le\frac{\var{X}}{\expt{X}^2}.\]
\end{lemma}

\begin{lemma}[\cite{csiszar-korner-1981}]
\label{lem:ck}
Given arbitrary finite sets $\cU$ and $\cX$, for every $R>0$, sufficiently large $n$ and $\tau_\bfx\in\cP^{(n)}(\cX)$, there are $M=2^{nR}$ vectors $\cC=\curbrkt{\vx_i}_{1\le i \le M}\subset\cT_\vbfx(\tau_\bfx)$, such that for every $\vu\in\cU^n$ and conditional type $\tau_{\bfx|\vu}\in\cP^{(n)}(\cX|\vu)$, we have
\[\card{\cC\cap\cT_\vbfx\paren{\tau_{\bfx|\vu}}}\le3(n+1)^{\card{\cX}}2^{n\sqrbrkt{R-I(\bfu;\bfx)}^+},\]
where $I(\bfu;\bfx)$ is evaluated w.r.t. $\tau_{\vu,\bfx}=\tau_\vu\tau_{\bfx|\vu}$.
\end{lemma}

\begin{fact}[Binomial identities]
\begin{enumerate}
For any non-negative integers $n,K\in\bZ_{\ge0}$ and $0\le k\le n$, we have
    \begin{align}
        \binom{n}{k} = &\binom{n}{n - k}, \label{eqn:binom_symm}\\
        \binom{n}{k} = &\frac{n}{k}\binom{n - 1}{k - 1},\label{eqn:binom_recurse}\\
        \binom{n}{k} + \binom{n}{k + 1} = &\binom{n + 1}{k + 1},\label{eqn:binom_pascal}\\
        2^K = &\sum_{i = 0}^K\binom{n}{i}.\label{eqn:binom_thm}
    \end{align}
\end{enumerate}
\end{fact}

We list several  basic (in)equalities concerning information measures that we will frequently refer to.
\begin{fact}[Information (in)equalities]
The following inequalities hold for any random variables/distributions over finite sets.
\begin{align*}
    H(\bfx,\bfy) = &H(\bfx) + H(\bfy|\bfx)\\
    =&H(\bfy) + H(\bfx|\bfy)\\
    =&H(\bfx|\bfy) + H(\bfy|\bfx) + I(\bfx;\bfy)\\
    =&H(\bfx) + H(\bfy) - I(\bfx;\bfy),\\
    I(\bfx;\bfy) = & H(\bfx) - H(\bfx|\bfy)\\
    =&H(\bfy) - H(\bfy|\bfx)\\
    = &D\paren{P_{\bfx,\bfy}\|P_\bfx P_\bfy}.
\end{align*}
\end{fact}

\section{Basic definitions}
\label{sec:basic_def}

\begin{definition}[Adversarial channels]
An adversarial channel $\cA=(\cX,\lambda_\bfx,\cS,\lambda_\bfs,\cY,W_{\bfy|\bfx,\bfs})$ (Fig. \ref{fig:channel_general}) is a sextuple consisting of 
\begin{enumerate}
    \item an input alphabet $\cX$;
    \item a set of input constraints $\lambda_\bfx\subseteq\cP^{(n)}(\cX)$;
    \item a noise alphabet $\cS$;
    \item a set of noise constraints $\lambda_\bfs\subseteq\cP^{(n)}(\cS)$;
    \item an output alphabet $\cY$;
    \item a channel law given by a transition probability $W_{\bfy|\bfx,\bfs}\in\Delta(\cY|\cX\times\cS)$.
\end{enumerate}
\end{definition}

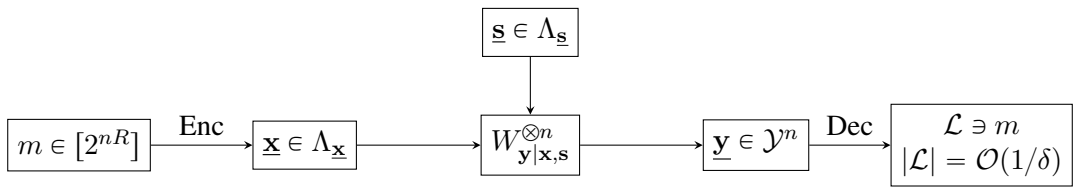
\begin{figure}[h]
	\centering
	\begin{tikzpicture}
	\node (m)[draw, align=center] {$m\in[2^{nR}]$};
	\node (x)[draw, align=center, right of=m, xshift = 2cm] {$\vbfx\in\Lambda_\vbfx$};
	\draw [->,>=stealth] (m) -- node[above]{Enc}   (x);
	
	\node (+)[draw, align=center, right of=x, xshift = 2cm]  {$W_{\bfy|\bfx,\bfs}^\tn$};
	\draw [->,>=stealth] (x) -- (+);
	
	\node (y)[draw, align=center, right of=+, xshift = 2cm] {$\vbfy\in\cY^n$};
	\draw [->,>=stealth] (+) -- (y);
	\node (s)[draw, align=center, above of=+, yshift = 0.5cm] {$\vbfs\in\Lambda_\vbfs$};
	\draw [->,>=stealth] (s) -- (+);
	
	\node (M)[draw, align=center, right of=y, xshift = 2cm] {$\cL\ni m$\\ $|\cL|=\cO(1/\delta)$ };
	\draw [->,>=stealth] (y) -- node[above]{Dec}   (M);
	\end{tikzpicture}
	\caption{General adversarial channels.}
	\label{fig:channel_general}
\end{figure}

\begin{remark}
In this paper, we are only concerned with finite alphabets of constant size independent of the blocklength $n$. 
\end{remark}

Specifically,
\begin{itemize}
    \item Though the alphabets $\cX,\cS$ and $\cY$ can be arbitrary finite sets, it is without loss of generality to realize them using the first $\cardX,\cardS$ and $\cardY$ positive integers, i.e., $\cX = \sqrbrkt{\cardX},\cS = \sqrbrkt{\cardS}$ and $\cY = \sqrbrkt{\cardY}$.\footnote{Under such realizations, these sets are \emph{not} necessarily equipped with real arithmetics or modular arithmetics. The metric, if one cares, would be specified by the channel function.}
    \item The input and noise constraint sets $\lambda_\bfx$ and $\lambda_\bfs$ are subsets of types $\cP^{(n)}(\cX)$ and $\cP^{(n)}(\cS)$. In this paper we assume they are \emph{convex} sets. Since there are polynomially many types in total, we can also think these collections of types as defined by intersections of hyperplanes or halfspaces, that is, types satisfying a certain finite number of linear (in the entries of the types) (in)equality constraints.
    \item In this paper, for technical simplicity, we assume that the channel transition function has only \emph{singleton} mass. That is, for each $x\in\cX,s\in\cS$, $W_{\bfy|\bfx,\bfs}(y|x,s) = 1$ only for one $y\in\cY$ and is zero for all other outputs. Equivalently, such degenerate distributions can be alternatively thought as \emph{deterministic} functions 
    \begin{align*}
        \begin{array}{rlll}
            W\colon & \cX\times\cS &\to &\cY \\
             &(x,s) &\mapsto &y,
        \end{array}
    \end{align*}
    where $y$ is the \emph{unique} output which is assigned the full probability, $W_{\bfy|\bfx,\bfs}(y|x ,s) =1$. Here we slightly abuse the notation and use the same letter for the channel transition distribution and the channel transition function (when the distribution is degenerate). Moreover, we use $\vy = W(\vx,\vs)$ (with the superscript $\tn$ being dropped) to denote the output of $n$ uses of the channel, or equivalently, the $n$-letter output of the function which acts on $(\vx,\vs)$ component by component.
    
    It seems this is a severe restriction (and turns out indeed to be so). Nevertheless, it is still a very first and significant step towards understanding general adversarial channels in full generality. The case where $W_{\bfy|\bfx,\bfs}$ is an arbitrary conditional distribution, or equivalently, the function $W$ is \emph{non-deterministic}, is interesting as well and is left as a future direction.  
    \item For notational convenience, let
    \begin{align*}
        \Lambda_\vbfx\coloneqq&\curbrkt{\vx\in\cX^n\colon \tau_\vx\in\lambda_\bfx}\\
        =&\bigcup_{\tau_\bfx\in\lambda_\bfx}\cT_\vbfx\paren{\tau_\bfx},\\
        \Lambda_\vbfs\coloneqq&\curbrkt{\vs\in\cS^n\colon \tau_\vs\in\lambda_\bfs}\\
        =&\bigcup_{\tau_\bfs\in\lambda_\bfs}\cT_\vbfs\paren{\tau_\bfs},
    \end{align*}
     be sets of codewords and error patterns of admissible types. 
\end{itemize}

\begin{example}
Our framework covers a large family of channel models, including most of the popular and well-studied ones. 
\begin{enumerate}
	\item The standard bit-flip channels. $\cX = \bF_2,\lambda_\bfx = \cP^{(n)}(\bF_2),\cS = \bF_2,\lambda_\bfs = \curbrkt{ \tau_\bfs\in\cP^{(n)}(\bF_2)\colon \tau_\bfs(1)\le p },\cY = \bF_2, y = W(x,s) = x\XOR s$.
	\item The standard $q$-ary channels. $\cX = \bZ_q,\lambda_\bfx = \cP^{(n)}(\bZ_q),\cS = \bZ_q,\lambda_\bfs = \curbrkt{\tau_\bfs\in\cP^{(n)}(\bZ_q)\colon \tau_\bfs(1) + \cdots + \tau_\bfs(q-1)\le p}$, $\cY = \bZ_q,y = W(x,s) = x+s\mod q$.
	\item The standard erasure channels. $\cS = \bZ_q,\lambda_\bfx = \cP^{(n)}(\bZ_q),\cS = \bF_2,\lambda_\bfs = \curbrkt{ \tau_\bfs\in\cP^{(n)}(\bF_2)\colon \tau_\bfs(1)\le p },\cY = \bZ_q\cup\curbrkt{\erasure}$,
	\begin{align*}
		y= W(x,s) = & \begin{cases}
			x,&s = 0\\
			\erasure,&s = 1
		\end{cases}.
	\end{align*}
	\item Weight constrained channels. Any of the above channels with $\lambda_\bfx = \curbrkt{ \tau_\bfx\in\cP^{(n)}(\cX)\colon 1-\tau_\bfx(0)\le w }$.
	\item $Z$-channels (or multiplier/AND channels). $\cX = \bF_2,\lambda_\bfx= \cP^{(n)}(\bF_2),\cS=\bF_2,\lambda_\bfs = \curbrkt{ \tau_\bfs\in\cP^{(n)}(\bF_2)\colon \tau_\bfs(1)\le p },\cY = \bF_2$,
	\begin{align*}
		y = W(x,s) = &
		\begin{cases}
			0,&s = 0\text{ or }x = 0\\
			x,&s = 1\text{ and }x=1
		\end{cases},
	\end{align*}
	or equivalently $y = W(x,s) = x\AND s$.
	\item Adder channels. $\cX=\curbrkt{0,1,\cdots,q-1},\lambda_\bfx = \cP^{(n)}(\cX), \cS = \curbrkt{0,1,\cdots,q-1}$, 
	\[\lambda_\bfs = \curbrkt{ \tau_\bfs\in\cP^{(n)}(\cS)\colon \tau_\bfs(1) + \cdots + \tau_\bfs(q-1)\le p },\] 
	$\cY = \curbrkt{0,1,\cdots,2(q-1)},y = W(x,s) = x+s$, where the addition is over $\bR$.
	\item Noisy typewriter channels. $\cX = \bZ_q,\lambda_\bfx = \cP^{(n)}(\bZ_q),\cS = \bF_2,\lambda_\bfs = \cP^{(n)}(\bF_2),\cY = \bZ_q,y = W(x,s) = x+s\mod q$.
	\item   OR channels (or $\reflectbox Z$-channels). $\cX = \bF_2,\lambda_\bfx = \cP^{(n)}(\bF_2),\cS = \bF_2,\lambda_\bfs = \curbrkt{\tau_\bfs\in\cP^{(n)}(\bF_2)\colon \tau_\bfs(1)\le p},\cY = \bF_2, y = W(x,s) = x \OR s$, 
	\item Channels under Lee distance. $\cX = \bZ_q,\lambda_\bfx= \cP^{(n)}(\bZ_q),\cS = \curbrkt{-\floor{\frac{q}{2}},-\floor{\frac{q}{2}} + 1,\cdots, \floor{\frac{q}{2}} - 1,\floor{\frac{q}{2}}}$,
	\[\lambda_\bfs = \curbrkt{ \tau_\bfs\in\cP^{(n)}(\cS)\colon\sum_{s = 1}^{\floor{q/2}}\paren{\tau_\bfs(s)-\tau_\bfs(-s)}\cdot s\le p },\]
	$\cY = \bZ_q,y = W(x,s) = x+s$ over the reals.
	\item Other more complicated channels, e.g., the one we defined in Sec. \ref{sec:warmup}.
\end{enumerate}

\begin{figure}
    \centering
    \begin{subfigure}{0.3\textwidth}
    	\centering
    	\includegraphics[scale = 1.5]{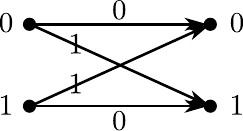}
    	\caption{Bit-flip channels.} 
    	\label{fig:channel_bitflip}
    \end{subfigure}~
    \begin{subfigure}{0.3\textwidth}
    	\centering
    	\includegraphics[scale = 1.5]{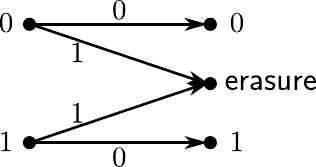}
    	\caption{Erasure channels.}
    	\label{fig:channel_erasure}
    \end{subfigure}~
    \begin{subfigure}{0.3\textwidth}
    	\centering
    	\includegraphics[scale = 1.5]{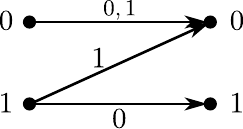}
    	\caption{$Z$-channels (or multiplier/AND channels).}
    	\label{fig:channel_z}
    \end{subfigure}\\
    \begin{subfigure}{0.3\textwidth}
    	\centering
    	\includegraphics[scale = 1.5]{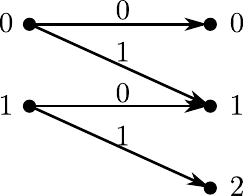}
    	\caption{Adder channels.}
    	\label{fig:channel_adder}
    \end{subfigure}~
    \begin{subfigure}{0.3\textwidth}
    	\centering
    	\includegraphics[scale = 1.5]{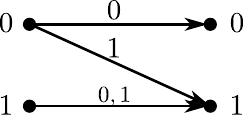}
    	\caption{OR channels.} 
    	\label{fig:channel_or}
    \end{subfigure}~
    \begin{subfigure}{0.3\textwidth}
    	\centering
    	\includegraphics[scale = 1.5]{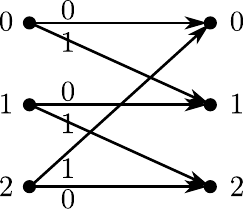}
    	\caption{Ternary noisy typewriter channels.}
    	\label{fig:channel_typewriter}
    \end{subfigure}
    \caption{Examples of various well-studied channel models.}
    \label{fig:ex_channel}
\end{figure}
\end{example}

\begin{definition}[Self-couplings]\label{def:self_couplings}
A joint distribution $P_{\bfx_1,\cdots,\bfx_L}\in\Delta(\cX^L)$ is said to be a \emph{$(P_\bfx,L)$-self-coupling} for some $P_\bfx\in\Delta(\cX)$ if all of its marginals equal $P_\bfx$, i.e.,  $\sqrbrkt{P_{\bfx_1,\cdots,\bfx_L}}_{\bfx_i}=P_\bfx$ for all $i\in[L]$. The set of all $(P_\bfx,L)$-self-couplings is denoted by $\cJ^\tl\paren{P_\bfx}$. 
\end{definition}

\begin{definition}[Codes]
\label{def:codes}
In general, a code $\cC$ is a subset of $\cX^n$. A code $\cC$ for an adversarial channel $\cA = (\cX,\lambda_\bfx,\cS,\lambda_\bfs,\cY,W_{\bfy|\bfx,\bfs})$ is a subset of $\Lambda_\vbfx$. $n$ is called the blocklength. Elements in $\cC$ are called codewords. The rate $R(\cC)$ of $\cC$ is defined as $R(\cC) = \paren{\log\cardC}/n$.
\end{definition}

\begin{definition}[Constant composition codes]\label{def:const_comp_codes}
A code $\cC\subset\cX^n$  is said to be $P_\bfx$-constant composition for some $P_\bfx\in\Delta(\cX)$ if the type of each codeword  is $P_\bfx$, i.e., $\tau_{\vx}=P_\bfx$ for every $\vx\in\cC$.
\end{definition}

\begin{lemma}
\label{lem:cc_code_suff}
For any code $\cC\subset\cX^n$ of rate $R$, there is a constant composition subcode $\cC'\subseteq\cC$ of asymptotically the same rate.
\end{lemma}
\begin{proof}
Let $\cC' = \cC\cap\cT_\vbfx\paren{\tau_\bfx^*}$, where
\[\tau_\bfx^*  = \argmax{ \tau_\bfx\in\cP^{(n)}(\cX) }\card{\cC\cap\cT_\vbfx\paren{\tau_\bfx}}\]
is the most common type in  $\cC$. By Lemma \ref{eqn:number_type} and Lemma \ref{lem:markov}, 
\[\card{\cC'}\ge\frac{\cardC}{(n+1)^\cardX} = 2^{nR + \cardX\log(n+1)},\]
which implies that $R(\cC')\asymp R(\cC)$ as $n$ grows.
\end{proof}

\begin{definition}[Confusability of tuples of vectors]\label{def:conf_tuples}
A list of $L$ distinct codewords $\vx_{1},\cdots,\vx_{L}\in\cX^n$ is said to be \emph{$L$-confusable} if there are $\vy\in\cY^n$ and $\vs_1,\cdots,\vs_{L}\in\Lambda_\vbfs$ such that 
$W\paren{\vx_{i},\vs_i}=\vy$
for all $i\in[L]$.
\end{definition}

\begin{definition}[Confusability of joint distributions]\label{def:conf_distr}
A $(P_\bfx,L)$-self-coupling $P_{\bfx_1,\cdots,\bfx_L}\in\cJ^\tl(P_\bfx)$ is said to be \emph{$L$-confusable} if it has an extension $P_{\bfx_1,\cdots,\bfx_L,\bfs_1,\cdots,\bfs_L,\bfy}\in\Delta\paren{\cX^L\times\cS^L\times\cY}$ such that
\begin{enumerate}
    \item $\sqrbrkt{P_{\bfx,\cdots,\bfx_L,\bfs_1,\cdots,\bfs_L,\bfy}}_{\bfx_1,\cdots,\bfx_L} = P_{\bfx_1,\cdots,\bfx_L}$;
    \item $P_{\bfs_i}\in\lambda_\bfs$ for all $i\in[L]$;
    \item $P_{\bfx_i,\bfs_i,\bfy}=P_{\bfx}P_{\bfs_i|\bfx_i}W_{\bfy|\bfx_i,\bfs_i}$ for all $i\in[L]$.
\end{enumerate}
\end{definition}

\begin{definition}[Confusability set]
The  $(P_\bfx,L)$-confusability set $\cK^\tl\paren{P_\bfx}$ of a channel $\cA = \paren{\cX,\lambda_\bfx,\cS,\lambda_\bfs,\cY,W_{\bfy|\bfx,\bfs}}$ is defined as
\begin{align*}
    \cK^\tl\paren{P_\bfx}\coloneqq
    &\curbrkt{P_{\bfx_1,\cdots,\bfx_L}\in\cJ^{\otimes L}(P_\bfx)\colon P_{\bfx_1,\cdots,\bfx_L}\text{ is $L$-confusable}}.
\end{align*}
\end{definition}
\begin{remark}
In the above definitions, we overload the notion of confusability for types and distributions. 
\begin{align*}
    \cK^\tl\paren{P_\bfx} = \bigcup_{n=1}^\infty \curbrkt{\tau_{\vx_1,\cdots,\vx_L}\colon \paren{\vx_1,\cdots,\vx_L}\text{ is $L$-confusable};\; \vx_i\in\cT_\vbfx(P_\bfx),\;\forall i\in[L]}.
\end{align*}
\end{remark}

\begin{definition}[List decodable codes]\label{def:list_dec_codes}
A code $\cC\subset\cX^n$ is said to be \emph{$(L-1)$-list decodable} if no size-$L$ list is confusable, i.e., for any $\cL\in\binom{\cC}{L}$, $\cL$ is non-$L$-confusable.
\end{definition}

\begin{definition}[Achievable rate and list decoding capacity]
A rate $R$ is said to be achievable under $(L-1)$-list decoding if there is an infinite sequence of $(L-1)$-list decodable codes $\curbrkt{\cC_i}_{i\ge1}$ of blocklength $n_i\in\bZ_{>0}$ (such that $\curbrkt{n_i}$ is a non-vanishing sequence) and rate $R(\cC)\ge R$. 

The $(L-1)$-list decoding capacity is defined as the maximal achievable rate.
\[C\coloneqq\limsup_{n\to\infty}\max_{\substack{\cC\subseteq\Lambda_\vbfx\\ \paren{L-1}\text{-list decodable}}}R(\cC).\]
\end{definition}


\section{List decoding capacity}
\label{sec:ld_cap}
\begin{theorem}[List decoding capacity]
\label{thm:ld_cap_thm}
For any adversarial channel $\cA=(\cX,\lambda_\bfx,\cS,\lambda_\bfs,\cY,W)$, let 
\begin{equation}
    C\coloneqq\max_{P_\bfx\in\lambda_\bfx}\min_{P_{\bfs|\bfx}\in\lambda_{\bfs|\bfx}}I(\bfx;\bfy),
    \label{eqn:list_dec_cap}
\end{equation}
which can be viewed as a generalized sphere-packing bound. The mutual information is evaluated w.r.t.
\[P_{\bfx,\bfy}=\sqrbrkt{P_\bfx P_{\bfs|\bfx}W_{\bfy|\bfx,\bfs}}_{\bfx,\bfy}.\]
Then 
\begin{enumerate}
    \item\label{itm:ld_cap_thm_ach} (Achievability) For any $\delta>0$ and sufficiently large $n$, there exists $\cC$ of rate $C-\delta$ such that it can be $\cO(1/\delta)$ list decoded.
    \item\label{itm:ld_cap_thm_conv} (Converse) For any $\cC$ of rate $C+\delta$, $\cC$ is $2^{\Omega(n\delta)}$-list decodable.
\end{enumerate}
\end{theorem}
\begin{proof} We follow the  idea used in the proof of list decoding theorem \ref{thm:ld_cap_bitflip} under the standard bit-flip model but conduct the calculations under our generalized setting \cite{sarwate-thesis}. 
\begin{enumerate}
    \item (Achievability) Let  $R=C-\delta$. Fix $P_\bfx^*\in\lambda_\bfx$ to be a maximizer of expression \eqref{eqn:list_dec_cap}. Generate a random code by sampling $2^{nR}$ codewords independently and uniformly  from $\cT_{\vbfx}(P_\bfx^*)$. 
    We will actually show that 
    \begin{lemma}\label{lem:ld_cap_ach_lem}
    For any $\delta>0$ and sufficiently large $n$,  a random $P_\bfx^*$-constant composition code of rate $R=C-\delta$ as defined above is $\paren{\frac{1+\log\cardY}{\delta} - 1}$-list decodable with probability at least $1-2^{-n(1-R)}$.
    \end{lemma}
    For every $\vy\in\cY^n$, define \emph{conditional typical set} 
    \begin{align*}
        \cA_{\vbfx|\vy}\coloneqq\curbrkt{\vx\in\cT_\vbfx\paren{P_\bfx^*}\colon \exists\vs\in\Lambda_\vbfs,\;\vy = W(\vx,\vs)}
    \end{align*}
    to be the set of all $\vx$ of type $P_\bfx^*$ that can reach $\vy$ via allowable $\vs\in\Lambda_\vbfs$. Note that $\cA_{\vbfx|\vy}$ is precisely the list of codewords around $\vy$ whose size we would like to bound. In favour of proceeding calculations, we write $\cA_{\vbfx|\vy}$ in terms of types and estimate its size. 
    We say that a type $\tau_{\bfx,\bfs,\bfy}\in\cP^{(n)}(\cX\times\cS\times\cY)$ is \emph{valid} if
    \begin{enumerate}
        \item $\sqrbrkt{\tau_{\bfx,\bfs,\bfy}}_\bfx=P_\bfx^*$;
        \item $\sqrbrkt{\tau_{\bfx,\bfs,\bfy}}_{\bfs}\in\lambda_\bfs$;
        \item $\tau_{\bfx,\bfs,\bfy}=P_\bfx^*\tau_{\bfs|\bfx}W_{\bfy|\bfx,\bfs}$.
    \end{enumerate}
    Then it is not hard to see that
    \[\cA_{\vbfx|\vy}=\bigcup_{\tau_{\bfx,\bfs,\bfy}\text{ valid}}\cT_{\vbfx|\vy}\paren{\tau_{\bfx|\vy}},\]
    where 
    $\tau_{\bfx|\vy}$ is obtained from $\tau_{\bfx,\bfs,\bfy}$.
    Note that there is only a polynomial number of types and the volume of each $\cT_{\vbfx|\vy}\paren{\tau_{\bfx|\vy}}$ is dot equal to $2^{nH(\bfx|\bfy)}$, where $H(\bfx|\bfy)$ is evaluated w.r.t. $\sqrbrkt{\tau_{\bfx,\bfs,\vy}}_{\bfx,\vy} = \tau_\vy\tau_{\bfx|\vy}$. 
    Hence the volume of $\cA_{\vbfx|\vy}$ is
    \begin{align}
        \frac{1}{n}\log\card{\cA_{\vbfx|\vy}}\xrightarrow{n\to\infty}&\max_{\tau_{\bfx,\bfs,\bfy}\text{ valid}}H(\bfx|\bfy)\label{eqn:first}\\
        =&\max_{\substack{ P_\bfx^*\tau_{\bfs|\bfx}W_{\bfy|\bfx,\bfs}\colon \\\sqrbrkt{P_\bfx^*\tau_{\bfs|\bfx}W_{\bfy|\bfx,\bfs}}_{\bfs}\in\lambda_\bfs }}H(\bfx|\bfy)\label{eqn:second}\\
        \to&\max_{P_{\bfs|\bfx}\in\lambda_{\bfs|\bfx}}H(\bfx|\bfy)\label{eqn:dense}.
    \end{align}
    In Eqn. \eqref{eqn:first} and \eqref{eqn:second}, the conditional entropy is evaluated w.r.t. $\sqrbrkt{\tau_{\bfx,\bfs,\bfy}}_{\bfx,\bfy}$ and $\sqrbrkt{P_\bfx^*\tau_{\bfs|\bfx}W_{\bfy|\bfx,\bfs}}_{\bfx,\bfy}$, respectively. In Eqn. \eqref{eqn:dense}, the conditional entropy is evaluated w.r.t. $\sqrbrkt{P_\bfx^* P_{\bfs|\bfx}W_{\bfy|\bfx,\bfs}}_{\bfx,\bfy}$. This equality holds in the limit as $n$ approaches infinity since types are asymptotically dense in distributions. 
    Note that $\cA_{\vbfx|\vy}\subset\cT_{\vbfx}(P_\bfx^*)$.
    We have that the probability $q$ that a random codeword $\vbfx$ is able to result in $\vy$ via some admissible $\vs\in\Lambda_\bfs$ is
    \begin{align}
        \frac{1}{n}\log q\coloneqq&\frac{1}{n}\log\prob{\vbfx\in\cA_{\vbfx|\vy}}\notag\\
        =&\frac{1}{n}\log\frac{\card{\cA_{\vbfx|\vy}}}{\card{\cT_{\vbfx}(P_\bfx^*)}}\label{eqn:unif_type}\\
        \xrightarrow{n\to\infty}& \max_{P_{\bfs|\bfx}\in\lambda_{\bfs|\bfx}}H(\bfx|\bfy) - H(\bfx) \label{eqn:size_typ_set}\\
        =&-\max_{P_\bfx\in\lambda_\bfx}\min_{P_{\bfs|\bfx}\in\lambda_{\bfs|\bfx}}I(\bfx;\bfy)\label{eqn:choice_px_star}\\
        =&-C.\notag
    \end{align}
    Eqn. \eqref{eqn:unif_type} follows since codewords are picked uniformly from  $\cT_{\vbfx}(P_\bfx^*)$. Eqn. \eqref{eqn:size_typ_set} is by Eqn. \eqref{eqn:dense} and Eqn. \eqref{lem:size_type_classes}. Eqn. \eqref{eqn:choice_px_star} is by the choice of $P_\bfx^*$.
    The probability that there is a large list clustered around $\vy$ is given by
    \begin{align*}
        \probover{\cC}{\card{\cA_{\vbfx|\vy}\cap\cC}\ge L}\doteq&\sum_{i=L}^{2^{nR}}\binom{2^{nR}}{i}q^i(1-q)^{2^{nR}-i}.
    \end{align*}
    Let $S_i$ denote the summand
    \[S_i\coloneqq\binom{2^{nR}}{i}q^i(1-q)^{2^{nR}-i}.\]
    Note that
    \begin{align}
        \frac{S_i}{S_{i+1}} = & \frac{i+1}{2^{nR}-i}\frac{1-q}{q}\notag\\
        \ge& \frac{2}{2^{n(C-\delta)}}\frac{1-2^{-nC}}{2^{-nC}}\label{eqn:lb_i}\\
        =&{2\cdot\frac{1}{2}}\cdot2^{n\delta}\label{eqn:lb_n}\\
        >&1,\notag
    \end{align}
    where Eqn. \eqref{eqn:lb_i} follows since $i\ge L\ge 1$ and Eqn. \eqref{eqn:lb_n} follows since $1-2^{-nC}\ge\frac{1}{2}$ when $n\ge\frac{1}{C}$.
    The largest summand is the first term. Therefore we can bound the error probability by replacing each term with the first one. 
    \begin{align*}
        \probover{\cC}{\card{\cA_{\vbfx|\vy}\cap\cC}\ge L}\le&2^{nR}\binom{2^{nR}}{L}q^L(1-q)^{2^{nR}-L}\\
        \le&2^{nR}2^{nRL}2^{-nCL}\\
        =&2^{-n((L+1)\delta - C)}.
    \end{align*}
    Finally taking a union bound over all $\vy\in\cY^n$, we know that the probability of list decoding error is at most 
    \begin{align*}
        \prob{\exists\vy\in\cY^n,\;\card{\cA_{\vbfx|\vy}\cap\cC}\ge L}\le&\card{\cY}^n2^{-n((L+1)\delta - C)}\\
        =&2^{-n\paren{(L+1)\delta - C - \log\cardY}},
    \end{align*}
    which is $2^{-\Omega(n)}$ if $L>\frac{1+\log\cardY}{\delta}-1$. Specifically, taking $L=\frac{1+\log\cardY}{\delta}$, we have that the list decoding error probability is at most $2^{-n(1+\delta-C)} = 2^{-n(1-R)}$, as desired.
    \item (Converse) Given any code $\cC$ of rate $C+\delta$, choose the $\tau_{\bfx}^*\in\cP^{(n)}(\cX)$ such that $|\cC\cap\cT_\vbfx(\tau_{\bfx}^*)|$ is maximized. 
    By Lemma \ref{lem:cc_code_suff}, $R(\cC')\asymp R(\cC)$.
    For this $\tau_{\bfx}^*$, choose legitimate $\tau_{\bfs|\bfx}^*\in\lambda_{\bfs|\bfx}$ such that
    \[\tau_{\bfs|\bfx}^*\coloneqq\argmin{ \tau_{\bfs|\bfx}\in\lambda_{\bfs|\bfx} }I(\bfx;\bfy),\]
    where
    $I(\bfx;\bfy)$ is evaluated according to $\sqrbrkt{\tau_\bfx^*\tau_{\bfs|\bfx}W_{\bfy|\bfx,\bfs}}_{\bfx,\bfy}$. 
    Now define $\tau_{\bfx,\bfs,\bfy}^*\coloneqq\tau_\bfx^*\tau_{\bfs|\bfx}^*W_{\bfy|\bfx,\bfs}$,  $\tau_{\bfx,\bfy}^*\coloneqq\sqrbrkt{\tau_{\bfx,\bfs,\bfy}^*}_{\bfx,\bfy}$ and  $\tau_{\bfy}^*\coloneqq\sqrbrkt{\tau^*_{\bfx,\bfy}}_\bfy$. Over the randomness of selecting  $\vbfy$  uniformly from $\cT_\vbfy\paren{\tau_{\bfy}^*}$, the average number of codewords in $\cA_{\vbfx|\vbfy}$ is dot equal to
    \begin{align}
        \exptover{\vbfy}{\card{{\cA}_{\vbfx|\vbfy}\cap\cC'}} =& \exptover{\vbfy}{\sum_{\vx\in\cC'}\indicator{{\cA}_{\vbfx|\vbfy}\ni\vx}}\notag\\
        =&\sum_{\vx\in\cC'}\probover{\vbfy}{{\cA}_{\vbfx|\vbfy}\ni\vx}\label{eqn:lin_exp_typ_set}\\
        =&\sum_{\vx\in\cC'}\probover{\vbfy}{\cT_{\vbfx|\vbfy}\paren{\tau_{\bfx|\bfy}^*}\ni\vx}\label{eqn:degenerate_typ_set}\\
        =&\sum_{\vx\in\cC'}\probover{\vbfy}{ \tau_{\vx|\vbfy} = \tau_{\bfx|\bfy}^* }\label{eqn:typ_y_given_x}\\
        =&\sum_{\vx\in\cC'}\frac{1}{\card{\cT_\vbfy\paren{\tau^*_\bfy}}}\prod_{x\in\cX}\binom{\tau_\bfx^*(x)n}{\tau_\bfy^*(1)n\cdot\tau_{\bfx|\bfy}^*(x|1),\cdots,\tau_\bfy^*(\cardY)n\cdot\tau_{\bfx|\bfy}^*(x|\cardY)}.\label{eqn:compute_typ_y_given_x}
    \end{align}
    Eqn. \eqref{eqn:lin_exp_typ_set} is linearity of expectation.  Note that by our choice of $\tau_\bfx^*$ and $\tau_{\bfs|\bfx}^*$ (hence $\tau_{\bfx,\bfs,\bfy}^*$ and $\tau_{\bfx,\bfy}^*$), $\cA_{\vbfx|\vy}$ only contains one type class $\cT_{\vbfx|\vy}\paren{\tau_{\bfx|\bfy}^*}$, where $\tau_{\bfx|\bfy}^*$ is computed from $\tau_{\bfx,\bfy}^*$. Eqn. \eqref{eqn:degenerate_typ_set}  then follows. Eqn. \eqref{eqn:typ_y_given_x} follows from the definition of type classes (Definition \ref{def:def_type_classes}). Eqn. \eqref{eqn:compute_typ_y_given_x} is by analyzing the sampling procedure from the first principle. The product is exactly, given $\vx\in\cC'$, the number of ways to pick $\vy$ from $\cT_{\vbfx|\vy}\paren{\tau_\bfy^*}$ such that $\tau_{\vx|\vy} = \tau_{\bfx|\bfy}^*$. 
    We compute the exponent of the above expectation.
    \begin{align}
        \frac{1}{n}\log\exptover{\vbfy}{\card{{\cA}_{\vbfx|\vbfy}\cap\cC'}}\xrightarrow{n\to\infty}&R'-H\paren{\tau_\bfy^*} + \sum_{x\in\cX}\tau_\bfx^*(x)\sum_{y\in\cY}\frac{\tau_\bfy^*(y)\tau_{\bfx|\bfy}^*(x|y)}{\tau_\bfx^*(x)}\log\frac{\tau_\bfx^*(x)}{\tau_\bfy^*(y)\tau_{\bfx|\bfy}^*(x|y)}\label{eqn:indep_of_x}\\
        =&R-H\paren{\tau_\bfy^*}+\sum_{x\in\cX}\tau_\bfx^*(x)H(\bfy|\bfx=x)\label{eqn:def_cond_ent}\\
        =&R-H(\bfy)+H(\bfy|\bfx)\label{eqn:type_to_distr}\\
        =&R-I(\bfx;\bfy)\notag\\
        \ge&R-C\label{eqn:compare_i_c}\\
        =&\delta.\notag
    \end{align}
    Since codewords in the subcode $\cC'$ are $\tau_\bfx^*$-constant composition, the summand in Eqn. \eqref{eqn:compute_typ_y_given_x} is independent of particular choices of $\vx$. Eqn. \eqref{eqn:indep_of_x} then follows from Stirling's approximation (Lemma \ref{lem:stirling}). In Eqn. \eqref{eqn:def_cond_ent}, $H(\bfy|\bfx = x)$ is drawn according to the conditional type \[\tau_{\bfy|\bfx}^*(\cdot|x) = \frac{\tau_\bfy^*(\cdot)\tau_{\bfx|\bfy}^*(x|\cdot)}{\tau_\bfx^*(x)}.\] 
    In Eqn. \eqref{eqn:type_to_distr}, we pass types to distributions by the fact that types are dense in distributions asymptotically in $n$. $H(\bfy)$ and $H(\bfy|\bfx)$ are evaluated using distribution $\sqrbrkt{\tau_\bfx^*P_{\bfs|\bfx}^*W_{\bfy|\bfx,\bfs}}_{\bfx,\bfy}$, where
    \[P_{\bfs|\bfx}^*\coloneqq\argmin{P_{\bfs|\bfx}\in\lambda_{\bfs|\bfx}}I(\bfx;\bfy),\]
    and the objective function $I(\bfx;\bfy)$ is evaluated using $\sqrbrkt{\tau_\bfx^* P_{\bfs|\bfx}W_{\bfy|\bfx}}_{\bfx,\bfy}$. Eqn. \eqref{eqn:compare_i_c} is  by the definition of $C$ (Eqn. \eqref{eqn:list_dec_cap}). $\tau_\bfx^*$ always gives rise to mutual information no larger than the maximizer in $C$.
    
    Therefore, we have shown that there exists at least one $\vy\in\cY^n$ such that the corresponding list around $\vy$ has size at least $2^{n(\delta-o(1))}$.
\end{enumerate}
\end{proof}

\section{List sizes of random codes}
\label{sec:ld_rand}
In this section, we show that, if $L$ has order lower than $1/\delta$, then the code used in the proof of achievability (part \ref{itm:ld_cap_thm_ach}) of the list decoding capacity theorem (Theorem \ref{thm:ld_cap_thm}) is list decodable with vanishingly small probability. 
This coupled with Theorem \ref{thm:ld_cap_thm}  implies that, for the majority (an exponentially close to 1 fraction) of random  constant composition capacity-achieving (within gap $\delta$) codes, $\Theta(1/\delta)$ is actually the \emph{correct} order of their list sizes. 
\begin{corollary}
For $\delta>0$ and sufficiently large $n$,  at least a $1-2^{-n(1-R)}-2^{-n\delta+\frac{2}{\delta}\log\frac{1}{\delta}}$ fraction of $P_\bfx^*$-constant composition codes ($P_\bfx^*$ as defined in Eqn. \eqref{eqn:def_px_star}) of rate $R=C-\delta$ is $(L-1)$-list decodable, where $L=\Theta\paren{1/\delta}$ lies within the following range
\[L\in\sqrbrkt{\frac{C}{\delta},\frac{1+\log\cardY}{\delta}}.\]
\end{corollary}

\begin{theorem}
For an adversarial channel $\cA = \paren{\cX,\lambda_\bfx,\cS,\lambda_\bfs,,\cY,W_{\bfy|\bfx,\bfs}}$, take  an optimizing input distribution $P_\bfx$ which attains the list decoding capacity $C$,
\begin{equation}
    P_{\bfx}^*\coloneqq\argmax{P_\bfx\in\lambda_\bfx}\min_{P_{\bfs|\bfx}\in\lambda_{\bfs|\bfx}}I(\bfx;\bfy).
    \label{eqn:def_px_star}
\end{equation}
For any $\delta>0$, for each sufficiently large blocklength $n$, sample a random code $\cC$ of rate $R = C-\delta$  whose codewords are selected independently and uniformly from $\cT_\vbfx\paren{P_\bfx^*}$. Then $\cC$ is $<\paren{C/\delta - 1}$-list decodable with probability at most $2^{-n\delta+\frac{2}{\delta}\log\frac{1}{\delta}}$. 
\end{theorem}

The theorem follows from second moment calculations and generalizes similar theorems for list decodability of random error/erasure correction codes over $\bF_q$ \cite{guruswami2013combinatorial}.

\begin{proof}
Let $M\coloneqq2^{nR}$. Define \emph{typical set}
\[\cA_\vbfy\coloneqq\curbrkt{W\paren{\vx,\vs}\in\cY^n\colon \vx\in\cT_\vbfx\paren{P_\bfx^*},\;\vs\in\Lambda_\vbfs}.\]
Put in the language of types, it can also be written as
\[\cA_\vbfy = \bigcup_{\tau_{\bfx,\bfs,\bfy}\text{ valid}}\cT_\vbfy\paren{\tau_\bfy},\]
where $\tau_\bfy = \sqrbrkt{\tau_{\bfx,\bfs,\bfy}}_\bfy$.
Define random variable $W$ as a witness for non-list decodability of $\cC$
\[W\coloneqq\sum_{\vy\in\cA_{\vbfy}}\sum_{\curbrkt{m_1,\cdots,m_L}\in\binom{[M]}{L}}\indicator{\curbrkt{\vbfx_{m_1},\cdots,\vbfx_{m_L}}\subset\cA_{\vbfx|\vy}}.\]
Then by Chebyshev's inequality,
\begin{align}
    \prob{\cC\text{ is }(L-1)\text{-list decodable}}=&\prob{\bigcap_{\vy\in\cY^n}\curbrkt{ \card{\cA_{\vbfx|\vy}\cap\cC}<L }}\label{eqn:prob_ld_to_bd}\\
    \le&\prob{\bigcap_{\vy\in\cA_{\vbfy}(P_\bfy)}\curbrkt{ \card{\cA_{\vbfx|\vy}\cap\cC}<L }}\notag\\
    =&\prob{\paren{ \bigcup_{\vy\in\cA_\vbfy}\curbrkt{ \card{\cA_{\vbfx|\vy}\cap\cC}\ge L } }^c}\notag\\
    =&\prob{W=0}\label{eqn:w_zero_iff}\\
    \le&\frac{\var{W}}{\expt{W}^2},\notag
\end{align}
where Eqn. \eqref{eqn:w_zero_iff} follows since $W=0$ if and only if none of the events $\curbrkt{\card{\cA_{\vbfx|\vy}\cap\cC}\ge L}$ ($\vy\in\cA_\vbfy$) happens. 
In what follows, we will  obtain an upper bound on $\var{W}$ and a lower bound on $\expt{W}$, and hence an upper bound on the probability \eqref{eqn:prob_ld_to_bd}.

\noindent\textbf{Lower bounding $\expt{W}$.} We can get a lower bound on the expected value of $W$ from a straightforward calculation.
\begin{align}
    \expt{W}=&\sum_{\vy\in\cA_{\vbfy}}\sum_{\curbrkt{m_1,\cdots,m_L}\in\binom{[M]}{L}}\prob{\curbrkt{\vbfx_{m_1},\cdots,\vbfx_{m_L}}\subset\cA_{\vbfx|\vy}}\notag\\
    =&\sum_{\vy\in\cA_{\vbfy}}\sum_{\curbrkt{m_1,\cdots,m_L}\in\binom{[M]}{L}}\prob{\vbfx\in\cA_{\vbfx|\vy}}^L\label{eqn:indep_factor}\\
    \doteq&\card{\cA_{\vbfy}}\binom{M}{L}2^{-nCL}\label{eqn:prob_x_in_cond_typ_set}\\
    \ge&\card{\cA_{\vbfy}}\paren{\frac{M}{L}}^L2^{-nCL}\notag\\
    =&\card{\cA_{\vbfy}}2^{-n\delta L - L\log L}.\notag
\end{align}
Eqn. \eqref{eqn:indep_factor} follows since codewords are independent. Eqn. \eqref{eqn:prob_x_in_cond_typ_set} is by Eqn. \eqref{eqn:choice_px_star}. 

\noindent\textbf{Upper bounding $\var{W}$.}
Define, for any $\vy\in\cY^n$ and $\cL\in\binom{[M]}{L}$,
\begin{align*}
    \bI\paren{\vy,\cL}\coloneqq& \indicator{ \curbrkt{\vbfx_m}_{m\in\cL}\subset\cA_{\vbfx|\vy} } \\
    =&\prod_{m\in\cL}\indicator{\vbfx_m\in\cA_{\vbfx|\vy}},
\end{align*}
as the indicator function of the event $\bigcap_{m\in\cL}\curbrkt{\vbfx_m\in\cA_{\vbfx|\vy}}$ that the list $\cL$ is $L$-confusable w.r.t. $\vy$.

Now the variance of $W$ can be upper bounded as follows.
\begin{align}
    \var{W}=&\expt{W^2} - \expt{W}^2\label{eqn:var_def}\\
    =&\sum_{\vy_1,\vy_2\in\cA_{\vbfy}}\sum_{\cL_1,\cL_2\in\binom{[M]}{L}}\expt{\bI\paren{\vy_1,\cL_1}\bI\paren{\vy_2,\cL_2}}-\expt{\bI\paren{\vy_1,\cL_1}}\expt{\bI\paren{\vy_2,\cL_2}}\label{eqn:exp_lin}\\
    \le&\sum_{\substack{\cL_1,\cL_2\in\binom{[M]}{L}\\\cL_1\cap\cL_2\ne\emptyset}}\sum_{\vy_1,\vy_2\in\cA_\vbfy}\expt{\bI\paren{\vy_1,\cL_1}\bI\paren{\vy_2,\cL_2}}\label{eqn:ub_drop}\\
    =&\card{\cA_{\vbfy}}^{2}\sum_{\ell=1}^L\sum_{\card{\cL_1\cap\cL_2}=\ell}\probover{\vbfy_1,\vbfy_2,\cC}{\cE}.\label{eqn:rand_centers}
\end{align}
Eqn. \eqref{eqn:var_def} follows from the definition of variance and Eqn. \eqref{eqn:exp_lin} follows from linearity of expectation. 
 Note that $\bI\paren{\vy_1,\cL_1}$ and $\vI\paren{\vy_2,\cL_2}$ are independent if and only if $\cL_1\cap\cL_2 = \emptyset$. When they are independent, the first expectation factors and the summand vanishes. The inequality \eqref{eqn:ub_drop} follows by dropping the negative term in the summand. In Eqn. \eqref{eqn:rand_centers}, we rewrite the summation by randomizing the centers $\vy_1,\vy_2$ of the lists $\cL_1,\cL_2$. The probability is taken over $\vbfy_1$ and $\vbfy_2$ chosen uniformly at random from $\cA_{\vbfy}$ and over the random code sampling procedure. 
 We use $\cE$ to denote the event that the lists $\cL_1$ and $\cL_2$ are simultaneously $L$-confusable w.r.t. $\vbfy_1$ and $\vbfy_2$, respectively,
\[\cE\coloneqq\bigcap_{m_1\in\cL_1}\curbrkt{\vbfx_{m_1}\in\cA_{\vbfx|\vbfy_1}}\cap\bigcap_{m_2\in\cL_2}\curbrkt{\vbfx_{m_2}\in\cA_{\vbfx|\vbfy_2}}.\]

It then suffices to bound $\prob{\cE}$. To this end, first define conditional typical set, for $\vx\in\cX^n$,
\begin{align*}
    \cA_{\vbfy|\vx}\coloneqq & \curbrkt{ W\paren{\vx,\vs}\in\cY^n\colon \vs\in\Lambda_\vbfs }\\
    =&\bigcup_{\tau_{\bfx,\bfs,\bfy}\text{ valid}}\cT_\vbfy\paren{\tau_{\bfy|\vx}},
\end{align*}
where $\tau_{\bfy|\vx}$ is computed from $\tau_{\bfx,\bfs,\bfy}$ and $\tau_\vx$, $\tau_{\bfy|\vx} = \sqrbrkt{\tau_{\bfx,\bfs,\bfy}}_{\bfx,\bfy}/\tau_\vx$.
Then define the following events in favour of bounding $\prob{\cE}$. 
\begin{align*}
    \cE_1\coloneqq&\curbrkt{\vbfy_1\in\cA_{\vbfy|\vbfx_m}}\cap\curbrkt{\vbfy_2\in\cA_{\vbfy|\vbfx_m}} ,\\
    \cE_2\coloneqq&\bigcap_{m_1\in\cL_1\setminus\curbrkt{m}}\curbrkt{\vbfx_{m_1}\in\cA_{\vbfx|\vbfy_1}} ,\\
    \cE_3\coloneqq&\bigcap_{m_2\in\cL_2\setminus\cL_1}\curbrkt{\vbfx_{m_2}\in\cA_{\vbfx|\vbfy_2}} ,
\end{align*}
where $m\in\cL_1\cap\cL_2$ is any message that appears in both $\cL_1$ and $\cL_2$.
It is easy to verify  that $\cE\subset\cE_1\cap\cE_2\cap\cE_3$ (see Fig. \ref{fig:lb}). 
\begin{figure}
    \centering
    \includegraphics{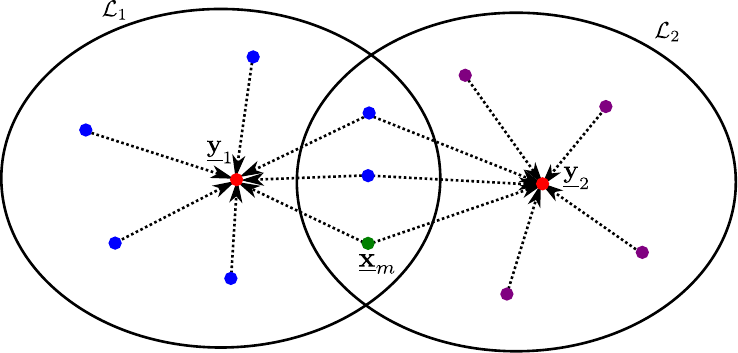}
    \caption{$\cE\subset\cE_1\cap\cE_2\cap\cE_3$. We upper bound $\prob{\cE}$ by neglecting the fact that codewords $\vbfx_i$ for $i\in\paren{\cL_1\cap\cL_2}\setminus\curbrkt{m}$ are simultaneously $\vbfy_1$-confusable and $\vbfy_2$-confusable, or equivalently, neglecting that $\vbfy_1,\vbfy_2$ should simultaneously belong to $\cA_{\vbfy|\vbfx_{m'}}$ for all $m'\in\cL_1\cap\cL_2$, not only the particular $m$ we have chosen.}
    \label{fig:lb}
\end{figure}
Note that $\cE_2$ and $\cE_3$ are independent conditioned on $\cE_1$ since $\cL_1\setminus\curbrkt{m}$ and $\cL_2\setminus\cL_1$ are disjoint. The probabilities of the above events  can be computed precisely. 
\begin{align}
    \prob{\cE_1} = &\prob{\vbfy\in\cA_{\vbfy|\vbfx_m}}^2\label{eqn:prob_y_factors}\\
    =&\paren{\frac{\card{\cA_{\vbfy|\vbfx_m}}}{\card{\cA_\vbfy}}}^2,\label{eqn:prob_y_vol_ratio}
\end{align}
where Eqn. \eqref{eqn:prob_y_factors} is because $\vbfy_1$ and $\vbfy_2$ are independent, and Eqn. \eqref{eqn:prob_y_vol_ratio} follows since $\vbfy$ is chosen uniformly from $\cA_\vbfy$. We now compute the exponent of $\prob{\cE}$.
\begin{align}
    \frac{1}{n}\log\card{\cA_\vbfy}\xrightarrow{n\to\infty}&\max_{\tau_{\bfx,\bfs,\bfy}\text{ valid}}H(\bfy)\label{eqn:eval_tau_y}\\
    =&\max_{P_{\bfs|\bfx}\in\lambda_{\bfs|\bfx}}H(\bfy)\label{eqn:vol_typ_y_similar},
\end{align}
where in  Eqn. \eqref{eqn:eval_tau_y} the entropy is computed w.r.t. $\tau_\bfy = \sqrbrkt{\tau_{\bfx,\bfs,\bfy}}_\bfy$; Eqn. \eqref{eqn:vol_typ_y_similar} follows from similar calculations as done for $\cA_{\bfx|\vy}$ (Eqn. \eqref{eqn:first}) and the entropy is evaluated using $\sqrbrkt{P_\bfx^* P_{\bfs|\bfx}W_{\bfy|\bfx,\bfs}}_\bfy$. 

Similarly, 
\begin{align}
    \frac{1}{n}\log\card{\cA_{\vbfy|\vbfx_m}}\xrightarrow{n\to\infty}&\max_{\tau_{\bfx,\bfs,\bfy}\text{ valid}}H(\bfy|\bfx)\label{eqn:ent_y_cond_x_eval_type}\\
    =&\max_{P_{\bfs|\bfx}\in\lambda_{\bfs|\bfx}}H(\bfy|\bfx),\label{eqn:ent_y_cond_x_eval_distr}
\end{align}
where the conditional entropies in Eqn. \eqref{eqn:ent_y_cond_x_eval_type} and \eqref{eqn:ent_y_cond_x_eval_distr} are evaluated w.r.t. $\tau_\vx\tau_{\bfy|\vx}$ and $\sqrbrkt{P_\bfx^* P_{\bfs|\bfx}W_{\bfy|\bfx,\bfs}}_{\bfx,\bfy}$ (since $\tau_\vx\to P_\bfx^*$ as $n$ approaches infinity), respectively. Continuing with Eqn. \eqref{eqn:prob_y_vol_ratio},  putting Eqn. \eqref{eqn:vol_typ_y_similar} and Eqn. \eqref{eqn:ent_y_cond_x_eval_distr} together, we have
\begin{align}
    \prob{\cE_1} \doteq& \paren{2^{n\max_{P_{\bfs|\bfx}\in\lambda_{\bfs|\bfx}}H(\bfy|\bfx) - H(\bfy)}}^2\notag\\
    =&2^{ -2n\min_{P_{\bfs|\bfx}\in\lambda_{\bfs|\bfx}}I(\bfx;\bfy) }\notag\\
    =&2^{-2nC},\label{eqn:choice_px_star_eone}
\end{align}
where Eqn. \eqref{eqn:choice_px_star_eone} is by the choice of $P_\bfx^*$ (Eqn. \eqref{eqn:def_px_star}).

We also have
\begin{align}
    \prob{\cE_2|\cE_1}=&\prob{\left.\vbfx\in\cA_{\vbfx|\vbfy_1}\right|\cE_1}^{L-1}\doteq2^{-nC(L-1)},\label{eqn:size_l1}\\
    \prob{\cE_3|\cE_1}=&\prob{\left.\vbfx\in\cA_{\vbfx|\vbfy_1}\right|\cE_1}^{L-\ell}\doteq2^{-nC(L-\ell)},\label{eqn:size_intersection}
\end{align}
where Eqn. \eqref{eqn:size_l1} and Eqn. \eqref{eqn:size_intersection} follow since $\card{\cL_1} = \card{\cL_2} = L$ and $\card{\cL_1\cap\cL_2}=\ell$.
We thus have, from Eqn. \eqref{eqn:choice_px_star_eone}, \eqref{eqn:size_l1} and \eqref{eqn:size_intersection}, that
\begin{align}
    \prob{\cE}\le&\prob{\cE_1\cap\cE_2\cap\cE_3}\notag\\
    =&\prob{\cE_1}\prob{\cE_2|\cE_1}\prob{\cE_3|\cE_1}\notag\\
    \doteq&2^{-nC(2L-\ell+1)}.\label{eqn:prob_e}
\end{align}
Note that the number of pairs of lists $\cL_1$ and $\cL_2$ with intersection size $\ell$ is 
\begin{align}
    \binom{M}{\ell}\binom{M-\ell}{L-\ell}\binom{M-\ell}{L-\ell}\le& M^\ell M^{L-\ell} M^{L-\ell}\notag\\
    \le &M^{2L-\ell}.\label{eqn:number_of_lists}
\end{align}

Therefore, the variance of $W$ can be bounded as follows.
\begin{align}
    \var{W}\le&\card{\cA_\vbfy}^2\sum_{1\le\ell\le L}M^{2L-\ell}2^{-nC(2L-\ell+1)}\label{eqn:put_three_together}\\
    =&\card{\cA_\vbfy}^22^{ - nC}\sum_{1\le\ell\le L}2^{-n\delta(2L-\ell)}\label{eqn:by_def_m_choice_r}\\
    \le&\card{\cA_\vbfy}^22^{ - nC}2^{-n\delta(2L-\ell)+\log L},\label{eqn:ub_dom_term}
\end{align}
where Eqn. \eqref{eqn:put_three_together} is by Eqn. \eqref{eqn:rand_centers}, \eqref{eqn:number_of_lists} and \eqref{eqn:prob_e}; Eqn. \eqref{eqn:by_def_m_choice_r} is by the definition of $M$ and the choice of $R$; Eqn. \eqref{eqn:ub_dom_term} is by replacing each term with the largest one in the summation. 

\noindent\textbf{Putting them together.} 
\begin{align*}
    \prob{\cC\text{ is }(L-1)\text{-list decodable}}\le&\frac{\var{W}}{\expt{W}^2}\\
    \le&2^{-nC  + n\delta L + (2L+1)\log L}.
\end{align*}
The above probability vanishes in $n$  if $L<C/\delta$. Say $L=C/\delta - 1$, then it is at most 
\[2^{-n\delta+(2(C/\delta-1)+1)\log (C/\delta-1)}\le2^{-n\delta+\frac{2}{\delta}\log\frac{1}{\delta}}.\]
\end{proof}

\section{Achievability}
\label{sec:achievability}
In this section, we are going to show, via concrete random code constructions, that as long as some completely positive $(P_\bfx,L)$-self-coupling  of order $L$ lies outside the order-$L$ confusability set of the channel, the $(L-1)$-list decoding capacity is positive.

Let $\cp_{\card{\cX}}^\tl (P_\bfx)\coloneqq\cp_{\card{\cX}}^\tl\cap\cJ^\tl\paren{P_\bfx}$.
\begin{theorem}[Achievability]
\label{thm:achievability}
For any given general adversarial channel $\cA = (\cX,\lambda_\bfx,\cS,\lambda_\bfs,\cY,W_{\bfy|\bfx,\bfs})$, its $(L-1)$-list decoding capacity is positive if there is a completely positive $(P_\bfx,L)$-self-coupling $P_{\bfx_1,\cdots,\bfx_L}\cp_\cardX^\tl(P_\bfx)$ outside $\cK^\tl(P_\bfx)$ for some $P_\bfx\in\lambda_\bfx$.
\end{theorem}

We first state a lemma concerning the rate of a random constant composition code.
\begin{lemma}[Constant composition codes]
\label{lem:cc_conc}
Let $\cC= \curbrkt{\vbfx_i}_{i=1}^{2^{nR}}$  be a random code of rate $R$ in which each codeword is selected according to product distribution $P_\bfx^\tn$ independently. Let $\cC'$ be the $P_\bfx$-constant composition subcode of $\cC$, $\cC'=\cC\cap\cT_{\vbfx}(P_\bfx)$. Then 
\[\prob{ \card{\cC'}\notin(1\pm1/2)\frac{2^{nR}}{\nu(n)} }\le 2\exp\paren{ -\frac{2^{nR}}{12\nu(n)} }.\]
\end{lemma}
\begin{proof}
The lemma is a simple consequence of concentration of measure (Lemma \ref{lem:chernoff}). 
\begin{align}
	\prob{\cardCp\notin (1\pm1/2)\frac{2^{nR}}{\nu(n)}} = &\prob{ \sum_{i = 1}^{2^{nR}}\indicator{ \tau_{\vbfx_i} = P_\bfx } \notin (1\pm1/2)\frac{2^{nR}}{\nu(n)} }\notag\\
	\le&2\exp\paren{ -\frac{(1/2)^2}{3}\mu }\label{eqn:expt}\\
	=&2\exp\paren{ -\frac{2^{nR}}{12\nu(n)} }.\notag
\end{align}
where in Eqn. \eqref{eqn:expt}, we note that
\begin{align*}
	\expt{ \sum_{i = 1}^{2^{nR}}\indicator{ \tau_{\vbfx_i} = P_\bfx }  } =& 2^{nR}\prob{\vbfx\in\cT_\vbfx(P_\bfx)}\\
	=&\frac{2^{nR}}{\nu(n)}\\
	\eqqcolon&\mu.
\end{align*}
\end{proof}

\subsection{Low rate codes}
Let us proceed gently. We first show that a purely random code with each entry i.i.d. w.r.t. some distribution $P_\bfx$ is $(L-1)$-list decodable w.h.p. as long as $P_\bfx^\tl$ is not $L$-confusable. 
\begin{lemma}\label{lem:low_rate}
For any general adversarial channel $\cA=(\cX,\lambda_\bfx,\cS,\lambda_\bfs,\cY,W_{\bfy|\bfx,\bfs})$,
if there exists a legitimate input distribution $P_\bfx\in\lambda_\bfx$ such that $P_\bfx^\tl\notin\cK^\tl\paren{P_\bfx}$, then the $(L-1)$-list decoding capacity of $\cA$ is positive.
\end{lemma}

\begin{proof}
Let $M=2^{nR}$ for some rate $R$ to be specified momentarily. Sample a code $\cC=\curbrkt{\vbfx_1,\cdots,\vbfx_M}$ where each $\vbfx_i\iid P_\bfx^\tn$. The expected joint type $\tau_{\vbfx_{i_1},\cdots,\vbfx_{i_L}}$ ($1\le i_1<\cdots<i_L\le M$) of any list $\vbfx_{i_1},\cdots,\vbfx_{i_L}$ is $P_\bfx^\tl$. (See Fig. \ref{fig:low_rate_code_from_prod_distr}.) 
\begin{figure}
    \centering
    \includegraphics[scale = 2]{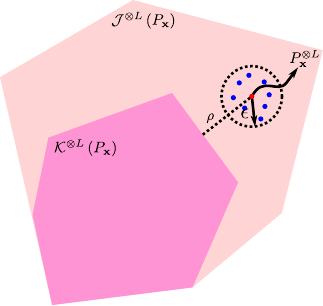}
    \caption{Low rate codes from product distribution. If the product distribution $P_\bfx^\tl$ is strictly separated away from $\cK^\tl(P_\bfx)$, then we could hope for a positive rate achieved by a random code with each entry sampled from $P_\bfx$. This is because w.h.p. the joint types of all (ordered) lists are contained in a $\normmav{\cdot}$-ball which is completely outside the confusability set.}
    \label{fig:low_rate_code_from_prod_distr}
\end{figure}

Let $\cC' = \cC\cap\cT_\vbfx(P_\bfx)$ be the $P_\bfx$-constant composition subcode of $\cC$.
Let 
\[\rho\coloneqq\inf_{P_{\bfx_1,\cdots,\bfx_L}\in\cK^\tl(P_\bfx)}\normmav{P_\bfx^\tl - P_{\bfx_1,\cdots,\bfx_L}}\]
be the max-absolute-value tensor distance from the product distribution to the confusability set. 
Let $R=\frac{\log e}{12}\frac{\rho^2}{L}-\delta$ for some small constant $\delta>0$.
We will show that 
\begin{lemma}
The random $P_\bfx$-constant composition code $\cC'$ as constructed above has rate $R=\frac{\log e}{12}\frac{\rho^2}{L}-\delta$  and is $(L-1)$-list decodable with probability at least $1-2\exp\paren{-2^{nR}/\nu(n)}-2^{-n\delta+L\log\cardX+1}$.
\end{lemma}

Let $\eps\coloneqq\rho/2$. 
Define error events
\begin{align*}
	\cE_1\coloneqq& \curbrkt{\cardCp\notin(1\pm1/2)\frac{2^{nR}}{\nu(n)}},\\
	\cE_2\coloneqq&\curbrkt{\cC'\text{ is not }(L-1)\text{-list decodable}}.
\end{align*}

By Lemma \ref{lem:cc_conc}, 
\[\prob{\cE_1}\le2\exp\paren{-\frac{2^{nR}}{\nu(n)}}.\]
Hence the rate $R'$ of $\cC'$ is asymptotically equal to $R$  w.h.p.

By Chernoff bound,
\begin{align}
    &\prob{ \normmav{\tau_{\vbfx_{i_1},\cdots,\vbfx_{i_L}} - P_{\bfx}^\tl}\ge\eps }\notag\\
    =&\prob{\exists \paren{x_1,\cdots,x_L}\in\cX^L,\;\abs{\tau_{\vbfx_{i_1},\cdots,\vbfx_{i_L}}\paren{x_1,\cdots,x_L}-P_\bfx\paren{x_1}\cdots P_\bfx\paren{x_L}}\ge\eps}\label{eqn:def_mavnorm}\\
    \le&\cardX^L\prob{ \abs{\sum_{j=1}^n\indicator{\paren{\vbfx_{i_1}(j),\cdots,\vbfx_{i_L}(j)} = (x_1,\cdots,x_L)}-nP_\bfx(x_1)\cdots P_\bfx(x_L)}\ge n\eps }\label{eqn:def_type}\\
    =&\cardX^L\prob{ \sum_{j=1}^n\indicator{\paren{\vbfx_{i_1}(j),\cdots,\vbfx_{i_L}(j)} = (x_1,\cdots,x_L)}\notin\paren{ 1\pm\frac{n\eps}{\mu} }\mu }\label{eqn:def_mu}\\
    \le&\cardX^L\cdot2\exp\paren{ -\frac{1}{3}\paren{\frac{n\eps}{\mu}}^2\mu }\label{eqn:apply_chernoff}\\
    =&\cardX^L\cdot2\exp\paren{ -\frac{n\eps^2}{3P_\bfx^\tl(x_1,\cdots,x_L)} }\label{eqn:sub_mu}\\
    \le&\cardX^L\cdot2\exp\paren{-\frac{n}{3}\paren{\frac{\rho}{2}}^2}\label{eqn:choice_eps}\\
    =&2\cdot\cardX^L\cdot\exp\paren{-\frac{\rho^2}{12}n}.\notag
\end{align}
Eqn. \eqref{eqn:def_mavnorm} follows from the definition of max-absolute-value norm. Eqn. \eqref{eqn:def_type} is obtained by taking a union bound and expanding the type using definition. In Eqn. \eqref{eqn:def_mu}, we define 
\[\mu\coloneqq nP_\bfx^\tl(x_1,\cdots,x_L),\]
which equals
\[\expt{ \sum_{j=1}^n\indicator{\paren{\vbfx_{i_1}(j),\cdots,\vbfx_{i_L}(j)} = (x_1,\cdots,x_L)} }.\]
Eqn. \eqref{eqn:apply_chernoff} is by Chernoff bound (Lemma \ref{lem:chernoff}). Eqn. \eqref{eqn:sub_mu} is by the definition of $\mu$.   Eqn. \eqref{eqn:choice_eps} is by the choice of $\eps$ and that $P_\bfx^\tl(x_1,\cdots,x_L)\le 1$ for any $(x_1,\cdots,x_L)\in\cX^L$.
Taking a union bound over all lists $\paren{i_1,\cdots,i_L}\in\binom{\cM}{L}$, 
\begin{align*}
    &\prob{\exists\paren{i_1,\cdots,i_L}\in\binom{\cM}{L},\;\norminf{\tau_{\vbfx_{i_1},\cdots,\vbfx_{i_L}} - P_{\bfx}^\tl}\ge\eps}\\
    \le&\binom{M}{L}2\cdot\cardX^L\cdot\exp\paren{-\frac{\rho^2}{12}n}\\
    \le&2^{-n\paren{ \frac{\rho^2\log e}{12} - RL } + L\log\cardX+1}.
\end{align*}
We therefore get that $\cC$ is $(L-1)$-list decodable with probability at least 
$1-2^{-n\delta+L\log\cardX+1}$
as long as 
\[R=\frac{\log e}{12}\frac{\rho^2}{L}-\delta.\]

Overall, we have that
\begin{align*}
	\prob{\cE_1\cup\cE_2}\le&\prob{\cE_1}+\prob{\cE_2}\\
	\le&2\exp\paren{-\frac{2^{nR}}{\nu(n)}} + \prob{\cC\text{ is not }(L-1)\text{-list decodable}}\\
	\le&2\exp\paren{-\frac{2^{nR}}{\nu(n)}} + 2^{-n\delta+L\log\cardX+1}.
\end{align*}
\end{proof}

\subsection{Random codes with expurgation}
In the previous section, we only got an $(L-1)$-list decodable code of positive rate without making the effort to optimize the rate. In this section, we provide a lower bound on the  $(L-1)$-list decoding capacity. It is achieved  by  a different  code construction (random code with expurgation). However, we can only show  the \emph{existence} of such codes instead of showing that they attain the following bound w.h.p.
\begin{lemma}
The $(L-1)$-list decoding capacity of a channel $\cA$ is at least
\begin{align}
    C_{L-1}\ge&\max_{P_\bfx\in\lambda_\bfx}\min_{P_{\bfx_1,\cdots,\bfx_L}\in\cK^\tl\paren{P_\bfx}}\frac{1}{L-1}D\paren{P_{\bfx_1,\cdots,\bfx_L}\|P_\bfx^\tl}.
    \label{eqn:expurgation_bound}
\end{align}
\end{lemma}

\begin{proof}
Fix any $P_\bfx\in\lambda_\bfx$ to be the maximizer of Eqn. \eqref{eqn:expurgation_bound}. Let $M=2^{nR}$ for some rate $R$ to be determined. Generate a random code $\cC$ of size $2 M$ by sampling each entry of the codebook independently from $P_\bfx$.

For any $\vbfx\in\cC$, by Lemma \ref{lem:prob_vec_type},
\begin{align*}
    \prob{\tau_{\vbfx}= P_\bfx} 
    =&1/\nu(n).
\end{align*}
Hence the expected number of codewords with type   $P_\bfx$ is  $2M/\nu(n)$.

For any $\paren{\vbfx_1,\cdots,\vbfx_L}\in\binom{\cC}{L}$,
\begin{align*}
    \prob{\tau_{\vbfx_1,\cdots,\vbfx_L}\in\cK^\tl\paren{P_\bfx}}
    \doteq&\sup_{P_{\bfx_1,\cdots,\bfx_L}\in\cK^\tl\paren{P_\bfx}}2^{-nD\paren{P_{\bfx_1,\cdots,\bfx_L}\|P_\bfx^\tl}},
\end{align*}
by Sanov's theorem \ref{thm:sanov}.
 Let $P^*\in\cK^{\tl}(P_\bfx)$ be the extremizer for the above supremum. Hence the expected number of confusable lists is at most
\begin{align*}
    \binom{2 M}{L}2^{-nD\paren{P^*\|P_\bfx^\tl}}
    \le&\paren{2 M}^L2^{-nD\paren{P^*\|P_\bfx^\tl}}.
\end{align*}
Pick $M$ such that
\[\paren{2 M}^L2^{-nD\paren{P^*\|P_\bfx^\tl}}\le M/\nu(n),\]
i.e.,
\begin{align*}
    L+nRL-nD\paren{P^*\|P_\bfx^\tl}\le nR-\log\nu(n).
\end{align*}
That is, $R$ can be taken arbitrarily close to $\frac{1}{L-1}D\paren{P^*\|P_\bfx^\tl}$.
\begin{align*}
    R\le&\frac{D\paren{P^*\|P_\bfx^\tl}}{L-1}-\frac{\log \nu(n)}{(L-1)n}-\frac{L}{(L-1)n}\\
    \stackrel{n\to\infty}{\to}&\frac{D\paren{P^*\|P_\bfx^\tl}}{L-1}.
\end{align*}
Now, we remove all codewords of types different from $P_\bfx$. We also remove one codeword from each of the confusable lists.  In expectation, this process reduces the size of the code by at most $2M-2M/\nu(n)$ (due to the first expurgation) plus  $\paren{2 M}^L2^{-nD\paren{P^*\|P_\bfx^\tl}}\le M/\nu(n)$ (due to the second expurgation). After expurgation, we get an $(L-1)$-list decodable $P_\bfx$-constant composition  code $\cC'$ of size at least 
\[ 2M-(2M/\nu(n)-2M/\nu(n))-M/\nu(n)=M/\nu(n).\]
 The rate $R'$ of $\cC'$ is asymptotically the same as $R$. 
\begin{align*}
    R'=&R-\frac{\log \nu(n)}{n}\\
    \stackrel{n\to\infty}{\to}&R.
\end{align*}
This finishes the proof. 
\end{proof}

\subsection{Cloud codes}
\begin{lemma}
If there is a $(P_\bfx,L)$-self-coupling ($P_\bfx\in\lambda_\bfx$)  
$P_{\bfx_1,\cdots,\bfx_L}\in\cJ^\tl\paren{P_\bfx}\setminus\cK^\tl\paren{P_\bfx}$ which can be decomposed into 
\begin{align*}
    &P_{\bfx_1,\cdots,\bfx_L}\paren{\vx_1,\cdots,\vx_L}\\
    =&\sum_{u\in\cU}P_\bfu\paren{u}P_{\bfx|\bfu}^\tl\paren{\vx_1,\cdots,\vx_L|u}\\
    =&\sum_{u\in\cU}P_\bfu\paren{u}\prod_{i=1}^LP_{\bfx|\bfu}\paren{\vx_i|u}.
\end{align*}
for some distributions $P_\bfu\in\Delta(\cU)$ of finite support $\cardU$ and $P_{\bfx|\bfu}\in\Delta(\cX|\cU)$. See Fig. \ref{fig:low_rate_code_from_cp_distr}.
\end{lemma}
\begin{proof}
The proof follows from a time-sharing argument combined with the previous low rate code construction (Lemma \ref{lem:low_rate}). 

\begin{figure}
    \centering
    \includegraphics[scale = 2]{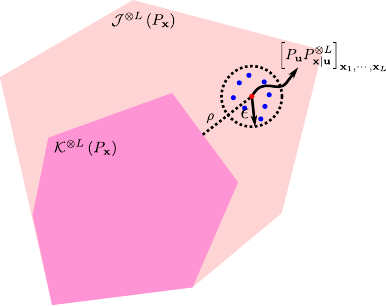}
    \caption{Low rate codes from $\cp$ distribution. If there is a $\cp$ distribution strictly outside $\cK^\tl(P_\bfx)$, then we can get a positive rate from random code using  time-sharing. The only variation is that we divide codebook into chunks according to $P_\bfu$ and construct random codes of shorter length for each chunk $u$ using distribution $P_{\bfx|\bfu = u}$.}
    \label{fig:low_rate_code_from_cp_distr}
\end{figure}

Fix $R$ to be determined later. Sample $2^{nR}$ codewords in $\cC$ independently from the following distribution. Divide each length-$n$ codeword into $\card{\cU}$ chunks $1,\cdots,\card{\cU}$. For the $u$-th ($u\in\cU$) chunk, sample $P_\bfu(u)n$ components in the chunk independently using distribution $P_{\bfx|\bfu=u}$.  Let $P_{\bfu,\bfx} = P_\bfu P_{\bfx|\bfu}$ and $P_\bfx = \sqrbrkt{P_{\bfu,\bfx}}_\bfx$. Let $\cC'$ be all codewords in $\cC$ of type $P_\bfx$. (See Fig. \ref{fig:cloud_codes}.)
\begin{figure}
    \centering
    \includegraphics{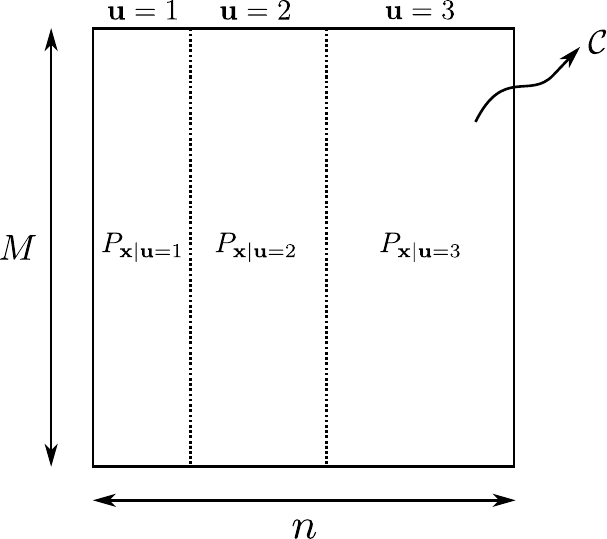}
    \caption{An example of cloud code construction in which $\cU = \curbrkt{1,2,3}$. The codebook is divided into 3 chunks and symbols in the $i$-th chunk are sampled independently from $P_{\bfx|\bfu = i}$ ($i = 1,2,3$).}
    \label{fig:cloud_codes}
\end{figure}
Define
\[\rho\coloneqq\inf_{P'_{\bfx_1,\cdots,\bfx_L}\in\cK^\tl(P_\bfx)}\normmav{ P_{\bfx_1,\cdots,\bfx_L} - P'_{\bfx,\cdots,\bfx_L} }.\]
Let
\[u^*\coloneqq \argmin{u\in\cU}P_\bfu(u).\]
Note that $P_\bfu\paren{u^*}>0$ since $\cardU$ is the \emph{support} of $P_\bfu$.
Let $R=\frac{P_\bfu(u^*)\log e}{12}\frac{\rho^2}{L}-\delta$.
We will show that
\begin{lemma}
A random $P_\bfx$-constant composition cloud code  as constructed above has rate $R=\frac{P_\bfu(u^*)\log e}{12}\frac{\rho^2}{L}-\delta$ and is $(L-1)$-list decodable with probability at least
\[1-2\exp\paren{ -\frac{2^{nR}}{12 \prod_{u\in\cU}\nu(P_\bfu(u)n)} }-2^{-n\delta + L\log\cardX+\log\cardU+1}.\]
\end{lemma}

We write a length-$n$ codeword as the concatenation of $\card{\cU}$ chunks,
\[\vbfx = \paren{ \vbfx^{(1)},\cdots,\vbfx^{(\cardU)} }.\]

First we argue that w.h.p. the code $\cC$ is almost $P_\bfx$-constant composition. 
The expected size of $\cC'$ is
\begin{align}
\expt{\cardCp} = &\expt{ \card{ \cC\cap\cT_{\vbfx}(P_{\bfx|\bfu}) } }  \notag\\
=&\sum_{i\in[M]}\prob{ \vbfx_i\in\cT_{\vbfx}(P_{\bfx|\bfu}) } \label{eqn:exp_lin_cloud}\\
=&\sum_{i\in[M]}\prob{ \bigcap_{u\in\cU} \curbrkt{ \vbfx_i^{(u)}\in\cT_{\vbfx^{(u)}}(P_{\bfx|\bfu = u}) } } \notag\\
=&\sum_{i\in[M]}\prod_{u\in\cU}\prob{ \vbfx^{(u)}\in\cT_{\vbfx^{(u)}}(P_{\bfx|\bfu = u}) } \label{eqn:chunk_indep}\\
\asymp&M\prod_{u\in\cU}\nu(P_\bfu(u)n)^{-1}, \label{eqn:asymp_fall_into_type_class}
\end{align}
where Eqn. \eqref{eqn:exp_lin_cloud} is by linearity of expectation; Eqn. \eqref{eqn:chunk_indep} follows since different chunks are independent; Eqn. \eqref{eqn:asymp_fall_into_type_class} follows from Lemma \ref{lem:prob_vec_type}.
Then by Lemma \ref{lem:cc_conc} 
\begin{align}
\prob{ \cardCp\notin(1\pm1/2)\expt{\cardCp} }\le & 2\exp\paren{ -\frac{2^{nR}}{12 \prod_{u\in\cU}\nu(P_\bfu(u)n)} }. \notag
\end{align}

Secondly, for any list $1\le i_1<\cdots<i_L\le M$ of distinct ordered messages,
\begin{align}
	\prob{ \exists u\in\cU,\;\normmav{ \tau_{\vbfx_{i_1}^{(u)},\cdots,\vbfx_{i_L}^{(u)}} - P_{\bfx|\bfu = u}^\tl } \ge\eps }\le & \sum_{u\in\cU}2\cdot\cardX^L\cdot\exp\paren{ -\frac{\rho^2}{12}nP_\bfu(u) }\label{eqn:same_calc}\\
	\le &2 \cardU\cardX^L\exp\paren{ -\frac{\rho^2}{12}nP_\bfu\paren{u^*} },\label{eqn:bound_by_ustar}
\end{align}
where the first inequality \eqref{eqn:same_calc} follows from a union bound and same calculations as in Lemma \ref{lem:low_rate}. The second inequality \eqref{eqn:bound_by_ustar} follows from the definition of $u^*$.

Finally, by taking another union bound over lists $\cL\in\binom{[M]}{L}$, we get
\begin{align*}
	\prob{ \exists(i_1,\cdots,i_L)\in\binom{\cM}{L},\;\exists u\in\cU,\; \normmav{ \tau_{\vbfx_{i_1}^{(u)},\cdots,\vbfx_{i_L}^{(u)}} - P_{\bfx|\bfu = u}^\tl } \ge\eps}\le & 2^{-n\paren{ \frac{\rho^2\log e P_\bfu\paren{u^*}}{12} - RL } + L\log\cardX + \log\cardU + 1}.
\end{align*}
Therefore, we have  that the probability that the random $P_\bfx$-constant composition cloud code $\cC'$ constructed above has rate $R=\frac{P_\bfu(u^*)\log e}{12}\frac{\rho^2}{L}-\delta$ and is  $(L-1)$-list decodable with probability at least
\[1-2\exp\paren{ -\frac{2^{nR}}{12 \prod_{u\in\cU}\nu(P_\bfu(u)n)} }-2^{-n\delta + L\log\cardX+\log\cardU+1},\]
which completes the proof.
\end{proof}


\begin{figure}
    \centering
    \begin{subfigure}{0.45\textwidth}
    	\centering
    	\includegraphics[scale = 2]{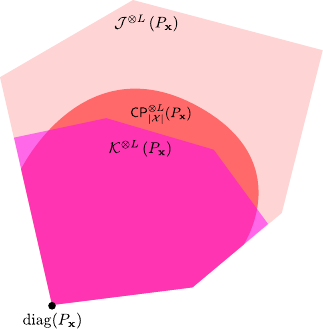}
    	\caption{``Below Plotkin point'', positive $(L-1)$-list decoding rate is possible. In this case, for some input distribution $P_\bfx\in\lambda_\bfx$, the slice of $P_\bfx$-self-coupling $\cp$ tensors is not entirely contained in the confusability set $\cK^\tl(P_\bfx)$.} 
    	\label{fig:below_plotkin_pt}
    \end{subfigure}~\quad
    \begin{subfigure}{0.45\textwidth}
    	\centering
    	\includegraphics[scale = 2]{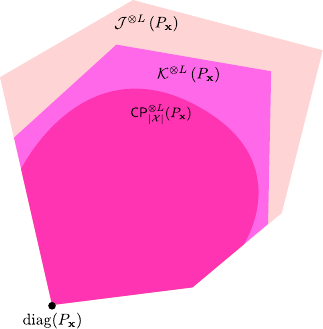}
    	\caption{``Above Plotkin point'', no positive rate for $(L-1)$-list decoding is achievable. In this case, for every input distribution $P_\bfx\in\lambda_\bfx$, the slice of $P_\bfx$-self-coupling $\cp$ tensors is  entirely contained in the confusability set $\cK^\tl(P_\bfx)$.}
    	\label{fig:above_plotkin_pt}
    \end{subfigure}
    \caption{A characterization of when positive rate generalized list decodable codes exist.}
    \label{fig:char_plotkin_pt}
\end{figure}
The above lemma apparently implies Theorem \ref{thm:achievability}.

\section{Converse}
\label{sec:converse}
Let $\cp_{\card{\cX}}^\tl (P_\bfx)\coloneqq\cp_{\card{\cX}}^\tl\cap\cJ^\tl\paren{P_\bfx}$ and $\sym_\cardX^\tl(P_\bfx) \coloneqq\sym_\cardX^\tl\cap\cJ^\tl(P_\bfx)$.

We have  shown in the previous section that if $\cp_{\card{\cX}}^\tl (P_\bfx)\cap\cK^\tl\paren{P_\bfx}^c\ne\emptyset$, then the $(L-1)$-list decoding capacity is positive.  In this section we are going to prove the converse. That is, such a condition is also necessary for positive rate being possible. Indeed, we will show that
\begin{theorem}[Converse]
Given a general adversarial channel $\cA = \paren{\cX,\lambda_\bfx,\cS,\lambda_\bfs,\cY,W_{\bfy|\bfx}}$, if for every admissible input distribution $P_\bfx\in\lambda_\bfx$, $\cp_\cardX^\tl(P_\bfx)\subseteq\cK^\tl(P_\bfx)$, then the $(L-1)$-list decoding capacity of $\cA$ is zero.
\end{theorem}

\subsection{Equicoupled subcode extraction}
\begin{definition}[Equicoupledness and $\eps$-equicoupledness]
A code $\cC$ is said to be \emph{$P_{\bfx_1,\cdots,\bfx_L}$-equicoupled}  if for all \emph{ordered}  lists $(\vx_{i_1},\cdots,\vx_{i_L})\in\binom{\cC}{ L}$ where $1\le i_1<\cdots<i_L\le\cardC$, $\tau_{\vx_{i_1},\cdots,\vx_{i_L}}=P_{\bfx_{1},\cdots,\bfx_{L}}$. A code $\cC$ is said to be \emph{$(\zeta,P_{\bfx_1,\cdots,\bfx_L})$-equicoupled} if for all ordered lists $(\vx_{i_1},\cdots,\vx_{i_L})\in\binom{\cC}{ L}$, where $1\le i_1<\cdots<i_L\le \cardC$, $\normsav{\tau_{\vx_{1},\cdots,\vx_{i_L}}-P_{\bfx_{1},\cdots,\bfx_{L}}}\le\eps$. 
\end{definition}
\begin{remark}
The above definition can also be overloaded for  sequences of random variables or their joint distributions. We say a sequence of random variables $\bfw_1,\cdots,\bfw_M$ or the joint distribution $P_{\bfw_1,\cdots,\bfw_M}$ is $P_{\bfx_1,\cdots,\bfx_L}$-equicoupled (or $(\zeta,P_{\bfx_1,\cdots,\bfx_L})$-equicoupled) if every order-$L$ marginal $P_{\bfw_{i_1},\cdots,\bfw_{i_L}}$ ($1\le i_1<\cdots<i_L\le M$) equals (or is $\zeta$-close to in $\normsav{\cdot}$) $P_{\bfx_1,\cdots,\bfx_L}$.
\end{remark}

Using the hypergraph Ramsey's theorem, we first show that any infinite sequence of codes of positive rate has an infinite sequence of subcodes which are $\zeta$-equicoupled. 
\begin{lemma}[Equicoupled subcode extraction]\label{lem:subcode_ext}
For any infinite sequence of codes $\curbrkt{\cC_{i}}_{i\ge1}$ of blocklengths $n_i$'s and positive rate, where $\curbrkt{n_i}_{i\ge1}$ is an infinite increasing integer sequence, for any $\zeta>0$ and any $M\in\bZ_{>0}$, there is an $N\in\bZ_{>0}$ such that if $\card{\cC_{i}}\ge N$ then $\cC'$ contains a subcode $\cC_{i}'$ satisfying that
\begin{itemize}
    \item $\card{\cC_{i}'}\ge M$;
    \item $\cC_{i}'$ is $(\zeta,P_{\bfx_1,\cdots,\bfx_L})$-equicoupled for some $P_{\bfx_1,\cdots,\bfx_L}$.
\end{itemize}
See Fig. \ref{fig:equicoupled_subcode_extraction}.
\end{lemma}
\begin{figure}
    \centering
    \includegraphics[scale=2]{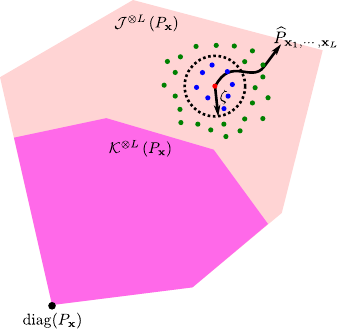}
    \caption{Equicoupled subcode extraction using hypergraph Ramsey's theorem. The union of green and blue dots represents the set of all joint types of ordered $L$-lists in $\cC$. The blue dots correspond to joint types of its subcode $\cC'$. (Note that they are all non-confusable.) They are clustered within a small ball (w.r.t. sum-absolute-value norm) centered at some distribution $\wh P_{\bfx_1,\cdots,\bfx_L}$. Since the hypergraph Ramsey number  is finite, there exists such $\cC'$ which is suitably large.}
    \label{fig:equicoupled_subcode_extraction}
\end{figure}

Again, this lemma is a consequence of the hypergraph Ramsey's theorem. Let $R^{(m)}_c(n_1,\cdots,n_c)$ be the smallest integer $n$ such that the complete $m$-uniform hypergraph on $n$ vertices with any $c$-colouring of hyperedges contains at least one of a clique of colour 1 and size $n_1$, ..., a clique of colour $c$ and size $n_c$. It is known that $R^{(m)}_c(n_1,\cdots,n_c)$ is finite (Lemma \ref{lem:finiteness_hypergraph_ramsey_number}), i.e., independent of the size $n$ of the hypergraph. 

\begin{proof}[Proof of Lemma \ref{lem:subcode_ext}]

Recall that we assume $\cp^\tl_{\cardX}(P_\bfx)\cap\cK^\tl(P_\bfx)^c = \emptyset$. Let $\rho$ be the gap between $\cp_\cardX^\tl(P_\bfx)$ and $\cK^\tl(P_\bfx)$,
\[\rho\coloneqq\inf_{\substack{P\in\cp^\tl_{\cardX}(P_\bfx)\\P'\in\cJ^\tl(P_\bfx)\setminus\cK^\tl(P_\bfx)}}\normsav{P - P'}.\]

\begin{definition}[$\eps$-net]
For a metric space $(\cX,d)$, an $\eps$-net $\cN\subset\cX$ is a subset  which is a discrete $\eps$-approximation of $\cX$ in the sense that for any $x\in\cX$, there is an $x'\in\cN$ such that $d(x,x')\le\eps$.
\end{definition}

We claim that
\begin{lemma}[Bound on size of $\eps$-net]
\label{lem:net}
There is an $\eps$-net $\cN$ of $\cJ^\tl(P_\bfx)\setminus\cK^\tl(P_\bfx)$ equipped with $\ell^1$ metric of size at most $\paren{\frac{\cardX^L}{2\eps}+1}^{\cardX^L}$.
\end{lemma}
\begin{proof}
The following construction is by no means optimal, but its size has a \emph{finite} upper bound which is enough for our purposes. Indeed, it suffices to take $\cN$ to be the coordinate-quantization net of $\cJ^\tl(P_\bfx)\setminus\cK^\tl(P_\bfx)$. Note that for any $P\in\cJ^\tl(P_\bfx)$, each entry of $P$ lies in $[0,1]$. Take $\delta \coloneqq \frac{2\eps}{\cardX^L}$. Divide $[0,1]$ into sub-intervals of length $\delta$ (possibly except the last sub-interval that may have length less than $\delta$). For each entry of $P$, there are at most $\frac{1}{\delta} + 1$ sub-intervals. Quantize each component of $P$ to the nearest middle point of these sub-intervals.  The set of all representatives whose components take values from the set of  middle points of the sub-intervals form a net $\cN$. In total, there are at most $\paren{\frac{1}{\delta} + 1}^{\cardX^L}$ such representatives. For any $P\in\cJ^\tl(P_\bfx)\setminus\cK^\tl(P_\bfx)$, let $Q_{\cN}(P)$ denote the quantization of $P$ using $\cN$, i.e., 
\[Q_{\cN}(P)\coloneqq\argmin{P'\in\cN}\normsav{P-P'}.\]
The quantization error is at most
\begin{align*}
    \normsav{P-Q_\cN(P)}\le&\sum_{(x_1,\cdots,x_L)\in\cX^L}\abs{ P(x_1,\cdots,x_L) - Q_\cN(P)(x_1,\cdots,x_L) }\\
    \le&\cardX^L\frac{\delta}{2}\\
    \le&\eps.
\end{align*}
We thus have shown that $\cN$ constructed as above is an $\eps$-quantizer of small cardinality.
\end{proof}

Let
\begin{equation}
    \lambda\coloneqq-\sup_{\wh P\in\paren{\cJ^\tl(P_\bfx)\setminus\cK^\tl(P_\bfx)}\cap\sym_{\cardX}^\tl(P_\bfx)}\inf_{Q\in\cop_{\cardX}^\tl}\inprod{\wh P}{Q}.
    \label{eqn:lambda}
\end{equation}
We know that $\cp$ cone and $\cop$ cone are dual (Theorem \ref{thm:cp_cop_duality}) in the space of symmetric tensor cone. Thus, for any non-$\cp$ \emph{symmetric} tensor $\wh P\in\sym_{\cardX}^\tl(P_\bfx)\setminus\cp_\cardX^\tl(P_\bfx)$, there must be a witness $Q$ with strictly negative inner product with $\wh P$. The infimum
\[\inf_{Q\in\cop_{\cardX}^\tl}\inprod{\wh P}{Q}<0.\]
$\lambda$ is the  absolute value of the smallest inner product  among all symmetric non-$\cp$ tensors. We know that $\lambda>0$, since $\cp_\cardX^\tl(P_\bfx)$ is \emph{strictly} contained in $\cK^\tl(P_\bfx)$.

Let  
\begin{equation}
    \zeta\coloneqq\frac{1}{2}\min\curbrkt{\rho,\frac{\lambda}{\cardX^L}}.
    \label{eqn:zeta}
\end{equation}
Take a $\zeta$-net of $\paren{\Delta\paren{\cX^L},\ell^1}$ as constructed in  Lemma \ref{lem:net}. Such a net has cardinality at most $K\coloneqq\paren{\frac{\cardX^L}{\rho}+1}^{\cardX^L}$.

Build an $L$-uniform complete hypergraph $\cH=(\cC,\cE)$ on $\cC$. The vertices of $\cH$ are codewords in $\cC$. For every tuple $\paren{\vx_{i_1},\cdots,\vx_{i_L}}\in\binom{\cC}{L}$ (where the indices $1\le i_1<\cdots<i_L\le \cardC$ are sorted in ascending order) of distinct codewords, there is a hyperedge connecting them. There are totally $\binom{\cardC}{L}$ hyperedges in $\cE$. We now label hyperedges using distributions in $\cN$. For each hyperedge $\paren{\vx_{i_1},\cdots,\vx_{i_L}}\in\cE$, label it using the unique element $Q_\cN\paren{\tau_{\vx_{i_1},\cdots,\vx_{i_L}}}$ from $\cN$. This can be viewed as an edge colouring of $\cH$ using at most $K$ colours. 

By hypergraph Ramsey's theorem (Theorem \ref{lem:finiteness_hypergraph_ramsey_number}), there is a constant $N$ 
such that if the size $|\cC|$ of the hypergraph is at least $ N$, then there is a monochromatic (each hyperedge in the sub-hypergraph has the same colour) clique $\cC'\subset\cC$ of size at least $M$. Indeed, we can take $N$ to be the hypergraph Ramsey number
$N = R_K^{(L)}(M,\cdots,M)$. By Theorem \ref{lem:bounds_hypergraph_ramsey_number}, there is a constant $c'>0$ such that
$N<t_L(c'\cdot K\log K)$, where $t_L(\cdot)$ is the  tower function of height $L$.
Put in another way, there exists a subcode $\cC'\subset\cC$ of size at least $M$ such that for some distribution $\wh P_{\bfx_1,\cdots,\bfx_L}\in\cN$, the joint type of every ordered tuple of $L$ distinct codewords in $\cC'$ is $\zeta$-close to $\wh P_{\bfx_1,\cdots,\bfx_L}$.  I.e., for every $\cL=(\vx_1,\cdots,\vx_L)\in\binom{\cC'}{L}$,
\[\normsav{\tau_{\vx_1,\cdots,\vx_L} - \wh P_{\bfx_1,\cdots,\bfx_L}}\le\zeta.\]
This completes the proof of Lemma \ref{lem:subcode_ext}.
\end{proof}

Before proceeding with the proof of converse, we first list several corollaries that directly follow from the above lemma. They are concerned with basic properties of $\paren{\zeta,P_{\bfx_1,\cdots,\bfx_L}}$-equicoupled codes.
\begin{corollary}
\label{cor:joint_type_close_to_eachother}
Any two lists of $L$ (ordered) codewords from $\cC'$ have joint types $2\zeta$ close to each other in sum-absolute-value distance. 
\end{corollary}

\begin{proof}
For any $\cL_1=(\vx_{i_1},\cdots,\vx_{i_L})$ and $\cL_2=(\vx_{j_1},\cdots,\vx_{j_L})$ in $\binom{\cC'}{L}$,
\begin{align}
    \normsav{ \tau_{\vx_{i_1},\cdots,\vx_{i_L}} - \tau_{\vx_{j_1},\cdots,\vx_{j_L}} }\le&\normsav{\tau_{\vx_{i_1},\cdots,\vx_{i_L}} - \wh P_{\bfx_1,\cdots,\bfx_L}} + \normsav{\wh P_{\bfx_1,\cdots,\bfx_L} - \tau_{\vx_{j_1},\cdots,\vx_{j_L}}}\notag\\
    \le&\zeta+\zeta\notag\\
    =&2\zeta.\label{eqn:equidist}
\end{align}
\end{proof}

\begin{corollary}
\label{cor:joint_type_short_close_to_eachother}
Any two size-$\ell$ ($1\le\ell\le L$) lists in $\cC'$ have joint type $2\zeta$ close to each other in sum-absolute-value distance, provided $|\cC'|>2L$.
\end{corollary}

\begin{proof}
For any $\cL_1'=(\vx_{i_1},\cdots\vx_{i_{L-1}})$ and $\cL_2'=(\vx_{j_1},\cdots,\vx_{j_{L-1}})$ in $\binom{\cC'}{L-1}$, take $\vx_{\iota}\in\cC'\setminus(\cL_1'\cup\cL_2')$. (This can be done as long as $|\cC'|>2L$.) Without loss of generality, assume $\iota>\max\{i_{L-1},j_{L-1}\}$. Let $\cL_1\coloneqq\cL_1'\cup\{\vx_\iota\},\cL_2\coloneqq\cL_2'\cup\{\vx_\iota\}$. We know that
\begin{align*}
    2\zeta\ge&\normsav{\tau_{\vx_{i_1},\cdots,\vx_{i_{L-1}},\vx_{\iota}} - \tau_{\vx_{j_1},\cdots,\vx_{j_{L-1}},\vx_\iota}}\\
    =&\sum_{(x_1,\cdots,x_{L-1},x)\in\cX^L}\abs{\tau_{\vx_{i_1},\cdots,\vx_{i_{L-1}},\vx_{\iota}}(x_1,\cdots,x_{L-1},x) - \tau_{\vx_{j_1},\cdots,\vx_{j_{L-1}},\vx_\iota}(x_1,\cdots,x_{L-1},x)}\\
    \ge&\sum_{(x_1,\cdots,x_{L-1})\in\cX^{L-1}}\abs{ \sum_{x\in\cX}\paren{ \tau_{\vx_{i_1},\cdots,\vx_{i_{L-1}},\vx_{\iota}}(x_1,\cdots,x_{L-1},x) - \tau_{\vx_{j_1},\cdots,\vx_{j_{L-1}},\vx_\iota}(x_1,\cdots,x_{L-1},x) } }\\
    =&\sum_{(x_1,\cdots,x_{L-1})\in\cX^{L-1}}\abs{ \tau_{\vx_{i_1},\cdots,\vx_{i_{L-1}}}(x_1,\cdots,x_{L-1}) - \tau_{\vx_{j_1},\cdots,\vx_{j_{L-1}}}(x_1,\cdots,x_{L-1}) }\\
    =&\normsav{\tau_{\vx_{i_1},\cdots,\vx_{i_{L-1}}} - \tau_{\vx_{j_1},\cdots,\vx_{j_{L-1}}}}.
\end{align*}
Similarly we can see that Eqn. \eqref{eqn:equidist} holds also for size-$\ell$ ($\ell\le L$) lists.
\end{proof}

For a subset $\cB\subset[n]$, we let $P_{\bfx_{\cB}}$ denote the marginalization of $P_{\bfx_1,\cdots,\bfx_L}$ onto the random variables indexed by elements in $\cB$, $\sqrbrkt{P_{\bfx_1,\cdots,\bfx_L}}_{\curbrkt{\bfx_i\colon i\in\cB}}$.
\begin{corollary}
\label{cor:marginal_close}
For any $1\le \ell< L$ and any subsets $\cL_1',\cL_2'\in\binom{[n]}{\ell}$, $P_{\bfx_{\cL_1'}}$ and $P_{\bfx_{\cL_2'}}$ are $3\zeta$ close to each other in sum-absolute-value distance, given $|\cC'|>2L$.
\end{corollary}

\begin{proof}
Given two subsets $\cL_1',\cL_2'\subset[n]$ both of cardinality $\ell< L$, as long as the code size $M$ is larger than $2L$, we can always find a tuple $1\le i_1<\cdots<i_\ell\le M$ such that it can be completed to  $L$-tuples $\cL_1,\cL_2$ in two different ways
\begin{align*}
\begin{array}{rllllllllllll}
\cL_1=&({i_1},&\cdots,&{i_{\ell-\ell'}},&{i_{\ell-\ell'+1}},&\cdots,&{i_\ell},&{j_{1}},&\cdots,&{j_{\ell-\ell'}},&l_1,&\cdots,&l_{L-(2\ell-\ell')}),\\
\cL_2=&({k_1},&\cdots,&{k_{\ell-\ell'}},&{i_{1}},&\cdots,&{i_\ell'},&{i_{\ell'+1}},&\cdots,&{i_\ell},&l_1,&\cdots,&l_{L-(2\ell-\ell')}),
\end{array}
\end{align*}
for some $1\le k_1<\cdots<k_{\ell-\ell'}<i_1<\cdots<i_\ell<j_1<\cdots<j_{\ell-\ell'}<l_1<\cdots<l_{L-(2\ell-\ell')}\le M$, where $\ell' = \card{\cL_1'\cap\cL_2'}$. See Fig. \ref{fig:marginal_close}.
\begin{figure}
    \centering
    \includegraphics[scale = 2]{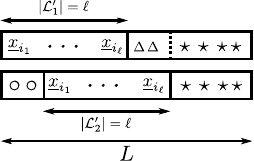}
    \caption{Two ways to complete the size-$\ell$ list  $i_1,\cdots,i_\ell$ to size-$L$ lists $\cL_1,\cL_2$, respectively. Triangles $\Delta$, circles $\circ$ and stars $\star$ represent indices $j$'s, $k$'s and $l$'s, respectively.}
    \label{fig:marginal_close}
\end{figure}
We know that
\begin{align*}
\normsav{\tau_{\cL_1} - P_{\bfx_1,\cdots,\bfx_L}}\le & \zeta,\\
\normsav{\tau_{\cL_2} - P_{\bfx_1,\cdots,\bfx_L}}\le & \zeta.
\end{align*}
Note that 
\begin{align*}
\zeta\ge&\normsav{\tau_{\vx_{\cL_1}} - P_{\bfx_1,\cdots,\bfx_L}}\\
 = &\sum_{\cL_1\in\curbrkt{0,1}^L}\abs{ \tau_{\vx_{\cL_1}}(\cL_1) - P_{\bfx_1,\cdots,\bfx_L}(\cL_1) }\\
 \ge&\sum_{i_1,\cdots,i_{\ell}}  \abs{ \sum_{\cL_1\setminus\curbrkt{i_1,\cdots,i_{\ell}}\in\curbrkt{0,1}^{L-\ell} } \tau_{\vx_{\cL_1}}(i_1,\cdots,i_\ell,\cL_1\setminus\curbrkt{i_1,\cdots,i_{\ell}}) - P_{\bfx_1,\cdots,\bfx_L}(i_1,\cdots,i_\ell,\cL_1\setminus\curbrkt{i_1,\cdots,i_{\ell}}) }\\
\le&\sum_{i_1,\cdots,i_{\ell}} \abs{\tau_{\vx_{i_1},\cdots,\vx_{i_\ell}}(i_1,\cdots,i_\ell) - P_{\bfx_{\cL_1'}}(i_1,\cdots,i_\ell) }\\
=&\normsav{\tau_{\vx_{i_1},\cdots,\vx_{i_\ell}} - P_{\bfx_{\cL_1'}}}.
\end{align*}
Similarly,
\[\normsav{\tau_{\vx_{i_1},\cdots,\vx_{i_\ell}} - P_{\bfx_{\cL_2'}}}\le\zeta.\]
By triangle inequality, 
\begin{align*}
\normsav{ P_{\bfx_{\cL_1'}} - P_{\bfx_{\cL_2'}} } \le&\normsav{ P_{\bfx_{\cL_1'}} -\tau_{\vx_{i_1},\cdots,\vx_{i_\ell}} } + \normsav{ \tau_{\vx_{i_1},\cdots,\vx_{i_\ell}} - P_{\bfx_{\cL_2'}} }\\
\le&2\zeta.
\end{align*}
\end{proof}

\begin{corollary}
\label{cor:joint_type_short_close_to_marginal}
A $\paren{\zeta,P_{\bfx_1,\cdots,\bfx_L}}$-equicoupled code $\cC'$ is $\paren{3\zeta,P_{\bfx_1,\cdots,\bfx_\ell}}$-equicoupled for any $1\le\ell\le L$, as long as $|\cC'|>2L$.
\end{corollary}

\begin{proof}
For any list of codewords $\vx_{i_1},\cdots,\vx_{i_\ell}$, we can always find a completion of $(i_1,\cdots,i_\ell)$ to an $L$-tuple. Let $\cT$ denote the set of locations of $i_1,\cdots,i_\ell$ in the completion. We know that
\[\normsav{ \tau_{\vx_{i_1},\cdots,\vx_{i_\ell}} - P_{\bfx_{\cT}} }\le\zeta.\]
By the previous corollary,
\begin{align*}
\normsav{ \tau_{\vx_{i_1},\cdots,\vx_{i_\ell}} - P_{\bfx_1,\cdots,\bfx_\ell} }\le&\normsav{ \tau_{\vx_{i_1},\cdots,\vx_{i_\ell}} - P_{\bfx_{\cT}} }  + \normsav{ P_{\bfx_{\cT}}  - P_{\bfx_1,\cdots,\bfx_\ell} }\\
\le&\zeta + 2\zeta\\
=& 3\zeta.
\end{align*}
\end{proof}

Now we apply the double counting trick used in the Plotkin-type bound for list decoding. We want to show that if $\wh P_{\bfx_1,\cdots,\bfx_L}$ is not completely positive, then any $(L-1)$-list decodable code cannot be large.

\begin{definition}[Symmetry of tensors]
\label{def:symm_tensor}
A tensor $T\in\ten_n^{\otimes m}$ is said to be \emph{symmetric} if its components are invariant under permutation of indices, i.e., for any $\sigma\in S_m$ and any $(t_1,\cdots,t_m)\in[n]^m$, 
\[T(t_1,\cdots,t_m) = T\paren{t_{\sigma(1)},\cdots,t_{\sigma(m)}}.\]
The set of dimension-$n$ order-$m$ symmetric tensors is denoted by $\sym_n^{\otimes m}$.
\end{definition}

\subsection{Symmetric case}
\label{sec:converse_symm}
In this subsection, assume $\wh P_{\bfx_1,\cdots,\bfx_L}$ is {symmetric} as a dimension-$\card{\cX}$ order-$L$ tensor. We are going to show that
\begin{lemma}[Converse, symmetric case]
\label{lem:converse_symm}
For a general adversarial channel $\cA=\paren{ \cX,\lambda_\bfx,\cS,\lambda_\bfs,\cY,W_{\bfy|\bfx,\bfs} }$ and an admissible input distribution $P_\bfx\in\lambda_\bfx$, if 
	$\cp_\cardX^\tl(P_\bfx)\subseteq\cK^\tl(P_\bfx)$,
the any $\paren{\zeta,P_{\bfx_1,\cdots,\bfx_L}}$-equicoupled $(L-1)$-list decodable code $\cC'$ has size at most 
\[|\cC'|\le \max\curbrkt{2(L-1),\frac{2^{L+1}L!}{\lambda}},\]
where   $\wh P_{\bfx_1,\cdots,\bfx_L}\in\sym_\cardX^\tl(P_\bfx)\cK^\tl(P_\bfx)$  is a symmetric, non-confusable joint distribution.
\end{lemma}

\begin{proof}
Since $\wh P_{\bfx_1,\cdots,\bfx_L}\in\sym_\cardX^\tl(P_\bfx)\setminus\cp_{|\cX|}^\tl$, by duality (Theorem \ref{thm:cp_cop_duality}) between the $\cp$ tensor cone and $\cop$ tensor cone,  there is a copositive tensor $Q\in\cop_{|\cX|}^\tl$ such that $\normf{Q}=1$ (by normalization) and 
\begin{equation}
    \inprod{P_{\bfx_1,\cdots,\bfx_L}}{Q}=-\eta
    \label{eqn:witness}
\end{equation}
for some $\eta>0$. Note that, by definition of $\lambda$, $\eta>\lambda$. We will bound 
\[\sum_{(i_1,\cdots,i_L)\in[|\cC'|]^L}\inprod{\tau_{\vx_{i_1},\cdots,\vx_{i_L}}}{Q}\] 
from above and below and argue that if $\card{\cC'}$ is larger than some constant\footnote{Note that we will actually show that the size of the code is upper bounded by a  \emph{constant} (independent of blocklength $n$), not just that the rate of the code is vanishing.}, then we get a strictly negative upper bound and a non-negative lower bound. Such a contradiction implies that no positive rate is possible for $(L-1)$-list decoding if $\wh P_{\bfx_1,\cdots,\bfx_L}$ is a non-$\cp$ symmetric distribution. 

\noindent\textbf{Upper bound} \\
\noindent\emph{\underline{Case when $i_1,\cdots,i_L\in[\cardCp]$ are not all distinct.} }
For  $i_1\le \cdots\le i_L\in[\cardCp]$ not all distinct, 
\begin{align}
    \inprod{\tau_{\vx_{i_1},\cdots,\vx_{i_L}}}{Q}\le&\normf{\tau_{\vx_{i_1},\cdots,\vx_{i_L}}}\normf{Q}\label{eqn:cs}\\
    \le&\normsav{\tau_{\vx_{i_1},\cdots,\vx_{i_L}}}\normf{Q}\label{eqn:norm_f_to_1}\\
    \le&1.\label{eqn:norm_f_q}
\end{align}
Eqn. \eqref{eqn:cs} is by Cauchy--Schwarz inequality. Eqn. \eqref{eqn:norm_f_to_1} is because $q$-norm of a vector is non-increasing in $q$. Eqn. \eqref{eqn:norm_f_q} is because a probability/type vector has one-norm $1$ and $Q$ is normalized to have $F$-norm $1$.

Thus 
\[\sum_{\substack{(i_1,\cdots,i_L)\in[|\cC'|]^L\\\text{not all distinct}}}\inprod{\tau_{\vx_{i_1},\cdots,\vx_{i_L}}}{Q}\le|\cC'|^L-\binom{\card{\cC'}}{L}L!.\]

\noindent\emph{\underline{Case when $i_1,\cdots,i_L\in[\cardCp]$ are  all distinct.}} By Lemma \ref{lem:subcode_ext},  for any  $\vx_{i_1},\cdots,\vx_{i_L}\in\cC'$ distinct, 
\begin{align*}
    \normmav{\tau_{\vx_{i_1},\cdots,\vx_{i_L}} - \wh P_{\bfx_{1},\cdots,\bfx_{L}}}\le&\normsav{\tau_{\vx_{i_1},\cdots,\vx_{i_L}} - \wh P_{\bfx_{1},\cdots,\bfx_{L}}}\\
    \le&\zeta.
\end{align*}
For any $(i_1,\cdots,i_L)\in\binom{\cardCp}{L}$ distinct, let $\Delta_{i_1,\cdots,i_L}\coloneqq \tau_{\vx_{i_1},\cdots,\vx_{i_L}} - \wh P_{\bfx_{1},\cdots,\bfx_{L}}$. Immediately, $\normmav{\Delta_{i_1,\cdots,i_L}}\le\zeta$.

Now,
\begin{align*}
    \inprod{\tau_{\vx_{i_1},\cdots,\vx_{i_L}}}{Q} = \inprod{\Delta_{i_1,\cdots,i_L}}{Q} + \inprod{\wh P_{\bfx_{1},\cdots,\bfx_L}}{Q}.
\end{align*}
Note that
\begin{align}
    \abs{\inprod{\Delta_{i_1,\cdots,i_L}}{Q}}=&\abs{\sum_{(x_1,\cdots,x_L)\in\cX^L}\Delta_{i_1,\cdots,i_L}(x_1,\cdots,x_L)Q(x_1,\cdots,x_L)}\notag\\
    \le&\sum_{(x_1,\cdots,x_L)\in\cX^L}\abs{\Delta_{i_1,\cdots,i_L}(x_1,\cdots,x_L)}\label{eqn:bound_q}\\
    \le&\card{\cX}^L\cdot\zeta,\label{eqn:error_term}
\end{align}
where Eqn. \eqref{eqn:bound_q} follows from triangle inequality and  $\normmav{Q}\le\normsav{Q}\le\normf{Q} = 1$.

Hence 
\begin{align}
    \inprod{\tau_{\vx_{i_1},\cdots,\vx_{i_L}}}{Q}\le&-\eta + \card{\cX}^L\zeta\label{eqn:bound_ip_t_q}\\
    \le&-\lambda+\frac{\lambda}{2}\label{eqn:choice_zeta}\\
    =&-\frac{\lambda}{2},\notag
\end{align}
where Eqn. \eqref{eqn:bound_ip_t_q} follows from Eqn. \eqref{eqn:witness} and Eqn. \eqref{eqn:error_term}, Eqn. \eqref{eqn:choice_zeta} is by the definition of $\lambda$ (Eqn. \eqref{eqn:lambda}) and the choice of $\zeta$ (Eqn. \eqref{eqn:zeta}).

Therefore,
\begin{align*}
    \sum_{(i_1,\cdots,i_L)\in[|\cC'|]^L\text{ distinct}} \inprod{\tau_{\vx_{i_1},\cdots,\vx_{i_L}}}{Q}\le&-\frac{\lambda}{2}\binom{\card{\cC'}}{L}L!.
\end{align*}

Overall,
\begin{align}
    \sum_{(i_1,\cdots,i_L)\in[|\cC'|]^L}\inprod{\tau_{\vx_{i_1},\cdots,\vx_{i_L}}}{Q}\le&|\cC'|^L-\binom{\card{\cC'}}{L}L!-\frac{\lambda}{2}\binom{\card{\cC'}}{L}L!\notag\\
    <&0\label{eqn:bound_C_prime}
\end{align}
if $|\cC'|$ is sufficiently large. To see this, note that $p\paren{\card{\cC'}}\coloneqq|\cC'|^L-\binom{\card{\cC'}}{L}L!$ is a polynomial in $\card{\cC'}$ of degree $L-1$, while $-\frac{\lambda}{2}\binom{\card{\cC'}}{L}L!$ is a polynomial in $\card{\cC'}$ of degree $L$. To give an explicit bound on $\card{\cC'}$, note that the RHS of \eqref{eqn:bound_C_prime} equals
\begin{align*}
    p\paren{\card{\cC'}} - \frac{\lambda}{2}\card{\cC'}\paren{\card{\cC'} - 1}\cdots\paren{\card{\cC'} - (L-1)}\le&L\cdot (L-1)!\cdot \card{\cC'}^{L-1} - \frac{\lambda}{2}\paren{\card{\cC'} - (L-1)}^L\\
    =&L!\cdot \card{\cC'}^{L-1} - \frac{\lambda}{2}\paren{\card{\cC'} - (L-1)}^L.
\end{align*}
In the above inequality, to upper bound $p\paren{\card{\cC'}}$, we replace each term of $p$ with a monomial with the largest possible coefficient in absolute value and the largest possible degree. To make the RHS negative, we want
\[(L!)^{\frac{1}{L}}\card{\cC'}^{1-\frac{1}{L}}<\paren{\frac{\lambda}{2}}^{\frac{1}{L}}\card{\cC'} - \paren{\frac{\lambda}{2}}^{\frac{1}{L}}(L-1).\]
One can easily check that when $\card{\cC'}>2(L-1)$,
\[\frac{1}{2}\paren{\frac{\lambda}{2}}^{\frac{1}{L}}\card{\cC'}<\paren{\frac{\lambda}{2}}^{\frac{1}{L}}\card{\cC'} - \paren{\frac{\lambda}{2}}^{\frac{1}{L}}(L-1).\]
Moreover, when $\card{\cC'}>\frac{2^{L+1}L!}{\lambda}$, 
\[(L!)\card{\cC'}^{1-\frac{1}{L}}<\frac{1}{2}\paren{\frac{\lambda}{2}}^{\frac{1}{L}}\card{\cC'}\]
is satisfied, so is the original inequality \eqref{eqn:bound_C_prime}.

Overall, we have that \[\sum_{(i_1,\cdots,i_L)\in[|\cC'|]^L}\inprod{\tau_{\vx_{i_1},\cdots,\vx_{i_L}}}{Q}<0\] as long as
\begin{align}
    \card{\cC'}>\max\curbrkt{2(L-1),\frac{2^{L+1}L!}{\lambda}}.
    \label{eqn:final_bound_C_prime}
\end{align}
Though the bound \eqref{eqn:final_bound_C_prime} is crude, it is a \emph{constant} not depending on the blocklength $n$.

\noindent\textbf{Lower bound}
\begin{align}
    &\sum_{(i_1,\cdots,i_L)\in[|\cC'|]^L}\inprod{\tau_{\vx_{i_1},\cdots,\vx_{i_L}}}{Q}\notag\\
    =&\sum_{(i_1,\cdots,i_L)\in[|\cC'|]^L}\sum_{(x_1,\cdots,x_L)\in\cX^L}\tau_{\vx_{i_1},\cdots,\vx_{i_L}}(x_1,\cdots,x_L)Q(x_1,\cdots,x_L)\notag\\
    =&\sum_{(x_1,\cdots,x_L)\in\cX^L}\sum_{(i_1,\cdots,i_L)\in[|\cC'|]^L}\frac{1}{n}\sum_{j=1}^n\indicator{\vx_{i_1}(j)=x_1,\cdots,\vx_{i_L}(j)=x_L}Q(x_1,\cdots,x_L)\notag\\
    =&\frac{1}{n}\sum_{(x_1,\cdots,x_L)\in\cX^L}\sum_{j=1}^n\sum_{(i_1,\cdots,i_L)\in[|\cC'|]^L}\indicator{\vx_{i_1}(j)=x_1}\cdots\indicator{\vx_{i_L}(j)=x_L}Q(x_1,\cdots,x_L)\notag\\
    =&\frac{1}{n}\sum_{(x_1,\cdots,x_L)\in\cX^L}\sum_{j=1}^n\paren{\sum_{i\in[|\cC'|]}\indicator{\vx_i(j)=x_1} }\cdots\paren{\sum_{i\in[|\cC'|]}\indicator{\vx_i(j)=x_L} }Q(x_1,\cdots,x_L)  \notag\\
    =&\frac{|\cC'|^L}{n}\sum_{(x_1,\cdots,x_L)\in\cX^L}\sum_{j=1}^nP_\bfx^{(j)}(x_1)\cdots P_\bfx^{(j)}(x_L)Q(x_1,\cdots,x_L)\label{eqn:empirical}\\
    =& \frac{|\cC'|^L}{n}\sum_{j=1}^n\inprod{\paren{P_\bfx^{(j)}}^\tl}{Q} \notag\\
    \ge&0.\label{eqn:nonnegative}
\end{align}
To see  equality \eqref{eqn:empirical}, let $P_\bfx^{(j)}$ be the empirical distribution of the $j$-th \emph{column} of $\cC'$ as a $\cardCp\times n$ matrix, i.e., for $x\in\cX$,
\[P_\bfx^{(j)}(x)\coloneqq\frac{1}{|\cC'|}\sum_{i=1}^{|\cC'|}\indicator{\vx_i(j)=x}.\]
The last inequality \eqref{eqn:nonnegative} follows since $\paren{P_\bfx^{(j)}}^\tl$ is a completely positive tensor.

The lower bound and the upper bound are contradicting each other, which completes the proof.
\end{proof}


\subsection{Asymmetric case}
In this section, we handle the asymmetric case of the converse.
\begin{definition}[Asymmetry of tensors]
\label{def:asymm_tensor}
For a joint distribution $P_{\bfx_1,\cdots,\bfx_L}\in\Delta\paren{\cX^L}$, alternatively a tensor $P_{\bfx_1,\cdots,\bfx_L}\in\ten_\cardX^\tl$, define its \emph{asymmetry}  as 
\[\asymm(P_{\bfx_1,\cdots,\bfx_L}):=\max_{(x_1,\cdots,x_L)\in\cX^L}\max_{\sigma\in S_L\setminus\{\id\}}\abs{P_{\bfx_1,\cdots,\bfx_L}(x_1,\cdots,x_L)-P_{\bfx_1,\cdots,\bfx_L}(x_{\sigma(1)},\cdots,x_{\sigma(L)})}.\]
\end{definition}
\begin{remark}
If $\asymm(P_{\bfx_1,\cdots,\bfx_L}) = 0$, then $P_{\bfx_1,\cdots,\bfx_L}$ is symmetric in the sense of Definition \ref{def:symm_tensor}. 
\end{remark}

We will show that
\begin{lemma}[Converse, asymmetric case]
\label{lem:converse_asymm}
If $P_{\bfx_1,\cdots,\bfx_L}\in\ten_\cardX^\tl(P_\bfx)$ is asymmetric as a tensor in $\ten_\cardX^\tl(P_\bfx)$ and has asymmetry $\alpha$, then for any $0<\zeta<\alpha$, any $\paren{\zeta,P_{\bfx_1,\cdots,\bfx_L}}$-equicoupled (w.r.t. max-absolute-value distance)\footnote{Note that $\zeta$-equicoupledness w.r.t. sum-absolute-value distance implies $\zeta$-equicoupledness w.r.t. max-absolute-value distance. Hence this lemma directly applies to the subcode we obtained in the previous section.} code $\cC'$ has size at most
\[|\cC'|\le\exp\paren{\frac{c}{\alpha/\binom{L}{2}-\zeta}} + L-2\]
for some absolute constant $c>0$. 
\end{lemma}

Lemma \ref{lem:converse_asymm} is shown by reducing the problem, in a nontrivial way, from general values of $L$ to $L=2$ in which case it is known \cite{wang-budkuley-bogdanov-jaggi-2019-omniscient-avc} that such codes cannot be large.
\begin{lemma}[Reduction from general $L$ to $L=2$]
\label{lem:reduction}
If $P_{\bfx_1,\cdots,\bfx_L}\in\ten_\cardX^\tl$ has asymmetry $\asymm(P_{\bfx_1,\cdots,\bfx_L}) = \alpha$, then among the following distributions 
\[P_{\bfy_1,\bfz_1},\;P_{\bfy_2,\bfz_2},\cdots,P_{\bfy_{L-1},\bfz_{L-1}},\]
there is at least one distribution $P_{\bfy_{i^*},\bfz_{i^*}}$ ($i^*\in[L-1]$) with asymmetry at least  
\[\asymm\paren{P_{\bfy_{i^*},\bfz_{i^*}}} = \frac{\alpha}{\binom{L}{2}}.\] 
Here, for $i\in[L-1]$, $\bfy_i$ and $\bfz_i$ ($1\le i\le L-1$) are tuples of random variables defined as
\begin{align*}
    \begin{array}{rlllllllll}
        \bfy_i&\coloneqq& (\bfx_1,&\cdots,&\bfx_{i-1},&\bfx_i,& &\bfx_{i+2},&\cdots,&\bfx_L),\\
        \bfz_i&\coloneqq& (\bfx_1,&\cdots,&\bfx_{i-1},& & \bfx_{i+1},&\bfx_{i+2},&\cdots,&\bfx_L),
    \end{array}
\end{align*}
respectively. 
\end{lemma}
\begin{proof}
The proof is by contradiction. We will show that if all of $\curbrkt{P_{\bfy_i,\bfz_i}}_{1\le i\le L-1}$ have small asymmetry, then they do not not suffice to back propagate their asymmetry using transpositions to result in the  asymmetry $\alpha$ of $P_{\bfx_1,\cdots,\bfx_L}$. To make this intuition clear, assume, towards a contradiction, that  all of the  distributions $\curbrkt{P_{\bfy_i,\bfz_i}}_{1\le i\le L-1}$ have asymmetry strictly less than $\alpha'=\frac{\alpha}{\binom{L}{2}}$,
\begin{equation}
	\asymm\paren{ P_{\bfy_i,\bfz_i}}< \frac{\alpha}{\binom{L}{2}},\;\forall i\in[L-1].
	\label{eqn:asymm_assumption}
\end{equation}

Assume the asymmetry of $P_{\bfx_1,\cdots,\bfx_L}$ is witnessed by coordinates $(x_1,\cdots,x_L)\in\cX^L$ and permutation $\pi\in S_L$, i.e., 
\begin{align}
    \alpha=&\abs{P_{\bfx_1,\cdots,\bfx_L}(x_1,\cdots,x_L)-P_{\bfx_1,\cdots,\bfx_L}(x_{\pi(1)},\cdots,x_{\pi(L)})}\label{eqn:asymm_p_realized}\\
    =&\abs{P_{\bfx_1,\cdots,\bfx_L}(x_1,\cdots,x_L)-P_{\bfx_{\pi(1)},\cdots,\bfx_{\pi(L)}}(x_1,\cdots,x_L)}.\notag
\end{align}
Note that the set of transpositions $\curbrkt{\sigma_1,\cdots,\sigma_{L-1}}$ forms a generator set of $S_L$, where
\[\sigma_i\coloneqq\begin{pmatrix}
1 & \cdots & i - 1 & i & i + 1 & i + 2 & \cdots & L\\
1 & \cdots & i - 1 & i + 1 & i & i + 2 & \cdots & L
\end{pmatrix}.\]
Any permutation $\sigma\in S_L$ can be written as a product of $\sigma_i$'s, $\sigma=\sigma_{i_\ell}\cdots\sigma_{i_{1}}$ for some positive integer $\ell$ and a subset of transpositions, $i_j\in[L-1]$ for each $j\in[\ell]$. Such a representation, in particular the value of $\ell$, is not necessarily unique. Let 
\[\ell(\sigma)\coloneqq\min\curbrkt{\ell\in\bZ_{\ge0}\colon \sigma=\sigma_{i_{\ell}}\cdots\sigma_{i_1}\text{ transposition representation}}\] 
be the \emph{transposition length} of $\sigma$, i.e., the length of the shortest representation using product of transpositions. 
Let
\[\ell^*\coloneqq\max_{\sigma\in S_L}\ell(\sigma).\]
We claim that $\ell^*\le\binom{L}{2}$.  To see this, it suffices to bound $\ell(\sigma)$ for the worst case permutation 
\[\sigma=\begin{pmatrix}
1 & 2 & \cdots & L\\
L & L-1 & \cdots & 1
\end{pmatrix}.\]
The claim follows by noting that $\sigma$ can be written as 
\begin{equation}
	\sigma=\prod_{j=1}^{L-1}\prod_{i=j,j-1,\cdots,1}\sigma_i,
	\label{eqn:sigma_repr}
\end{equation}
which contains $\binom{L}{2}$ transpositions. 
\begin{remark}
A potential confusion may arise from two conflicting conventions that
\begin{enumerate}
	\item a product is usually written from left to right, i.e., 
	\[\prod_{i = 1}^\ell\sigma_i = \sigma_{1}\cdots\sigma_{\ell};\]
	\item a composition of permutations acts like functions on an element from right to left, i.e., for $\sigma,\pi\in S_L$ and $i\in[L]$,
	\[(\sigma\pi)(i) = \sigma(\pi(i)).\]
\end{enumerate}
With this kept in mind, the representation in Eqn. \eqref{eqn:sigma_repr} should be understood as
\begin{align*}
	\sigma = &(\sigma_1)(\sigma_2\sigma_1)\cdots(\sigma_{L-2}\cdots\sigma_2\sigma_{1})(\sigma_{L-1}\cdots\sigma_2\sigma_1).
\end{align*}
The product in the $(L-1)$-st parenthesis (from left to right) moves $L$ in the initial sequence $(L,L-1,\cdots,1)$ to the $L$-th position; the product in the $(L-2)$-nd parenthesis moves $L-1$ to the $(L-1)$-st position; ...; the permutation $\sigma_1$ in the 1-st parenthesis  moves 2 to the 2-st position, and automatically 1 is in the 1-st position. We get the target sequence $(1,2,\cdots,L)$.
\end{remark}
We can write
\begin{equation}
	\pi=\prod_{j=\ell,\ell-1,\cdots,1}\sigma_{i_j},
	\label{eqn:def_pi_trans_repr}
\end{equation}
for some $\ell\le\ell^*\le \binom{L}{2}$.

Our assumption Eqn. \eqref{eqn:asymm_assumption} implies that, for any $(x_1,\cdots,x_L)\in\cX^L$ and any transposition $\sigma_i$,
\begin{align}
    &\abs{P_{\bfx_1,\cdots,\bfx_{L}}(x_1,\cdots,x_L) - P_{\bfx_{\sigma_i(1)},\cdots,\bfx_{\sigma_i(L)}}(x_1,\cdots,x_L)}\notag\\
    =&\abs{P_{\bfx_1,\cdots,\bfx_{i-1},\bfx_{i},\bfx_{i+1},\bfx_{i+2},\cdots,\bfx_L}(x_1,\cdots,x_L) - P_{\bfx_1,\cdots,\bfx_{i-1},\bfx_{i+1},\bfx_i,\bfx_{i+2},\cdots,\bfx_L}(x_1,\cdots,x_L)}\notag\\
    =&|P_{(\bfx_1,\cdots,\bfx_{i-1},\bfx_i,\bfx_{i+2},\cdots,\bfx_L),(\bfx_1,\cdots,\bfx_{i-1},\bfx_{i+1},\bfx_{i+2},\cdots,\bfx_L)}((x_1,\cdots,x_{i-1},x_i,x_{i+2},\cdots,x_L),(x_1,\cdots,x_{i-1},x_{i+1},x_{i+2},\cdots,x_L))\notag\\
    &-P_{(\bfx_1,\cdots,\bfx_{i-1},\bfx_{i+1},\bfx_{i+2},\cdots,\bfx_L),(\bfx_1,\cdots,\bfx_{i-1},\bfx_i,\bfx_{i+2},\cdots,\bfx_L)}((x_1,\cdots,x_{i-1},x_i,x_{i+2},\cdots,x_L),(x_1,\cdots,x_{i-1},x_{i+1},x_{i+2},\cdots,x_L))|\notag\\
    =&\abs{P_{\bfy_i,\bfz_i}(y,z) - P_{\bfz_i,\bfy_i}(y,z)}\notag\\
    =&\abs{P_{\bfy_i,\bfz_i}(y,z) - P_{\bfy_i,\bfz_i}(z,y)}\notag\\
    <&\alpha',\label{eqn:alpha_prime_cor}
\end{align}
where
\begin{align*}
	\begin{array}{rlllllllll}
    y&\coloneqq&(x_1,&\cdots,&x_{i-1},&x_i,&&x_{i+2},&\cdots,&x_L),\\
    z&\coloneqq&(x_1,&\cdots,&x_{i-1},&&x_{i+1},&x_{i+2},&\cdots,&x_L).
    \end{array}
\end{align*}

Now
\begin{align}
    \alpha=&\abs{P_{\bfx_1,\cdots,\bfx_L}(x_1,\cdots,x_L)-P_{\bfx_{\pi(1)},\cdots,\bfx_{\pi(L)}}(x_1,\cdots,x_L)}\label{eqn:alpha_assumption}\\
    \le&\abs{P_{\bfx_1,\cdots,\bfx_L}(x_1,\cdots,x_L)-P_{\bfx_{\sigma_{i_1}(1)},\cdots,\bfx_{\sigma_{i_1}(L)}}(x_1,\cdots,x_L)}\notag\\
    &+ \abs{P_{\bfx_{\sigma_{i_1}(1)},\cdots,\bfx_{\sigma_{i_1}(L)}}(x_1,\cdots,x_L) - P_{\bfx_{\pi(1)},\cdots,\bfx_{\pi(L)}}(x_1,\cdots,x_L)}\label{eqn:triangle_1}\\
    <&\alpha' + \abs{P_{\bfx_{\sigma_{i_1}(1)},\cdots,\bfx_{\sigma_{i_1}(L)}}(x_1,\cdots,x_L) - P_{\bfx_{\pi(1)},\cdots,\bfx_{\pi(L)}}(x_1,\cdots,x_L)}\label{eqn:alpha_prime_cor_1}\\
    \le & \alpha' + \abs{P_{\bfx_{\sigma_{i_1}(1)},\cdots,\bfx_{\sigma_{i_1}(L)}}(x_1,\cdots,x_L) - P_{\bfx_{\sigma_{i_2}\sigma_{i_1}(1)},\cdots,\bfx_{\sigma_{i_2}\sigma_{i_1}(L)}}(x_1,\cdots,x_L)}\notag\\
     & + \abs{P_{\bfx_{\sigma_{i_2}\sigma_{i_1}(1)},\cdots,\bfx_{\sigma_{i_2}\sigma_{i_1}(L)}}(x_1,\cdots,x_L) - P_{\bfx_{\pi(1)},\cdots,\bfx_{\pi(L)}}(x_1,\cdots,x_L)}\label{eqn:triangle_2}\\
     <&2\alpha' + \abs{P_{\bfx_{\sigma_{i_2}\sigma_{i_1}(1)},\cdots,\bfx_{\sigma_{i_2}\sigma_{i_1}(L)}}(x_1,\cdots,x_L) - P_{\bfx_{\pi(1)},\cdots,\bfx_{\pi(L)}}(x_1,\cdots,x_L)}\label{eqn:alpha_prime_cor_2}\\
     &\cdots\notag\\
     \le&(\ell-1)\alpha' + \abs{P_{\bfx_{\sigma_{i_{\ell-1}}\cdots\sigma_{i_1}(1)},\cdots,\bfx_{\sigma_{i_{\ell - 1}}\cdots\sigma_{i_1}(L)}}(x_1,\cdots,x_L) - P_{\bfx_{\pi(1)},\cdots,\bfx_{\pi(L)}}(x_1,\cdots,x_L)}\label{eqn:recursion}\\
     =&(\ell-1)\alpha' + \abs{P_{\bfx_{\sigma_{i_{\ell-1}}\cdots\sigma_{i_1}(1)},\cdots,\bfx_{\sigma_{i_{\ell - 1}}\cdots\sigma_{i_1}(L)}}(x_1,\cdots,x_L) - P_{\bfx_{\sigma_{i_{\ell}}\sigma_{i_{\ell-1}}\cdots\sigma_{i_1}(1)},\cdots,\bfx_{\sigma_{i_{\ell}}\sigma_{i_{\ell-1}}\cdots\sigma_{i_1}(L)}}(x_1,\cdots,x_L)}\label{eqn:pi_trans_repr}\\
     <&\ell\alpha'\label{eqn:alpha_prime_cor_3}\\
     \le&\binom{L}{2}\alpha'\notag\\
     =&\alpha.\label{eqn:def_alpha_prime}
\end{align}
\begin{enumerate}
	\item Eqn. \eqref{eqn:alpha_assumption} follows from Eqn. \eqref{eqn:asymm_p_realized}. 
	\item Eqn. \eqref{eqn:triangle_1}, \eqref{eqn:triangle_2}, etc. are by triangle inequality. 
	\item Eqn. \eqref{eqn:alpha_prime_cor_1}, \eqref{eqn:alpha_prime_cor_2}, \eqref{eqn:alpha_prime_cor_3}, etc. are by Eqn. \eqref{eqn:alpha_prime_cor}. 
	\item Eqn. \eqref{eqn:recursion} is by recursively applying the previous calculations. 
	\item Eqn. \eqref{eqn:pi_trans_repr} is by the transposition representation of $\pi$ (Eqn. \eqref{eqn:def_pi_trans_repr}). 
	\item Eqn. \eqref{eqn:def_alpha_prime} is by the choice of $\alpha'$.
\end{enumerate}

We reach a  contradiction that $\alpha$ is strictly less than itself. This finishes the proof. 
\end{proof}

Next, we   show the key lemma \ref{lem:converse_asymm} in this section. Note that, according to the statement, Lemma \ref{lem:converse_asymm} is independent of the channel that the code $\cC'$ is used for. Hence we will directly prove the random variable version of this lemma which is concerned with fundamental properties of joint distributions.  If the joint distribution of a sequence of random variables has all of its size-$L$ marginals being $\zeta$-close to some \emph{asymmetric} distribution, then such a sequence cannot be infinitely long. We will prove a finite upper bound on the length of the sequence by reducing this problem from the general $L>2$ case to the $L=2$ case. In the $L=2$ case, prior work \cite{wang-budkuley-bogdanov-jaggi-2019-omniscient-avc} shows that this is indeed the case.
\begin{lemma}[Converse, asymmetric case, $L=2$ \cite{wang-budkuley-bogdanov-jaggi-2019-omniscient-avc}]
\label{lem:asymm_l_equals_two}
Assume $P_{\bfx_1,\bfx_2}\in\Delta(\cX^2)$ has asymmetry $\asymm\paren{P_{\bfx_1,\bfx_2}}=\alpha$. Let $\bfw_1,\cdots,\bfw_M$ be a sequence of $M$ random variables supported on $\cX$ such that for every $1\le j_1<j_2\le M$, 
\[\normmav{P_{\bfw_{j_1},\bfw_{j_2}} - P_{\bfx_1,\bfx_2}}\le\zeta. \]
for some $0<\zeta<\alpha$. Then 
\[M\le\exp\paren{\frac{c}{\alpha - \zeta}}\]
for some universal constant $c>0$.
\end{lemma}

We are now ready to prove the restated version of Lemma \ref{lem:converse_asymm}.
\begin{lemma}[Converse, asymmetric case, general $L$]
\label{lem:asymm_general_l}
If a joint distribution $P_{\bfx_1,\cdots,\bfx_L}\in\Delta\paren{\cX^L}$ has asymmetry $\asymm\paren{P_{\bfx_1,\cdots,\bfx_L}} = \alpha$, and a sequence of $M$ random variables $\bfw_1,\cdots,\bfw_M$ supported on $\cX$ satisfies that for any $1\le j_1<\cdots<j_L\le M$, 
\begin{equation}
    \normmav{P_{\bfw_{j_1},\cdots,\bfw_{j_L}} - P_{\bfx_1,\cdots,\bfx_L} }\le\zeta.
    \label{eqn:equicoupled_assumption_l}
\end{equation}
Then 
\[M\le\exp\paren{\frac{c}{\alpha/\binom{L}{2}-\zeta}} + L-2\]
for some universal constant $c>0$.
\end{lemma}
\begin{proof}
Construct the following $L-1$ sequences $\curbrkt{\bfv^{(i)}}_{1\le i\le L-1}$ of random variables, each of which has length $M-L+2$,
\begin{align*}
    \begin{array}{llllll}
        \bfv^{(1)} &= & (\bfv_1^{(1)}, & \bfv_2^{(1)}, &\cdots, &\bfv_{M-L+2}^{(1)} ),  \\
        \bfv^{(2)} &= & (\bfv_2^{(2)}, & \bfv_2^{(2)}, &\cdots, &\bfv_{M-L+3}^{(2)} ),  \\
        \cdots & & & & & \\
        \bfv^{(L-1)} &= & (\bfv_{L-1}^{(L-1)}, & \bfv_2^{(1)}, &\cdots, &\bfv_{M}^{(L-1)} ).
    \end{array}
\end{align*}
For $1\le i\le L-1$ and $i\le j\le M-L+i+1$, $\bfv_j^{(i)}$ is defined as atuple
\[\bfv_j^{(i)}\coloneqq\paren{\bfw_1,\cdots,\bfw_{i-1},\bfw_j,\bfw_{M-L+i+2},\cdots,\bfw_M}.\]
Then, for any 
\begin{align*}
    \begin{array}{rlllllllll}
        v_1&\coloneqq&(x_1,&\cdots,&x_{i-1},&x_i,&&x_{i+2},&\cdots,&x_L)\in\cX^{L-1},\\
        v_2&\coloneqq&(x_1,&\cdots,&x_{i-1},&&x_{i+1},&x_{i+2},&\cdots,&x_L)\in\cX^{L-1},  \\
    \end{array}
\end{align*}
and $i\le j_1 < j_2\le M-L+i+1$, we have
\begin{align*}
    &\abs{P_{\bfv_{j_1}^{(i)},\bfv_{j_2}^{(i)}}(v_1,v_2) - P_{\bfy_i,\bfz_i}(v_1,v_2)}\\
    =&\left|P_{(\bfw_1,\cdots,\bfw_{i-1},\bfw_{j_1},\bfw_{M-L+i+2},\cdots,\bfw_M),(\bfw_1,\cdots,\bfw_{i-1},\bfw_{j_2},\bfw_{M-L+i+2},\cdots,\bfw_M)}\binom{(x_1,\cdots,x_{i-1},x_i,x_{i+2},\cdots,x_L),}{(x_1,\cdots,x_{i-1},x_{i+1},x_{i+2},\cdots,x_L)}\right.\\
    &\left.-P_{(\bfx_1,\cdots,\bfx_{i-1},\bfx_i,\bfx_{i+2},\cdots,\bfx_L),(\bfx_1,\cdots,\bfx_{i-1},\bfx_{i+1},\bfx_{i+2},\cdots,\bfx_L)}\binom{(x_1,\cdots,x_{i-1},x_i,x_{i+2},\cdots,x_L),}{(x_1,\cdots,x_{i-1},x_{i+1},x_{i+2},\cdots,x_L)}\right|\\
    =&\abs{P_{\bfw_1,\cdots,\bfw_{i-1},\bfw_{j_1},\bfw_{j_2},\bfw_{M-L+i+2},\cdots,\bfw_M}(x_1,\cdots,x_{i-1},x_{i},x_{i+1},x_{i+2},\cdots,x_L) - P_{\bfx_1,\cdots,\bfx_L}(x_1,\cdots,x_L)}\\
    \le&\zeta,
\end{align*}
by the assumption Eqn. \eqref{eqn:equicoupled_assumption_l}.
Therefore, all sequences $\bfv^{(i)}$'s are $(\zeta,P_{\bfy_i,\bfz_i})$-equicoupled, $1\le i\le L-1$. 

Since $P_{\bfx_1,\cdots,\bfx_L}$ is $\alpha$-asymmetric, by Lemma \ref{lem:reduction}, at least one of the distributions $P_{\bfy_i,\bfz_i}$'s ($1\le i\le L-1$) is at least $\alpha'$-asymmetric ($\alpha'={\alpha}/{\binom{L}{2}}$). Without loss of generality, assume $P_{\bfy_{i_0},\bfz_{i_0}}$ is $\ge\alpha'$-asymmetric. Then the $i_0$-th sequence $\bfv^{(i_0)}$ is short by Lemma \ref{lem:asymm_l_equals_two},
\[M-L+2\le \exp\paren{\frac{c}{\alpha'-\zeta}},\]
for some universal constant $c>0$. Hence \[M\le \exp\paren{\frac{c}{\alpha/\binom{L}{2}-\zeta}} + L-2,\]
which finishes the proof. 
\end{proof}

\begin{remark}[Asymmetric but projectively symmetric tensors]
Lemma \ref{lem:reduction} does not follow from na\"ively marginalizing an asymmetric distribution $P_{\bfx_1,\cdots,\bfx_L}$ and hoping that $P_{\bfx_i,\bfx_j}$ is asymmetric for some $1\le i<j\le L$. Just like there exist asymmetric matrices (self-couplings) with the same column sum and row sum, we should not expect that the asymmetry of a tensor is preserved under projections.

We say that a tensor $P_{\bfx_1,\cdots,\bfx_L}\in\ten_\cardX^\tl$ is \emph{$\ell$-projectively symmetric} ($1\le \ell<L$) if all of its order-$\ell$ projections are symmetric, i.e., for any $1\le i_1<\cdots<i_\ell\le L$,
\[P_{\bfx_{i_1},\cdots,\bfx_{i_\ell}}\coloneqq\sqrbrkt{P_{\bfx_1,\cdots,\bfx_L}}_{\bfx_{i_1},\cdots,\bfx_{i_\ell}}\in\ten_\cardX^{\otimes\ell}\]
is symmetric.

One can easily verify the following facts.
\begin{lemma}
\begin{enumerate} 
	Let $P_{\bfx_1,\cdots,\bfx_L}$ be a tensor of dimension $\cardX$ and order $L$.
	\item If $P_{\bfx_1,\cdots,\bfx_L}$ is $\ell$-projectively symmetric ($1\le\ell<L$), then all of its order-$\ell'$ ($1\le\ell'<\ell$) marginals are the same.
	\item If $P_{\bfx_1,\cdots,\bfx_L}$ is $\ell$-projectively symmetric ($1\le\ell<L$), then it is also $\ell'$-projectively symmetric for any $1\le\ell'<\ell$.
	\item A symmetric tensor $P_{\bfx_1,\cdots,\bfx_L}$ is also $\ell$-projectively symmetric for all $1\le\ell<L$. In particular, it is a self-coupling, i.e., $P_{\bfx_i}$ is the same for all $i\in[L]$.
\end{enumerate}
\end{lemma}

We provide an example showing that the asymmetry of a tensor cannot be recovered from all of its lower order projections. That is, there is an asymmetric tensor with every projection of one less order being symmetric.

We now construct a concrete example. In order for a dimension-2 order-3 tensor $T\colon[2]^3\to\bR$  to be symmetric, it has to satisfy the following system $\cE_1$ of linear equations, 
\begin{align*}
    t_{112} =& t_{121},\quad
    t_{121} = t_{211},\quad
    t_{212} = t_{122},\quad
    t_{122} = t_{221}.
\end{align*}
where $t_{ijk}\coloneqq T(i,j,k)$ for $i,j,k\in[2]$. 
On the other hand, for it to be projectively symmetric, it has to satisfy the following system $\cE_2$ of linear equations,
\begin{align*}
    t_{122} + t_{121} = & t_{212} + t_{211},\\
    t_{112} + t_{122} = & t_{211} + t_{221},\\
    t_{121} + t_{221} = & t_{112} + t_{212}.
\end{align*}
Additionally, for $T$ to represent a joint distribution, all entries should be non-negative and sum up to one. 
Note that $\cE_2$ is  a \emph{less determined} system  than $\cE_1$, which means that we should be able to find a solution to $\cE_2$ which does not satisfy $\cE_1$.

Indeed, consider the following explicit example of $T\in\ten_2^{\otimes3}$. (See Fig. \ref{fig:asymm_proj_symm_ten}.)
\begin{align*}
    t_{111} = & \frac{1}{60},\quad
    t_{121} =  \frac{1}{4},\quad
    t_{112} =  \frac{1}{6},\quad
    t_{122} =  \frac{1}{20},\\
    t_{211} = & \frac{1}{60},\quad
    t_{221} =  \frac{1}{5},\quad
    t_{212} =  \frac{17}{60},\quad
    t_{222} =  \frac{1}{60}.
\end{align*}
\begin{figure}
    \centering
    \includegraphics[scale = 2]{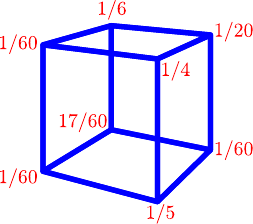}
    \caption{An asymmetric tensor $T\in\ten_2^{\otimes3}$ that is 2-projectively symmetric.}
    \label{fig:asymm_proj_symm_ten}
\end{figure}
It is asymmetric but projectively symmetric. Note that $T$ is forced to have multiple witnesses of asymmetry due to its projective symmetry. 
Indeed,
\begin{align*}
    t_{121} - t_{112} = & t_{212} - t_{221} = \frac{5}{60},\\
    t_{121} - t_{211} = & t_{212} - t_{122} = \frac{14}{60},\\
    t_{112} - t_{211} = & t_{221} - t_{122} = \frac{9}{60}.
\end{align*}
Therefore $\asymm(T)=\frac{14}{60}=\frac{7}{30}$, given by $t_{121} - t_{211}$ and $t_{212} - t_{122}$. All of its order-2 projections are given by
\begin{align*}
    \begin{bmatrix}
    \frac{11}{60}& \frac{3}{10}\\
    \frac{3}{10}&\frac{13}{60}
    \end{bmatrix},\quad
    \begin{bmatrix}
    \frac{4}{15} & \frac{13}{60}\\
    \frac{13}{60} & \frac{3}{10}
    \end{bmatrix},\quad
    \begin{bmatrix}
    \frac{1}{30} & \frac{9}{20}\\
    \frac{9}{20} & \frac{1}{15}
    \end{bmatrix}.
\end{align*}
All of their margins are equal to $\begin{bmatrix}\frac{29}{60} \\ \frac{31}{60}\end{bmatrix}$.

In general, for any dimension-$d$ order-$L$ tensor, such examples can always be constructed due to the gap of degrees of freedom between the homogeneous linear systems $\cE_1$ and $\cE_2$.
\end{remark}

\section{Rethinking the converse}
\label{sec:rethinking_converse}
\subsection{A cheap converse}
If for a general $\cA = (\cX,\lambda_\bfx,\cS,\lambda_\bfs,W_{\bfy|\bfx,\bfs})$, for every $P_\bfx\in\lambda_\bfx$, the confusability set is a halfspace defined by a single linear constraint
\[\cK^\tl(P_\bfx)\coloneqq\curbrkt{P_{\bfx_1,\cdots,\bfx_L}\in\cJ^\tl(P_\bfx)\colon\inprod{P_{\bfx_1,\cdots,\bfx_L}}{C}\le b},\]
for some tensor $C\in\ten_\cardX^\tl$ and constant $b$, then the converse can be significantly simplified. In particular, we do not have to handle symmetric and asymmetric cases separately. We describe the proof idea below.
\begin{proof}
The proof essentially follow from the following  observation. For any asymmetric $P_{\bfx_1,\cdots,\bfx_L}$, given any $P_{\bfx}$-constant composition  $(\zeta,P_{\bfx_1,\cdots,\bfx_L})$-equicoupled  code $\cC  = \curbrkt{\vx_i}_{i=1}^M$ in $\cX^n$ of size $M$, we can construct a  code $\cC' = \curbrkt{\vx_i'}_{i=1}^M$ in $\cX^{n\cdot M!}$ of the same size which is \emph{symmetric}. Indeed, we can permute the rows of $\cC$ using $\sigma\in S_{M}$ and juxtapose all possible ($M!$ of them in total) such row-permuted codes $\sigma(\cC)$. (See Fig. \ref{fig:cheap_converse}.) 
\begin{figure}
    \centering
    \includegraphics{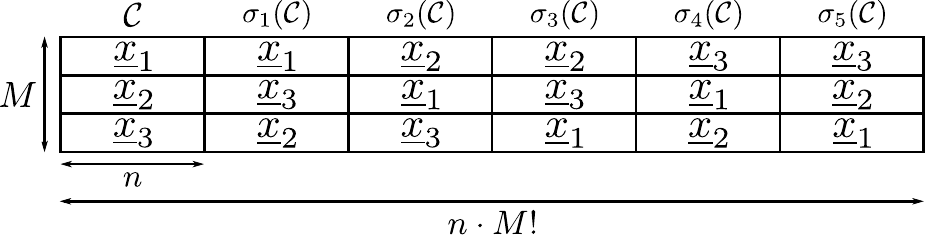}
    \caption{Construction of $\cC'$ by permuting rows of $\cC = \curbrkt{\vx_1,\vx_2,\vx_3}$ using $\sigma\in S_3$ (where $S_3 = \curbrkt{\id,\sigma_1,\cdots,\sigma_5}$) and juxtaposing  all $\sigma(\cC)$ (6 of them in total) together.}
    \label{fig:cheap_converse}
\end{figure}
The resulting code $\cC'$ is actually not only $L$-wise approximately equicoupled, but $M$-wise exactly equicoupled! For \emph{any} $L\in[M]$ and any $L$-sized (not necessarily ordered) subset $\{i_1,\cdots,i_L\}$ of $[M]$, the joint type of $\vx'_{i_1},\cdots,\vx'_{i_L}$ is \emph{exactly} equal to
\begin{align*}
    \tau_{\vx'_{i_1},\cdots,\vx'_{i_L}}=\frac{1}{\binom{M}{L}}\sum_{\{i_1,\cdots,i_L\}\in\binom{[M]}{L}}\frac{1}{L!}\sum_{\sigma\in S_L}\tau_{\vx_{\sigma(i_1)},\cdots,\vx_{\sigma(i_L)}},
\end{align*}
which is \emph{symmetric} and independent of the choice of the list $(i_1,\cdots,i_L)$ (hence let us denote it by $\wh P_{\bfx_1,\cdots,\bfx_L}$).   In particular, letting $L=M$, we get that
\begin{align*}
    \tau_{\vx_1',\cdots,\vx_M'}=\frac{1}{M!}\sum_{\sigma\in S_M}\tau_{\vx_{\sigma(1)},\cdots,\vx_{\sigma(M)}}.
\end{align*}
To see the above claims,
note that if we juxtapose two pairs of codewords $(\vx_1,\vx_2)$ and $(\vx_1',\vx_2')$, we get a pair of longer codewords $(\wt\vx_1,\wt\vx_2)\coloneqq(\vx_1\circ\vx_1',\vx_2\circ\vx_2')$ (where $\circ$ denotes concatenation) with joint type
\begin{align*}
    \tau_{\wt\vx_1,\wt\vx_2}=\frac{1}{2}(\tau_{\vx_1,\vx_2}+\tau_{\vx_1',\vx_2'}).
\end{align*}
This still holds if two pairs of codewords of different blocklengths are juxtaposed. Say, $(\vx_1,\vx_2)$ has blocklength $n$ while $(\vx_1',\vx_2')$ has blocklength $n'$. Then
\[\tau_{\wt\vx_1,\wt\vx_2}=\frac{n}{n+n'}\tau_{\vx_1,\vx_2}+\frac{n'}{n+n'}\tau_{\vx_1',\vx_2'}.\]

Back to the proof of the converse in such a spacial case, since the confusability set is defined by a single linear constraint, any convex combinations of non-confusable joint types is still outside the confusability set, in particular, $\wh P_{\bfx_1,\cdots,\bfx_L}$. We hence reduce the problem to the symmetric case and the rest of the proof is handled by Theorem \ref{lem:converse_symm}.
\end{proof}

\subsection{Towards a unifying converse}
We feel it unusual that we have to use drastically different techniques to prove the symmetric and the asymmetric parts of the converse. We suspect that it can be proved in a unifying way using the duality between $\cp$ and $\cop$ tensors which is the  source of contradiction in our current proof of the symmetric case. 

Note that the duality holds only in the space of \emph{symmetric} tensors. To be specific, traditionally, $\cp$ and $\cop$ tensors are defined to be symmetric. And they are dual cones living in the ambient space $\sym_n^{\otimes}$. If we extend the definitions of $\cp$ and $\cop$ tensors to the set of \emph{all} (including asymmetric) tensors, then it is unclear whether duality still holds. Indeed, there are pairs of cones which are dual to each other in a certain ambient space but are no long dual in a larger ambient space. In a word, the ambient space that the dual cone is computed with respect to matters much. 

We provide evidence showing that the symmetric and asymmetric parts of the converse can be potentially unified by the Plotkin-type bound since duality between $\cp$ and $\cop$ tensors--the core of the double counting argument--fortunately holds in larger generality.

\noindent\textbf{Duality.}
We know that $\cp_\cardX^\tl$ and $\cop_\cardX^\tl$ are dual cones in the space $\sym_\cardX^\tl$ of \emph{symmetric} tensors. However, $\wh P_{\bfx_1,\cdots,\bfx_L}$ (associated to the equicoupled subcode extracted using hypergraph Ramsey's theorem) is not guaranteed to be  symmetric. We claim that duality still holds in the space $\ten_{|\cX|}^\tl$ of \emph{all} tensors. Hence, copositive witness $Q$ of a non-$\cp$ $\wh P_{\bfx_1,\cdots,\bfx_L}$ exists even when $\wh P_{\bfx_1,\cdots,\bfx_L}$ is asymmetric.
\begin{claim}
$\cp_{|\cX|}^\tl$ and $\cop_{|\cX|}^\tl$ are dual cones in $\ten_{|\cX|}^\tl$.
\end{claim}
\begin{proof}
By definition, 
\begin{align*}
    \paren{\cp_{|\cX|}^\tl}^*\coloneqq\curbrkt{B\in\ten_{|\cX|}^\tl\colon \forall A\in\cp_{|\cX|}^\tl,\;\inprod{A}{B}\ge0}.
\end{align*}
Note that it is important that $B$ is now taken from $\ten_{|\cX|}^\tl$ rather than $\sym_{|\cX|}^\tl$. Also recall that
\begin{align*}
    \cop_{|\cX|}^\tl\coloneqq\curbrkt{B\in\ten_{|\cX|}^\tl\colon \forall\vx\in{\bR_{\ge0}^{|\cX|}},\;\inprod{B}{\vx^\tl}\ge0}. 
\end{align*}
Note that this definition \emph{differs} from the standard one \ref{def:cop_tensor} and this cone is potentially larger.\footnote{Indeed, we will see shortly that it is strictly larger.}
The goal is to show $\paren{\cp_\cardX^\tl}^* = \cop_\cardX^\tl$.

The direction $\cop_{|\cX|}\subseteq\paren{\cp_{|\cX|}^\tl}^*$ is trivial, since the definitions of $\cp$ and $\cop$ tensors remain the same but the dual cone is computed w.r.t. a larger space. The new dual cone we are considering is no smaller than the old one. The inclusion that used to hold in the traditional setting should continue to hold now. Indeed, take any $B\in\cop_{|\cX|}^\tl$, for any $A=\sum_i\vx_i^\tl\in\cp_{|\cX|}^\tl$, where $\vx_i\in\bR_{\ge0}^{|\cX|}$,
\begin{align*}
    \inprod{A}{B}=&\inprod{\sum_i\vx_i^\tl}{B}=\sum_i\inprod{B}{\vx_i^\tl}.
\end{align*}
Since $B\in\cop_{|\cX|}^\tl$, by definition, all $\inprod{B}{\vx_i^\tl}$'s  are non-negative, hence so is $\inprod{A}{B}$. Therefore $B\in\paren{\cp_\cardX^\tl}^*$.
    
Now we show $\paren{\cp_{|\cX|}^\tl}^*\subseteq\cop_{|\cX|}^\tl$. Take any $B\in\paren{\cp_{|\cX|}^\tl}^*$ and any $\vx\in\bR_{\ge0}^{|\cX|}$. Then
    $\inprod{B}{\vx^\tl}\ge0$,
since $\vx^\tl\in\cp_{|\cX|}^\tl$ and $B\in\paren{\cp_{|\cX|}^\tl}^*$. This finishes the whole proof.
\end{proof}

\begin{remark}
In general, duality does not necessarily hold in a larger ambient space. Namely, computing dual cone w.r.t. a larger space may result in a larger cone. For instance, $\psd_{|\cX|}$ cone is known to be self dual in $\sym_{|\cX|}$, i.e., $\psd_{|\cX|}^*=\psd_{|\cX|}$. However, in $\mat_{|\cX|}$, $\psd_{|\cX|}^*$ is strictly containing $\psd_{|\cX|}$. To see this, note that any skew symmetric matrix $B$ is in $\psd_{|\cX|}^*$ since for any $\psd$ (hence symmetric) matrix $A$, $\inprod{A}{B}=0\ge0$; while $B$ is not necessarily $\psd$.
\end{remark}

Define, for $\sigma\in S_L$, $\sigma\paren{P_{\bfx_1,\cdots,\bfx_L}} \coloneqq P_{\bfx_{\sigma(1)},\cdots,\bfx_{\sigma(L)}}$. Though duality holds for all symmetric and asymmetric tensors, we do not have a full proof of the converse using duality, since we have trouble bounding the term
\begin{align*}
    \inprod{\sigma\paren{P_{\bfx_1,\cdots,\bfx_L}}}{Q}=&\inprod{P_{\bfx_1,\cdots,\bfx_L}}{\sigma(Q)}
\end{align*}
which does not necessarily equal $\inprod{P_{\bfx_1,\cdots,\bfx_L}}{Q}$ for \emph{asymmetric} $Q$. 

We next show that such asymmetric witness $Q$ does exist and is sometimes necessary in the sense that, some asymmetric (hence non-$\cp$) tensors have \emph{no} symmetric witness. This means that the dual cone of $\cop$ w.r.t. $\ten_\cardX^\tl$ (instead of $\sym_\cardX^\tl$) is strictly larger.

\noindent\textbf{Asymmetric distributions without symmetric $\cop$ witness.}
Let $L=2$. We construct an {asymmetric} self-coupling $P_{\bfx_1,\bfx_2}\in\Delta\paren{[3]^2}$ without symmetric $\cop$ witness $Q$ such that $\inprod{P_{\bfx_1,\bfx_2}}{Q}<0$.
Indeed, let 
\[P_{\bfx_1,\bfx_2}=\begin{bmatrix}
\frac{4}{9}&\frac{7}{48}&\frac{11}{144}\\
\frac{3}{16}& \frac{1}{16}&0\\
\frac{5}{144}&\frac{1}{24}&\frac{1}{144}
\end{bmatrix}.\]
Note that 
\[P_{\bfx_1}=P_{\bfx_2}=\begin{bmatrix}\frac{2}{3}\\\frac{1}{4}\\\frac{1}{12}\end{bmatrix}\eqqcolon P_{\bfx}.\]
Then 
\[\frac{P_{\bfx_1,\bfx_2}+P_{\bfx_1,\bfx_2}^\top}{2}=\begin{bmatrix}
\frac{4}{9}&\frac{1}{6}&\frac{1}{18}\\
\frac{1}{6}&\frac{1}{16}&\frac{1}{48}\\
\frac{1}{18}&\frac{1}{48}&\frac{1}{144}
\end{bmatrix}=\begin{bmatrix}\frac{2}{3}\\\frac{1}{4}\\\frac{1}{12}\end{bmatrix}\begin{bmatrix}\frac{2}{3}&\frac{1}{4}&\frac{1}{12}\end{bmatrix}= P_{\bfx}P_{\bfx}^\top.\]
If there was a symmetric $\cop$ $Q$ such that $\inprod{P_{\bfx_1,\bfx_2}}{Q}<0$, then 
\begin{align*}
    \inprod{P_{\bfx}P_{\bfx}^\top}{Q}=&\frac{1}{2}\paren{\inprod{P_{\bfx_1,\bfx_2}}{Q} + \inprod{P_{\bfx_1,\bfx_2}^\top}{Q}}\\
    =&\frac{1}{2}\paren{\inprod{P_{\bfx_1,\bfx_2}}{Q}+\inprod{P_{\bfx_1,\bfx_2}}{Q^\top}}\\
    =&\inprod{P_{\bfx_1,\bfx_2}}{Q}<0.
\end{align*}
However, $P_{\bfx}P_{\bfx}^\top$ is $\cp$, so $\inprod{P_{\bfx}P_{\bfx}^\top}{Q}\ge0$, which is a contradiction. 

\section{Sanity checks}
\label{sec:sanity_checks}

Consider the bit-flip model.

In this section, we are going to verify the correctness of our characterization of the generalized Plotkin point using the bit-flip model as a running example. For $L=3,4$,\footnote{For $L=2$, i.e., the unique decoding case, the work \cite{wang-budkuley-bogdanov-jaggi-2019-omniscient-avc} already recovers the classic Plotkin bound $P_1 = 1/4$.} we will numerically recover Blinovsky's \cite{blinovsky-1986-ls-lb-binary} characterization of the Plotkin point $P_{L-1}$ for $(p,L-1)$-list decoding. In particular, $P_2 = 1/4$ and $P_3= 5/16$.

\subsection{$L=3$}
We first consider $(L-1)$-list decoding for $L-1=2$, i.e., $L=3$. It is known that the Plotkin point at $L-1=2$ is $P_2=1/4$.

Fix any input distribution $P_\bfx\coloneqq \bern(w)=\inputdistr{ w }$ for $0< w <1$. We first compute $\cJ^{\otimes 3}\paren{P_\bfx}$, $\cK^{\otimes 3}\paren{P_\bfx}$. Let $p_{i,j,k,\ell}\coloneqq P_{\bfx_1,\bfx_2,\bfx_3,\bfy}(i,j,k,\ell)$ where $i,j,k,\ell\in \curbrkt{0,1}$.
\begin{align*}
    \cJ^{\otimes 3}\paren{P_\bfx}=&\curbrkt{ P_{\bfx_1,\bfx_2,\bfx_3}\in\Delta\paren{\curbrkt{0,1}^3}\colon P_{\bfx_i}=P_\bfx,\;i=1,2,3 }\\
    =&\curbrkt{P_{\bfx_1,\bfx_2,\bfx_3}\colon \begin{array}{rl}
        p_{i,j,k}\ge&0,\; i,j,k\in \curbrkt{0,1}  \\
        \sum_{i,j,k}p_{i,j,k}=&1\\
        \sum_{i,j}p_{i,j,1}=& w \\
        \sum_{i,k}p_{i,1,k}=& w \\
        \sum_{j,k}p_{1,j,k}=& w 
    \end{array} }.\\
    \cK^{\otimes3}(P_\bfx)=&\curbrkt{ P_{\bfx_1,\bfx_2,\bfx_3}=[P_{\bfx_2,\bfx_2,\bfx_3,\bfy}]_{\bfx_1,\bfx_2,\bfx_3}\in\cJ^{\otimes3}(P_\bfx)\colon 
    \begin{array}{rl}
        P_{\bfx_1,\bfx_2,\bfx_3,\bfy}\in&\Delta\paren{\curbrkt{0,1}^4}\\
        P_{\bfx_i,\bfy}(0,1)+P_{\bfx_i,\bfy}(1,0)\le& p,\;i=1,2,3
    \end{array}
    }\\
    =&\curbrkt{\sqrbrkt{ P_{\bfx_1,\bfx_2,\bfx_3,\bfy}}_{\bfx_1,\bfx_2,\bfx_3}\in\cJ^{\otimes3}(P_\bfx)\colon 
    \begin{array}{rl}
        p_{i,j,k,\ell}\ge&0,\;i,j,k,\ell\in\curbrkt{0,1}\\
        \sum_{i,j,k,\ell}p_{i,j,k,\ell} =& 1\\
        \sum_{j,k}p_{0,j,k,1}+p_{1,j,k,0}\le& p\\
        \sum_{i,k}p_{i,0,k,1}+p_{i,1,k,0}\le& p\\
        \sum_{i,j}p_{i,j,0,1}+p_{i,j,1,0}\le& p
    \end{array}
    }.
\end{align*}

$\wh\cJ^{\otimes(L+1)}(P_\bfx)$ and $\wh\cK^{\otimes(L+1)}(P_\bfx)$ are
extended formulations of $\cJ^{\otimes L}(P_\bfx)$ and $\cK^{\otimes L}(P_\bfx)$, respectively.

\begin{align*}
    \wh\cJ^{\otimes4}(P_\bfx)=&\curbrkt{P_{\bfx_1,\bfx_2,\bfx_3,\bfy}\colon \begin{array}{rl}
        p_{i,j,k,\ell}\ge&0,\; i,j,k,\ell\in \curbrkt{0,1}  \\
        \sum_{i,j,k,\ell\in \curbrkt{0,1}}p_{i,j,k,\ell}=&1\\
        \sum_{i,j}p_{i,j,1}=& w \\
        \sum_{i,k}p_{i,1,k}=& w \\
        \sum_{j,k}p_{1,j,k}=& w 
    \end{array} }.\\
    \wh\cK^{\otimes4}(P_\bfx)
    =&\curbrkt{P_{\bfx_1,\bfx_2,\bfx_3,\bfy}\in\wh\cJ^{\otimes4}(P_\bfx)\colon 
    \begin{array}{rl}
        \sum_{j,k\in \curbrkt{0,1}}p_{0,j,k,1}+p_{1,j,k,0}\le& p\\
        \sum_{i,k\in \curbrkt{0,1}}p_{i,0,k,1}+p_{i,1,k,0}\le& p\\
        \sum_{i,j\in \curbrkt{0,1}}p_{i,j,0,1}+p_{i,j,1,0}\le& p
    \end{array}
    }.
\end{align*}

To verify the value of Plotkin point $P_{L-1}$ at $L=3$, it suffices to verify that, if $w=1/2$, then $P_\bfx^{\otimes3}\notin\cK^{\otimes3}(P_\bfx)$ iff $p<1/4$, since we know that the optimizing input distribution when codewords are weight unconstrained is uniform. To this end, define a hyperplane
\begin{align*}
    \cH\paren{ P_\bfx^{\otimes3} }\coloneqq\curbrkt{P_{\bfx_1,\bfx_2,\bfx_3,\bfy}\in\wh\cJ^{\otimes4}(P_\bfx)\colon [P_{\bfx_1,\bfx_2,\bfx_3,\bfy}]_{\bfx_1,\bfx_2,\bfx_3}= P_\bfx^{\otimes3} }.
\end{align*}
Note that $P_{\bfx}^{\otimes3}\notin\cK^{\otimes 3}\paren{P_\bfx}$ is equivalent to $\cH\paren{ P_\bfx^{\otimes3} }\cap\wh\cK^{\otimes4}(P_\bfx)=\emptyset$. 
Since $\cH\paren{ P_\bfx^{\otimes3} }$ depends on $ w $ and $\wh\cK^{\otimes4}(P_\bfx)$ depends on $ w ,p$, we write them as $\cH( w )$ and $\wh\cK^{\otimes4}( w ,p)$, respectively, for simplicity.

We claim that the Plotkin point $P_{L-1}$ is precisely the optimal value of the following LP, i.e., the smallest $p^*$ such that the hyperplane $\cH(1/2)$ has no intersection with the corresponding high-dimensional polytope $\wh\cK^{\otimes4}(1/2,p^*)$.
\[
\begin{array}{rl}
    \min &p  \\
    \text{subject to} &\cH(1/2)\cap\wh\cK^{\otimes4}(1/2,p)\ne\emptyset. 
\end{array}
\]
Equivalently, collecting all constraints together, we want to find the minimal $p$ so that the  polytope (the feasible region of the LP) defined by the following constraints is nonempty.
\begin{align*}
    P_{\bfx_1,\bfx_2,\bfx_3,\bfy}\in&\wh\cJ^{\otimes4}(P_\bfx)\\
    [P_{\bfx_1,\bfx_2,\bfx_3,\bfy}]_{\bfx_1,\bfx_2,\bfx_3}=& P_\bfx^{\otimes3} \\
    \sum_{j,k\in \curbrkt{0,1}}p_{0,j,k,1}+p_{1,j,k,0}\le& p\\
    \sum_{i,k\in \curbrkt{0,1}}p_{i,0,k,1}+p_{i,1,k,0}\le& p\\
    \sum_{i,j\in \curbrkt{0,1}}p_{i,j,0,1}+p_{i,j,1,0}\le& p.
\end{align*}
Expanding everything out and noting that the first constraint regarding constant composition $P_{\bfx_1,\bfx_2,\bfx_3,\bfy\in\wh\cJ^{\otimes4}(P_\bfx)}$ is redundant since it is the same as the constraint $[P_{\bfx_1,\bfx_2,\bfx_3,\bfy}]_{\bfx_1,\bfx_2,\bfx_3}= P_\bfx^{\otimes3} \in\cJ^{\otimes3}(P_\bfx)$, we simplify the defining (in)equalities of the polytope as follows,
\begin{align*}
    p_{i,j,k,\ell}\ge&0,\; i,j,k,\ell\in \curbrkt{0,1}  \\
    \sum_{i,j,k,\ell\in \curbrkt{0,1}}p_{i,j,k,\ell}=&1\\
    p_{i,j,k,0}+p_{i,j,k,1}=&1/8,\;i,j,k\in \curbrkt{0,1}\\
    \sum_{j,k\in \curbrkt{0,1}}p_{0,j,k,1}+p_{1,j,k,0}\le& p\\
    \sum_{i,k\in \curbrkt{0,1}}p_{i,0,k,1}+p_{i,1,k,0}\le& p\\
    \sum_{i,j\in \curbrkt{0,1}}p_{i,j,0,1}+p_{i,j,1,0}\le& p,
\end{align*}
since $P_\bfx^{\otimes3}(i,j,k) = P_{\bfx}(i)P_{\bfx}(j)P_{\bfx}(k)=1/8$ for all $i,j,k\in\curbrkt{0,1}$.

Let 
\[\vp\coloneqq\begin{bmatrix} p_{0,0,0,0}&\cdots&p_{1,1,1,1} \end{bmatrix}^\top.\]
The LP can be written in a compact form as 
\begin{align*}
\begin{bmatrix}
1&1&1&1&1&1&1&1&1&1&1&1&1&1&1&1\\
1&1&&&&&&&&&&&&&&\\
&&1&1&&&&&&&&&&&&\\
&&&&1&1&&&&&&&&&&\\
&&&&&&1&1&&&&&&&&\\
&&&&&&&&1&1&&&&&&\\
&&&&&&&&&&1&1&&&&\\
&&&&&&&&&&&&1&1&&\\
&&&&&&&&&&&&&&1&1
\end{bmatrix}\vp = &\begin{bmatrix}
1 \\ 1/8\\1/8\\1/8\\1/8\\1/8\\1/8\\1/8\\1/8
\end{bmatrix},\\
\begin{bmatrix}
&1&&1&&1&&1& 1&&1&&1&&1&\\
&1&&1& 1&&1&& &1&&1& 1&&1&\\
&1&1&& &1&1&& &1&1&& &1&1&
\end{bmatrix}\vp\le&\begin{bmatrix}
p\\p\\p
\end{bmatrix}\\
\vp\ge&\vzero.
\end{align*}
Observe that as $p$ increases, the linear system becomes monotonically easier to be satisfied. Checked by \texttt{Mathematica}, the above LP is feasible if $p>1/4$ (and hence the distribution $\inputdistrunif^{\otimes3}$ is confusable) and is infeasible if $p<1/4$ (and hence  $\inputdistrunif^{\otimes3}$ is not confusable). Therefore, the $(p,L-1)$-list decoding capacity  hits 0 precisely at $p=1/4$. 

\subsection{$L=4$}
For $L=4$, one can obtain a similar LP whose infeasibility is equivalent to $\cH\paren{\begin{bmatrix}1/2\\1/2\end{bmatrix}^{\otimes4}}$ and $\wh\cK^{\otimes5}\paren{\inputdistrunif^{\otimes4},p}$ bing disjoint. 
{\tiny
\begin{align*}
    \begin{bmatrix}
    1&1&1&1&1&1&1&1&1&1&1&1&1&1&1&1&1&1&1&1&1&1&1&1&1&1&1&1&1&1&1&1\\
    1&1& &&&&&&&&&&&&&&&&&&&&&&&&&&&&& \\
    && 1&1 &&&&&&&&&&&&&&&&&&&&&&&&&&&& \\
    && && 1&1 &&&&&&&&&&&&&&&&&&&&&&&&&& \\
    && && && 1&1 &&&&&&&&&&&&&&&&&&&&&&&& \\
    && && && && 1&1 &&&&&&&&&&&&&&&&&&&&&& \\
    && && && && && 1&1 &&&&&&&&&&&&&&&&&&&& \\
    && && && && && && 1&1 &&&&&&&&&&&&&&&&&& \\
    && && && && && && && 1&1 &&&&&&&&&&&&&&&& \\
    && && && && && && && && 1&1 &&&&&&&&&&&&&& \\
    && && && && && && && && && 1&1 &&&&&&&&&&&& \\
    && && && && && && && && && && 1&1 &&&&&&&&&& \\
    && && && && && && && && && && && 1&1 &&&&&&&& \\
    && && && && && && && && && && && && 1&1 &&&&&& \\
    && && && && && && && && && && && && && 1&1 &&&& \\
    && && && && && && && && && && && && && && 1&1 && \\
    && && && && && && && && && && && && && && && 1&1 
    \end{bmatrix}
    \vp=&
    \begin{bmatrix}
    1\\1/16\\1/16\\1/16\\1/16\\1/16\\1/16\\1/16\\1/16\\1/16\\1/16\\1/16\\1/16\\1/16\\1/16\\1/16\\1/16
    \end{bmatrix},\\
    \begin{bmatrix}
    &1&&1&&1&&1&&1&&1&&1&&1&1&&1&&1&&1&&1&&1&&1&&1 \\
    &1&&1&&1&&1&1&&1&&1&&1&&&1&&1&&1&&1&1&&1&&1&&1& \\
    &1&&1&1&&1&&&1&&1&1&&1&&&1&&1&1&&1&&&1&&1&1&&1& \\
    &1&1&&&1&1&&&1&1&&&1&1&&&1&1&&&1&1&&&1&1&&&1&1&
    \end{bmatrix}\vp\le&
    \begin{bmatrix}
    p\\p\\p\\p
    \end{bmatrix},\\
    \vp\ge&\vzero,
\end{align*}
}
where 
\[\vp = \begin{bmatrix}p_{0,0,0,0,0} & \cdots &p_{1,1,1,1,1}\end{bmatrix}^\top.\]
One can numerically check that the above LP is feasible if $p>5/16$ and infeasible otherwise. 

In general, to check whether
\[\cH\paren{\begin{bmatrix}1/2\\1/2\end{bmatrix}^\tl}\cap\wh\cK^{\otimes(L+1)}\paren{\inputdistrunif^\tl,p}\]
is empty, it boils down to checking the infeasibility of a linear program with $2^{L+1}$ variables and $2^{L+1}+1+2^L+L$ constraints,
$2^{L+1}$ of them for non-negativity of probability mass, 1 of them for probability mass summing up to one, $2^L$ of them for ensuring that $P_{\bfx_1,\cdots,\bfx_L}\in\cJ^\tl(P_\bfx)$ is a $(P_\bfx,L)$-self-coupling, $L$ of them for the non-confusability guarantee: $P_{\bfx_1,\cdots,\bfx_L}\notin\cK^\tl(P_\bfx)$.
The size of the program (or the number of defining constraints of the corresponding polytope) grows exponentially in $L$. However, since we are concerned with absolute constant $L$ in this paper, for any given $L$, the feasibility can be certified in constant time.
Observe that, since the LP in the bit-flip setting is so structured, one can write it down explicitly by hand for any given $L$. 

\section{Blinovsky \cite{blinovsky-1986-ls-lb-binary} revisited}
\label{sec:blinovsky_revisited}
In this section, we \emph{fully} recover Blinovsky's \cite{blinovsky-1986-ls-lb-binary} results on characterization of the Plotkin points $P_{L-1}$ for $(p,L-1)$-list decoding under the bit-flip model.

Let $\phi$ be the standard bijection between $\curbrkt{0,1}$ and $\curbrkt{-1,1}$,
\begin{align*}
    \begin{array}{lccc}
        \phi\colon & \bF_2 & \to & \curbrkt{-1,1} \\
         & 0 & \mapsto & 1\\
         & 1 & \mapsto & -1.
    \end{array}
\end{align*}
We identify the type $\tau_\vx\in\cP^{(n)}(\bF_2)$ of a binary length-$n$ vector $\vx\in\bF_2^n$ using a $\curbrkt{-1,1}$-valued random variable $\bfx$ defined as
\[
    \prob{\bfx = -1} = \frac{\wth{\vx}}{n},\quad\prob{\bfx = 1} = 1-\frac{\wth{\vx}}{n}.
\]
Indeed the distribution $P_\bfx\in\cP^{(n)}(\curbrkt{-1,1})$ of $\bfx$ is the type of the image $\phi(\vx)$ of $\vx$ under $\phi$.
\[P_\bfx(\phi(0)) = \tau_\vx(0),\quad P_\bfx(\phi(1)) = \tau_\bfx(1).\]
For a collection of vectors $\vx_1,\cdots,\vx_k\in\bF_2^n$, their joint type is now represented by a sequence of random variables $\bfx_1,\cdots,\bfx_k$ with joint distribution $P_{\bfx_1,\cdots,\bfx_k}$, for any $x_1,\cdots,x_k\in\curbrkt{-1,1}$,
\begin{align*}
    P_{\bfx_1,\cdots,\bfx_k}(x_1,\cdots,x_k) =& \prob{\bfx_1 = x_1,\cdots,\bfx_k = x_k} \\
    =& \tau_{\vx_1,\cdots,\vx_k}(\phi^{-1}(x_1),\cdots,\phi^{-1}(x_k)).
\end{align*}
It is easy to check that, for $\vx_1,\vx_2\in\bF_2^n$,
\begin{equation}
    \label{eqn:relation_dist_exp}
    \frac{\disth{\vx_1}{\vx_2}}{n}=\frac{1}{2}\paren{1-\exptover{(\bfx_1,\bfx_2)\sim P_{\bfx_1,\bfx_2}}{\bfx_1\bfx_2}}.
\end{equation}
Indeed 
\begin{align*}
    \text{RHS} = &\frac{1}{2}\paren{1-\tau_{\vx_1,\vx_2}(0,1)\cdot(-1) - \tau_{\vx_1,\vx_2}(1,0)\cdot(-1) - \tau_{\vx_1,\vx_2}(0,0)\cdot1 - \tau_{\vx_1,\vx_2}(1,1)\cdot1}\\
    =&\frac{1}{2}\paren{1+\frac{\disth{\vx_1}{\vx_2}}{n} - \paren{1-\frac{\disth{\vx_1}{\vx_2}}{n}}}\\
    =&\text{LHS}.
\end{align*}

Let
\begin{equation}
    r \coloneqq \exptover{(\bfx_1,\cdots,\bfx_L)\sim\curbrkt{-1,1}^L}{\abs{\bfx_1 + \cdots + \bfx_L}},
    \label{eqn:def_r_dist_rw}
\end{equation}
be the expected translation distance of a 1-dimensional unbiased random walk after $L$ steps.
Each $\bfx_i$ ($1\le i\le L$) is  independent and uniformly distributed on $\curbrkt{-1,1}$.
\begin{theorem}
\label{thm:plotkin_ld_binary}
The  Plotkin point $P_{L-1}$ for $(p,L-1)$-list decoding is given by 
\[P_{L-1} = \frac{1-r/L}{2}.\]
\end{theorem}

\begin{remark}
Note that the formula in Theorem \ref{thm:plotkin_ld_binary} agrees with the one by Blinovsky. To see this, we first compute $r$. For odd $L = 2k + 1$, where $k\in\bZ_{>0}$ is some strictly positive integer, it is easy to see that
\begin{align*}
    r =& \expt{\abs{\bfx_1 + \cdots + \bfx_L}}\\
    =&\sum_{i = 0}^k \frac{2\binom{L}{i}}{2^L}(L-2i).
\end{align*}
Recall that, by binomial theorem (Fact \eqref{eqn:binom_thm}),
\[2^L = \sum_{i = 0}^L\binom{L}{i} = \sum_{i = 0}^k 2\binom{L}{i}.\]
Now we simplify the formula in Theorem \ref{thm:plotkin_ld_binary}.
\begin{align}
    P_{L-1} = & \frac{1}{2} - \frac{r}{2L} \notag\\
    =& \sum_{i=0}^k\frac{\binom{L}{i}}{2^L} - \sum_{i=0}^k\paren{1-\frac{2i}{L}}\frac{\binom{L}{i}}{2^L} \notag\\
    =&\sum_{i = 0}^k\frac{2i}{L}\frac{\binom{L}{i}}{2^L} \notag\\
    =&\sum_{i = 1}^k\frac{i}{L}\frac{\frac{L}{i}\binom{L - 1}{i - 1}}{2^{L-1}} \label{eqn:comb_id}\\
    =&\frac{1}{2^{L-1}}\sum_{i = 0}^{k - 1}\binom{L-1}{i} \notag\\
    =&\frac{1}{2^{L-1}}\frac{1}{2}\paren{2^{L-1} - \binom{L - 1}{k}} \label{eqn:binom_thm_app}\\
    =&\frac{1}{2} - 2^{-L}\binom{2k}{k}, \notag
\end{align}
where 
     Eqn. \eqref{eqn:comb_id} is by Fact \eqref{eqn:binom_recurse};
     Eqn. \eqref{eqn:binom_thm_app} follows from binomial theorem (Fact \eqref{eqn:binom_thm}) again,
    \[2^{L-1} = \binom{2k}{k} + 2\sum_{i = 0}^{k -1} \binom{2k}{i}. \]

\end{remark}

\begin{lemma}[Lower bound]
\label{eqn:plotkin_pt_bitflip_lb}
The  Plotkin point $P_{L-1}$ for $(p,L-1)$-list decoding is lower bounded by 
\[P_{L-1} \ge \frac{1-r/L}{2}.\]
That is, if $p<P_{L-1}$, then the $(p,L-1)$-list decoding capacity is positive, i.e., there is an infinite sequence of $(p,L-1)$-list decodable codes of positive rate.
\end{lemma}
\begin{proof}
We will show that if $p = \frac{1-\frac{r+\eta}{L}}{2}<\frac{1-r/L}{2}$ for any $\eta>0$, then the product distribution $\bern^{\otimes L}\paren{1/2}$ lies outside the corresponding confusability set $\cK^{\otimes L}\paren{\bern\paren{1/2}}$. Using the framework developed in this paper, a random code of a suitable positive rate in which 
each codeword is sampled independently and uniformly from $\cT_\vbfx\paren{\bern(1/2)}$
is $(p,L-1)$-list decodable w.h.p. 

The proof is by contradiction. If $P_{\bfx_1,\cdots,\bfx_L}\coloneqq\bern^{\otimes L}\paren{1/2}$ is confusable, then, by the definition \ref{def:conf_tuples} of confusability of tuples,  an $L$-tuple of distinct codewords $\vx_1,\cdots,\vx_L$ of joint type $\tau_{\vx_1,\cdots,\vx_L}=P_{\bfx_1,\cdots,\bfx_L}$ can be covered by a ball of radius $np$ centered around some $\vy\in\bF_2^n$. Equivalently, by the definition \ref{def:conf_distr} of confusability of distributions, there is a refinement $P_{\bfx_1,\cdots,\bfx,\bfy}\in\Delta\paren{\curbrkt{-1,1}^{L+1}}$ such that
$\sqrbrkt{P_{\bfx_1,\cdots,\bfx_L,\bfy}}_{\bfx_1,\cdots,\bfx_L} = P_{\bfx_1,\cdots,\bfx_L}$, and for every $i\in[L]$, 
\[P_{\bfx_i,\bfy}(0,1) + P_{\bfx_i,\bfy}(1,0) \le p.\]
This means that for every $i\in[L]$,
\begin{equation}
    \expt{\bfx_i\bfy}\ge\frac{r+\eta}{L},\notag
\end{equation}
by the relation (Eqn. \eqref{eqn:relation_dist_exp}) between  Hamming distance between vectors and correlation of their random variable representations. 
Hence
\begin{equation}
    \expt{\paren{\bfx_1 + \cdots + \bfx_L}\bfy}\ge r+\eta.
    \label{eqn:corr_cor}
\end{equation}
The $\curbrkt{-1,1}$-valued random variable $\bfy$ that has the largest correlation with $\bfx_1 + \cdots + \bfx_L$ is  $\bfy=\maj\paren{\bfx_1,\cdots,\bfx_L}$, where
\begin{align*}
    \begin{array}{lccc}
        \maj\colon & \curbrkt{-1,1}^L & \to & \curbrkt{-1,1} \\
         & (x_1,\cdots,x_L) & \mapsto & \sgn\paren{x_1 + \cdots + x_L}.
    \end{array}
\end{align*}
is the majority function. 
To see this, just expand the above expectation,
\begin{align*}
    \expt{\paren{\bfx_1 + \cdots + \bfx_L}\bfy} = &\sum_{x_1,\cdots,x_L,y\in\curbrkt{-1,1}}P_{\bfx_1,\cdots,\bfx_L,\bfy}(x_1,\cdots,x_L,y)(x_1 + \cdots + x_L)y\\
    =&\sum_{x_1,\cdots,x_L\in\curbrkt{-1,1}}P_{\bfx_1,\cdots,\bfx_L}(x_1,\cdots,x_L)\sum_{y\in\curbrkt{-1,1}}P_{\bfy|\bfx_1,\cdots,\bfx_L}(y|x_1,\cdots,x_L)(x_1 + \cdots + x_L)y.
\end{align*}
Note that, each summand
\[P_{\bfy|\bfx_1,\cdots,\bfx_L}(1|x_1,\cdots,x_L)(x_1 + \cdots + x_L) - P_{\bfy|\bfx_1,\cdots,\bfx_L}(-1|x_1,\cdots,x_L)(x_1 + \cdots + x_L)\]
is maximized  when the conditional probability mass of $\bfy$ is concentrated on the singleton $\sgn(x_1 + \cdots + x_L)$,
\[P_{\bfy|\bfx_1,\cdots,\bfx_L}(\sgn(x_1 + \cdots + x_L)|x_1,\cdots,x_L) = 1,\quad P_{\bfy|\bfx_1,\cdots,\bfx_L}(-\sgn(x_1 + \cdots + x_L)|x_1,\cdots,x_L) = 0.\]
In this case, each summand attains its maxima 
\[\sgn(x_1+\cdots+x_L)(x_1 + \cdots+x_L) = \abs{x_1+\cdots+x_L}.\]
Overall, the corresponding maximal correlation is precisely
\begin{align}
    \expt{(\bfx_1+\cdots+\bfx_L)\maj(\bfx_1,\cdots,\bfx_L)} = & \sum_{x_1,\cdots,x_L\in\curbrkt{-1,1}}P_{\bfx_1,\cdots,\bfx_L}(x_1,\cdots,x_L)\abs{x_1+\cdots+x_L}\notag\\
    =&\exptover{(\bfx_1,\cdots,\bfx_L)\sim P_{\bfx_1,\cdots,\bfx_L}}{\abs{\bfx_1+\cdots+\bfx_L}}.\label{eqn:rw_dist_equals_max_corr}
\end{align}

Using the above observation, we get
\begin{align}
    r=&\exptover{(\bfx_1,\cdots,\bfx_L)\sim\curbrkt{-1,1}^L}{\abs{\bfx_1 + \cdots + \bfx_L}}\label{eqn:use_def_r}\\
    =&\expt{\paren{\bfx_1 + \cdots + \bfx_L}\maj\paren{\bfx_1,\cdots,\bfx_L}}\label{eqn:use_dist_equals_corr}\\
    \ge& r+\eta,\label{eqn:use_corr_cor}
\end{align}
Eqn. \eqref{eqn:use_def_r} is by the definition of $r$ (Eqn. \eqref{eqn:def_r_dist_rw}). Eqn. \eqref{eqn:use_dist_equals_corr} follows from Eqn. \eqref{eqn:rw_dist_equals_max_corr}. Eqn. \eqref{eqn:use_corr_cor} is by Eqn. \eqref{eqn:use_corr_cor}.
We hence reach a contradiction which finishes the proof. 
\end{proof}

\begin{lemma}[Upper bound]
\label{lem:plotkin_pt_bitflip_ub}
The  Plotkin point $P_{L-1}$ for $(p,L-1)$-list decoding is upper bounded by 
\[P_{L-1} \le \frac{1-r/L}{2}.\]
That is, if $p>P_{L-1}$, then no positive rate is possible, i.e, there is no infinite sequence of $(p,L-1)$-list decodable codes of positive rate. 
\end{lemma}
\begin{proof}
Our goal is to show that if $p>P_{L-1}$, then $C_{L-1} = 0$. Suppose $p = \frac{1-\frac{r-\eta}{L}}{2} $ for a  constant $\eta>0$.

We are going to show that any infinite  sequence of codes $\cC_n$ each of positive rate is {not} $(p,L-1)$-list decodable.
First, by the previous  argument in last section, we can extract a sequence of subcodes $\cC_n'\subseteq\cC_n$ of positive rate satisfying that, for every tuple of \emph{distinct} codewords $\vx_1,\cdots,\vx_L\in\cC'$ and $x_1,\cdots,x_L\in\bF_2$,
\[\abs{ \tau_{\vx_{1},\cdots,\vx_L}(x_1,\cdots,x_L) - \wh P_{\bfx_1,\cdots,\bfx_L}(x_1,\cdots,x_L) }\le\zeta\]
for some \emph{symmetric} distribution $\wh P_{\bfx_1,\cdots,\bfx_L}\in\Delta\paren{\cX^L}$ and some positive constant $\zeta>0$. In favour of the proceeding calculations, it suffices to take
\begin{equation}
\zeta = \frac{L}{(L-1)r2^{L+2}}\eta.
\label{eqn:zeta_blinovsky_ld_plotkin_recover}
\end{equation}

To show non-list decodability of $\cC'$ (and hence $\cC$), we will argue that there is a list $(\vx_{i_1},\cdots,\vx_{i_L})\in\binom{\cC'}{L}$ that can be covered by a ball of radius $np$ centered around $\maj\paren{\vx_{i_1},\cdots,\vx_{i_L}}$. The proof is by contradiction.  Suppose this is not the case, i.e., no list can be covered by the ball centered at its majority.
Define, for $(i_1,\cdots, i_L)\in\sqrbrkt{2^{nR}}^L$,
\[Q_{i_1,\cdots,i_L} = \paren{\bfx_{i_1} + \cdots + \bfx_{i_L}}\cdot\maj\paren{\bfx_{i_1},\cdots,\bfx_{i_L}} - r.\]
We will provide a \emph{strictly negative} upper bound and a \emph{non-negative} lower bound  on 
\[Q\coloneqq\mathop{\bE}_{(\bfi_1,\cdots,\bfi_L)\sim\sqrbrkt{2^{nR}}^L}{\exptover{(\bfx_{\bfi_1},\cdots,\bfx_{\bfi_L})\sim P_{\bfx_{\bfi_1},\cdots,\bfx_{\bfi_L}}}{Q_{\bfi_1,\cdots,\bfi_L}}},\]
which is a contradiction and finishes the proof.

\noindent\textbf{Upper bound on $Q$.} By the assumption of list decodability, for every $L$-tuple of distinct codewords $\vx_1,\cdots,\vx_L\in\cC'$, there is a codeword $\vx_i$ ($i\in[L]$) among them such that
\[\disth{\vx_i}{\maj\paren{\vx_1,\cdots,\vx_L}}\ge np.\]
Equivalently,
\[\expt{\bfx_i\maj\paren{\bfx_1,\cdots,\bfx_L}}\le\frac{r-\eta}{L}.\]

Since $\wh P_{\bfx_1,\cdots,\bfx_L}$ is symmetric and $\cC'$ is $\paren{\zeta,\wh P_{\bfx_1,\cdots,\bfx_L}}$-equicoupled,  we expect $\expt{\bfx_j\maj\paren{\bfx_1,\cdots,\bfx_L}}\lesssim\frac{r-\eta}{L}$ for all $j\in[L]$, potentially with some slack depending on $\zeta$. Indeed, for any $j\in[L]\setminus\curbrkt{i}$ (without loss of generality, assume $j>i$),
\begin{align}
    &\abs{\expt{\bfx_i\maj\paren{\bfx_1,\cdots,\bfx_L}} - \expt{\bfx_j\maj\paren{\bfx_1,\cdots,\bfx_L}}}\notag\\
    =&\left|\sum_{x_1,\cdots,x_L\in\curbrkt{-1,1}}\tau_{\vx_1,\cdots,\vx_L}(\phi^{-1}(x_1),\cdots,\phi^{-1}(x_L))x_i\maj\paren{x_1,\cdots,x_L}\right.\\
    &\left.- \sum_{x_1,\cdots,x_L\in\curbrkt{-1,1}}\tau_{\vx_1,\cdots,\vx_L}(\phi^{-1}(x_1),\cdots,\phi^{-1}(x_L))x_j\maj\paren{x_1,\cdots,x_L}\right|\notag\\
    =&\left|\sum_{x_1,\cdots,x_L\in\curbrkt{-1,1}}\tau_{\vx_1,\cdots,\vx_L}(\phi^{-1}(x_1),\cdots,\phi^{-1}(x_L))x_i\maj\paren{x_1,\cdots,x_L}\right.\notag\\
    &- \left.\sum_{x_{\sigma(1)},\cdots,x_{\sigma(L)}\in\curbrkt{-1,1}}\tau_{\vx_1,\cdots,\vx_L}(\phi^{-1}(x_{\sigma(1)}),\cdots,\phi^{-1}(x_{\sigma(L)}))x_{\sigma(j)}\maj\paren{x_{\sigma(1)},\cdots,x_{\sigma(L)}}\right|\label{eqn:perm_def}\\
    =&\left|\sum_{x_1,\cdots,x_L\in\curbrkt{-1,1}}\tau_{\vx_1,\cdots,\vx_L}(\phi^{-1}(x_1),\cdots,\phi^{-1}(x_L))x_i\maj\paren{x_1,\cdots,x_L}\right.\notag\\
    &- \left.\sum_{x_{1},\cdots,x_{L}\in\curbrkt{-1,1}}\tau_{\vx_{1},\cdots,\vx_{L}}(\phi^{-1}(x_{\sigma(1)}),\cdots,\phi^{-1}(x_{\sigma(L)}))x_{i}\maj\paren{x_{1},\cdots,x_{L}}\right|\notag\\
    =&\abs{\sum_{x_1,\cdots,x_L\in\curbrkt{-1,1}} \sqrbrkt{\begin{array}{c}
         \left(\begin{array}{c}
            \tau_{\vx_{1},\cdots,\vx_L}(\phi^{-1}(x_1),\cdots,\phi^{-1}(x_L))     \\
             - \wh P_{\bfx_1,\cdots,\bfx_L}(\phi^{-1}(x_1),\cdots,\phi^{-1}(x_L))
         \end{array}
        \right)  \\
         + \left(\begin{array}{c}
              \wh P_{\bfx_{1},\cdots,\bfx_{L}}(\phi^{-1}(x_{\sigma(1)}),\cdots,\phi^{-1}(x_{\sigma(L)})) \\
              - \tau_{\vx_{1},\cdots,\vx_{L}}(\phi^{-1}(x_{\sigma(1)}),\cdots,\phi^{-1}(x_\sigma(L))) 
         \end{array} \right)
    \end{array}} x_i\maj\paren{x_1,\cdots,x_L}}  \label{eqn:symm}\\
    \le&\paren{\begin{array}{c}
         \left|\begin{array}{c}
            \tau_{\vx_{1},\cdots,\vx_L}(\phi^{-1}(x_1),\cdots,\phi^{-1}(x_L))     \\
             - \wh P_{\bfx_1,\cdots,\bfx_L}(\phi^{-1}(x_1),\cdots,\phi^{-1}(x_L))
         \end{array}
        \right|  \\
         + \left|\begin{array}{c}
              \wh P_{\bfx_{1},\cdots,\bfx_{L}}(\phi^{-1}(x_{\sigma(1)}),\cdots,\phi^{-1}(x_{\sigma(L)})) \\
              - \tau_{\vx_{1},\cdots,\vx_{L}}(\phi^{-1}(x_{\sigma(1)}),\cdots,\phi^{-1}(x_\sigma(L))) 
         \end{array} \right|
    \end{array}}\abs{\sum_{x_1,\cdots,x_L\in\curbrkt{-1,1}}x_i\maj\paren{x_1,\cdots,x_L}} \label{eqn:tri_ineq}\\
    \le&2\zeta\cdot\frac{2^L}{L}\exptover{(\bfx_1,\cdots,\bfx_L)\sim\curbrkt{-1,1}^L}{\paren{\bfx_1 + \cdots + \bfx_L}\maj\paren{\bfx_1,\cdots,\bfx_L}}\label{eqn:expt}\\
    =&\frac{2^{L+1}}{L}\zeta\exptover{(\bfx_1,\cdots,\bfx_L)\sim\curbrkt{-1,1}^L}{\abs{\bfx_1+\cdots+\bfx_L}}\notag\\
    =&\frac{2^{L+1}r}{L}\zeta. 
\end{align}
In the above chain of equalities and inequalities, we used the following facts.
\begin{enumerate}
    \item In Eqn. \eqref{eqn:perm_def},  $\sigma\in S_L$ denotes the transposition which swaps the $i$-th and $j$-th element,
    \[\sigma = \paren{\begin{array}{ccccccccccc}
    1 & \cdots & i-1 & i & i+1 & \cdots & j-1 & j & j+1 & \cdots & L\\
    1 & \cdots & i-1 & j & i+1 & \cdots & j-1 & i & j+1 & \cdots & L
    \end{array}}.\]
    \item Eqn. \eqref{eqn:symm} is due to symmetry of $\wh P_{\bfx_1,\cdots,\bfx_L}$.
    \item Inequality \eqref{eqn:tri_ineq} is by triangle inequality of absolute value. 
    \item  Eqn. \eqref{eqn:expt} follows since
    \[\abs{\sum_{x_1,\cdots,x_L\in\curbrkt{-1,1}}x_i\maj\paren{x_1,\cdots,x_L}} =2^L\abs{\frac{1}{L}\sum_{i=1}^L\sum_{x_1,\cdots,x_L\in\curbrkt{-1,1}}\frac{1}{2^L}x_i\maj\paren{x_1,\cdots,x_L}},\]
    and the expectation is over $\bfx_i$'s which are independent and uniformly distributed on $\curbrkt{-1,1}$.
\end{enumerate}
Now, for any $j\in[L]\setminus\curbrkt{i}$,
\begin{align*}
    \expt{\bfx_j\maj\paren{\bfx_1,\cdots,\bfx_L}}=&\expt{\bfx_i\maj\paren{\bfx_1,\cdots,\bfx_L}} + \paren{\expt{\bfx_j\maj\paren{\bfx_1,\cdots,\bfx_L}} - \expt{\bfx_i\maj\paren{\bfx_1,\cdots,\bfx_L}}}\\
    \le&\frac{r-\eta}{L} + \frac{2^{L+1}r}{L}\zeta.
\end{align*}
Thus we have
\begin{align*}
    \expt{\paren{\bfx_1 + \cdots + \bfx_L}\maj\paren{\bfx_1,\cdots,\bfx_L}}\le&r-\eta + \frac{2^{L+1}r(L-1)}{L}\zeta.
\end{align*}
That is,
\begin{align}
    \expt{Q_{1,\cdots,L}} 
    =&\expt{ \paren{\bfx_1 + \cdots + \bfx_L}\maj\paren{\bfx_1,\cdots,\bfx_L} - r }\notag\\
    \le&-\eta + \frac{2^{L+1}r(L-1)}{L}\zeta\notag\\
    = & -\frac{\eta}{2},\label{eqn:plug_in_zeta_blinovsky_ld_plotkin}
\end{align}
where the last Eqn. \eqref{eqn:plug_in_zeta_blinovsky_ld_plotkin} follows by the choice of $\zeta$ (Eqn. \eqref{eqn:zeta_blinovsky_ld_plotkin_recover}).
Since the above calculations work for {any} list $\vx_1,\cdots,\vx_L\in\cC'$ of \emph{distinct} codewords, we have that for $(i_1,\cdots,i_L)\in\binom{[M']}{L}$, the same bound holds,
\[\expt{Q_{i_1,\cdots,i_L}}\le-\frac{\eta}{2}.\]

For lists $(i_1,\cdots,i_L)\in\sqrbrkt{M'}^L$ that are not all distinct, we use the trivial bound,
\begin{align*}
    \expt{Q_{i_1,\cdots,i_L}}=&\expt{\abs{\bfx_{i_1} + \cdots + \bfx_{i_L}} - r}\\
    \le&L-r.
\end{align*}

Overall we have
\begin{align}
    Q=&\mathop{\bE}_{(\bfi_1,\cdots,\bfi_L)\sim\sqrbrkt{2^{nR}}^L}{\expt{Q_{\bfi_1,\cdots,\bfi_L}}} \notag\\
    =&\frac{1}{2^{nRL}}\paren{\sum_{i_1,\cdots,i_L\in\sqrbrkt{2^{nR}}\text{ distinct}} Q_{i_1,\cdots,i_L} + \sum_{i_1,\cdots,i_L\in\sqrbrkt{2^{nR}}\text{  not distinct}}Q_{i_1,\cdots,i_L} } \notag\\
    \le&\frac{1}{2^{nRL}}\bigg[2^{nR}\paren{2^{nR} - 1}\cdots\paren{2^{nR}-L+1}\paren{-\frac{\eta}{2}}\notag\\
    &+\paren{2^{nRL} - 2^{nR}\paren{2^{nR} - 1}\cdots\paren{2^{nR}-L+1} } (L-r)\bigg]   \notag\\
    <&0.\label{eqn:negative_blinovsky_ld_plotkin_recover}
\end{align}
The last inequality \eqref{eqn:negative_blinovsky_ld_plotkin_recover} holds if 
\[\cardCp>\max\curbrkt{ 2(L-1),\frac{2^{L+1}L!(L+r)}{\eta} },\]
by similar calculations to Sec. \ref{sec:converse_symm}.

\noindent\textbf{Lower bound on $Q$.}
Following the calculations in the proof of generalized Plotkin bound for list decoding, we have
\begin{align}
    Q + r=&\mathop{\bE}_{(\bfi_1,\cdots,\bfi_L)\sim\sqrbrkt{2^{nR}}^L}{\expt{\abs{\bfx_{\bfi_1} + \cdots + \bfx_{\bfi_L}}}} \notag\\
    =&\frac{1}{2^{nRL}}\sum_{i_1,\cdots,i_L\in\sqrbrkt{2^{nR}}}\sum_{x_1,\cdots,x_L\in\curbrkt{-1,1}}\tau_{\vx_{i_1},\cdots,\vx_{i_L}}(\phi^{-1}(x_1),\cdots,\phi^{-1}(x_L))\abs{x_1+\cdots+x_L} \notag\\
    =&\frac{1}{2^{nRL}}\sum_{i_1,\cdots,i_L\in\sqrbrkt{2^{nR}}}\sum_{x_1,\cdots,x_L\in\curbrkt{-1,1}} \frac{1}{n}\sum_{j=1}^n\indicator{\vx_{i_1}(j) = \phi^{-1}(x_1)}\cdots\indicator{\vx_{i_L}(j) = \phi^{-1}(x_L)} \abs{x_1+\cdots+x_L} \label{eqn:joint_type_def}\\
    =&\frac{1}{n}\sum_{j=1}^n\sum_{x_1,\cdots,x_L\in\curbrkt{-1,1}}\prod_{\ell=1}^L\paren{\frac{1}{2^{nR}}\sum_{i\in\sqrbrkt{2^{nR}}}\indicator{\vx_{i}(j) = \phi^{-1}(x_\ell)} }\abs{x_1 + \cdots + x_L} \label{eqn:rearrange}\\
    =&\frac{1}{n}\sum_{j=1}^n\sum_{x_1,\cdots,x_L\in\curbrkt{-1,1}}\prod_{\ell=1}^LP_\bfx^{(j)}(\phi^{-1}(x_\ell))\abs{x_1+\cdots+x_L} \label{eqn:pj_def}\\
    =&\exptover{\bfj\sim[n]}{\exptover{\paren{\bfx_1^{(\bfj)},\cdots,\bfx_L^{(\bfj)}}\sim \paren{P_\bfx^{(\bfj)}}^{\otimes L}}{\abs{\bfx_1^{(\bfj)} +\cdots + \bfx_L^{(\bfj)}}}}. \label{eqn:q}
\end{align}
In the above calculations, we used the following definitions and facts.
\begin{enumerate}
    \item Eqn. \eqref{eqn:joint_type_def} follows from the definition of joint types.
    \item Eqn. \eqref{eqn:rearrange} is obtained by rearranging terms.
    \item In Eqn. \eqref{eqn:pj_def}, as before, we let, for $j\in[n]$, $x\in\bF_2$,
    \[P_\bfx^{(j)}(x)=\frac{1}{2^{nR}}\sum_{i\in\sqrbrkt{2^{nR}}}\indicator{\vx_i(j) = x}\]
    denote the empirical distribution of the $j$-th \emph{column} of $\cC'$ when viewed as an $M'\times n$ matrix. 
\end{enumerate}
In expression \eqref{eqn:q}, the $j$-th summand can be viewed as the translation distance of a non-lazy one-dimensional random walk after $L$ steps. The walker moves left ($x = 1$) with probability $P_\bfx^{(j)}(1)$ and moves right ($x = 0$) with probability $P_\bfx^{(j)}(0)$. It is not hard to check that the expected translation distance  is minimized when the walker is unbiased, i.e., when $P_\bfx^{(j)}(1) = P_\bfx^{(j)}(0) = 1/2$. This is formally justified in Appendix \ref{app:rw_dist}. Hence, for every $j\in[n]$,
\[\exptover{\paren{\bfx_1^{(j)},\cdots,\bfx_L^{(j)}}\sim \paren{P_\bfx^{(j)}}^{\otimes L}}{\abs{\bfx_1^{(j)} +\cdots + \bfx_L^{(j)}}} - r\ge0.\]
Since the above bound is valid for every $j\in[n]$, it is still valid  averaged over $\bfj\sim[n]$. Hence we have
$Q\ge0$.
\end{proof}

\section{GV rate vs. cloud rate}
\label{sec:gv_vs_cloud}
In this section, we are concerned with the question of  unique decoding (special case where $L-1 = 1$) under the bit-flip model.

In \cite{wang-budkuley-bogdanov-jaggi-2019-omniscient-avc}, bounds on achievable rates of codes for general adversarial channels are provided. A Gilbert--Varshamov-type expression was obtained  using a purely random code construction, and a rate lower bound (we call \emph{cloud rate}) that generalizes the GV-type expression was given by a cloud code construction. We evaluate both bounds under the bit-flip model. We show that the Gilbert--Varshamov-type bound for general adversarial channels indeed coincide with the classic GV bound in this particular setting. We also provide a convex program for evaluating the cloud rate.

We use the probability vector $\begin{bmatrix}P_\bfx(1)&\cdots&P_\bfx(\cardX)\end{bmatrix}^\top$
to denote a distribution $P_\bfx\in\Delta(\cX)$. Take any input distribution 
\[P_\bfx = \bern(w) = \inputdistr{w},\]
from $\Delta(\curbrkt{0,1})$,
we first explicitly compute the basic objects we are concerned with in this paper.
\begin{align*}
    \Delta\coloneqq&\Delta(\curbrkt{0,1})\\
    =&\curbrkt{P_{\bfx_1,\bfx_2}\in\bR^{2\times2}\colon \begin{array}{rl}
        P_{\bfx_1,\bfx_2}(x_1,x_2) &\ge0,\;\forall x_1,x_2  \\
        \sum_{x_1,x_2}P_{\bfx_1,\bfx_2}(x_1,x_2) &=1 
    \end{array}}\\
    =&\curbrkt{\begin{bmatrix}a&c\\d&b\end{bmatrix}\in\bR^{2\times2}\colon \begin{array}{rl}
        a,b,c,d & \ge0 \\
        a+b+c+d & = 1 
    \end{array} }\\
    =&\curbrkt{\begin{bmatrix}a&c\\1-a-b-c&b\end{bmatrix}\in\bR^{2\times2}\colon \begin{array}{rl}
        a,b,c   &\ge0  \\
         a+b+c&\le1 
    \end{array}  }.\\
    \cJ(w)\coloneqq&\cJ\paren{\inputdistr{w}}\\
    =&\curbrkt{P_{\bfx_1,\bfx_2}\in\Delta\colon P_{\bfx_1}=P_{\bfx_2}=P_\bfx}\\
    =&\curbrkt{\begin{bmatrix}a&c\\d&b\end{bmatrix}\in\bR^{2\times2}\colon\begin{array}{l}
        a,b,c,d\ge0\\
        a+b+c+d=1   \\
        d+b=w\\
        c+b=w
    \end{array}}\\
    =&\curbrkt{\begin{bmatrix}1-w-d&d\\d&w-d\end{bmatrix}\in\bR^{2\times2}\colon 0\le d\le \min\{w,1-w\}}.\\
    \cK(w,p)\coloneqq&\cK\paren{\inputdistr{w}}\\
    =&\curbrkt{P_{\bfx_1,\bfx_2}\in\cJ\paren{w}\colon P_{\bfx_1,\bfx_2}(0,1)+P_{\bfx_1,\bfx_2}(1,0)\le 2p}\\
    =&\curbrkt{\begin{bmatrix} 1-w-d&d\\d&w-d \end{bmatrix}\in\bR^{2\times2}\colon 0\le d\le\min\{w,1-w,p\}}.
\end{align*}
Since $\cp_2=\dnn_2$, we have
\begin{align*}
    \cp_2(w)=&\cp_2\cap\cJ(w)\\
    =&\curbrkt{\begin{bmatrix}w-d&d\\d&1-w-d\end{bmatrix}\colon 0\le d\le \min\{w,1-w\},\;(w-d)(1-w-d)-d^2\ge0}\\
    =&\curbrkt{\begin{bmatrix}w-d&d\\d&1-w-d\end{bmatrix}\colon 0\le d\le w-w^2}.
\end{align*}
Note that to ensure $\cp_2(w)\setminus\cK(w,p)\ne\emptyset$, we need 
\[\begin{array}{ll}
    0<p<1/4,\;w\in\paren{\frac{1-\sqrt{1-4p}}{2},\frac{1+\sqrt{1-4p}}{2}}.
\end{array}\]
In other words, $0<w<1$ and $0<p<w-w^2$. In this case,
\[\cK(w,p)=\curbrkt{\begin{bmatrix} 1-w-d&d\\d&w-d \end{bmatrix}\in\bR^{2\times2}\colon 0\le d\le p}.\]
Actually, if the above conditions hold, then when $1/3\le w <1$, the boundary of $\cK(w,p)$ is $p$ and the boundary of $\cp_2(w)$ is $w-w^2$. Note that the right boundary $\begin{bmatrix}(1-w)^2&w-w^2\\w-w^2&w^2\end{bmatrix}=\inputdistr{w}^{\otimes2}$ of $\cp_2(w)$ is the \emph{only} distribution in $\cp_2(w)$ of $\cprk$-1.

\noindent\textbf{GV rate.} We first state the GV-type expression given by in \cite{wang-budkuley-bogdanov-jaggi-2019-omniscient-avc}.
\begin{lemma}[Gilbert--Varshamov rate]
For a general adversarial channel $\cA = \curbrkt{\cX,\lambda_\bfx,\cS,\lambda_\bfs,\cY,W_{\bfy|\bfx,\bfs}}$, its unique decoding capacity is at least 
\[\RGV=\max_{P_\bfx\in\lambda_\bfx}\min_{P_{\bfx_1,\bfx_2\in\cK(P_\bfx)}}I(\bfx;\bfx'),\]
where the mutual information is calculated using $P_{\bfx_1,\bfx_2}$.
\end{lemma}
We now evaluate the above expression under the bit-flip model.
\begin{align*}
    \RGV=&\max_{P_\bfx\in\lambda_\bfx}\min_{P_{\bfx_1,\bfx_2\in\cK(P_\bfx)}}I(\bfx;\bfx')\\
    =&\max_{\inputdistr{w}\in\Delta}\min_{\begin{bmatrix}1-w-d&d\\d&w-d\end{bmatrix}\in\cK(w,p)}D\paren{\begin{bmatrix}1-w-d&d\\d&w-d\end{bmatrix}\left\|\inputdistr{w}^{\otimes2}\right.}\\
    =&\max_{0<w<1}\min_{0\le d\le p}(w-d)\log\frac{w-d}{w^2} + 2d\log\frac{d}{w(1-w)} + (1-w-d)\log\frac{1-w-d}{(1-w)^2}\\
    =&\max_{0<w<1}(w-p)\log\frac{w-p}{w^2} + 2p\log\frac{p}{w(1-w)} + (1-w-p)\log\frac{1-w-p}{(1-w)^2}\\
    =&(1/2-p)\log\frac{1/2-p}{(1/2)^2} + 2p\log\frac{p}{(1/2)(1-1/2)} + (1-1/2-p)\log\frac{1-1/2-p}{(1-1/2)^2}\\
    =&1-H(2p).
\end{align*}
This matches the classic GV bound given a greedy volume packing argument. 

\noindent\textbf{Cloud rate.} We now state the cloud rate expression given by \cite{wang-budkuley-bogdanov-jaggi-2019-omniscient-avc}.
\begin{lemma}[Cloud rate]

\end{lemma}
For a general adversarial channel $\cA = \curbrkt{\cX,\lambda_\bfx,\cS,\lambda_\bfs,\cY,W_{\bfy|\bfx,\bfs}}$, its unique decoding capacity is at least 
\begin{align*}
    \Rcloud=&\max_{P_{\bfx}\in\lambda_\bfx}\max_{P_{\bfx_1,\bfx_2}\in\cp_2(P_\bfx)\setminus\cK(P_\bfx)}\max_{\substack{P_\bfu,P_{\bfx|\bfu}\colon\\\sqrbrkt{P_\bfu P_{\bfx|\bfu}^{\otimes 2}}_{\bfx_1,\bfx_2} = P_{\bfx_1,\bfx_2}}}\min_{P_{\bfu,\bfx_1,\bfx_1}\in\Kcloud(P_{\bfu,\bfx})}D\paren{P_{\bfu,\bfx_1,\bfx_2}\left\|P_\bfu P_{\bfx|\bfu}^{\otimes 2}\right.},
\end{align*}
where 
\begin{align*}
    \Kcloud({P}_{\bfu,\bfx})\coloneqq&\curbrkt{[P_{\bfu,\bfx_1,\bfx_2,\bfs_1,\bfs_2,\bfy}]_{\bfu,\bfx_1,\bfx_2}\in\Delta\paren{\cU\times\cX^2} 
    \colon \begin{array}{rl}
        P_{\bfu,\bfx_1,\bfx_2,\bfs_1,\bfs_2,\bfy}\in&\Delta\paren{\cU\times\cX^2\times\cS^2\times\cY}\\
        P_{\bfs_1},P_{\bfs_2} \in&\lambda_\bfs  \\
        P_{\bfu,\bfx_1,\bfs_1,\bfy} =&{P}_{\bfu,\bfx}P_{\bfs_1|\bfu,\bfx_1}W_{\bfy|\bfx_1,\bfs_1}\\
        P_{\bfu,\bfx_2,\bfs_2,\bfy} =&{P}_{\bfu,\bfx}P_{\bfs_2|\bfu,\bfx_2}W_{\bfy|\bfx_2,\bfs_2}
    \end{array}}.
\end{align*}
\begin{remark}
The reason that \cite{wang-budkuley-bogdanov-jaggi-2019-omniscient-avc} has to define a different confusability set $\Kcloud$ when cloud code is using is that as a part of the code design, the distributions $P_\bfu,P_{\bfu|\bfx}$ are revealed to every party, including the adversary, hence he may be able to inject  noise patterns that are potentially more malicious compared with the case where he does not have such knowledge. We refer the readers to the proof in \cite{wang-budkuley-bogdanov-jaggi-2019-omniscient-avc}.
\end{remark}
In the bit-flip setting, it is easy to verify that
\begin{align*}
    \Kcloud({P}_{\bfu,\bfx})=&\curbrkt{P_{\bfu,\bfx_1,\bfx_2}\in\Delta\paren{\cU\times\cX^2}\colon \begin{array}{rl}
        P_{\bfu,\bfx_1}= P_{\bfu,\bfx_2}=&{P}_{\bfu,\bfx}  \\
        P_{\bfx_1,\bfx_2}(0,1)+P_{\bfx_1,\bfx_2}(1,0)\le  &2p
    \end{array}}\\
    =&\curbrkt{p\in\bR^{\cardU\times2\times2}\colon \begin{array}{rl}
        p_{u,x_1,x_2}\ge&0,\;\forall u,x_1,x_2\\
        \sum_{u,x_1,x_2}p_{u,x_1,x_2}=&1\\
        \sum_{x_2}p_{u,x_1,x_2} =&{ p}_{u,x_1},\;\forall u,x_1 \\
        \sum_{x_1}p_{u,x_1,x_2}= &{p}_{u,x_2},\;\forall u,x_2\\
        \sum_{u}p_{u,0,1}+p_{u,1,0}\le& 2p
    \end{array}}.
\end{align*}
We use the notation $p_{u,x_1,x_2}\coloneqq P_{\bfu,\bfx_1,\bfx_2}(u,x_1,x_2)$ and ${p}_{u,x}\coloneqq P_{\bfu,\bfx}(u,x)$ for all $u\in\cU,x_1,x_2\in\curbrkt{0,1}$.
The third maximization is over all extensions which correspond to $\cp$ decompositions of $P_{\bfx_1,\bfx_2}$. Note that for a $\cp$ matrix, its $\cp$ decomposition is not necessarily unique, even if we require the decomposition to meet the $\cprk$ \cite{groetzner-dur-2018-fac-cp}. A $\cp$ decomposition of a $\cp$ distribution can contain an arbitrarily large number of terms. Here we focus on decompositions which \emph{meet} the $\cprk$ of $P_{\bfx_1,\bfx_2}$. That is, $\cardU = \cprk(P_{\bfx_1,\bfx_2})$.

Note that the objective function KL-divergence also equals
\[D\paren{P_{\bfu,\bfx_1,\bfx_2}\left\|P_\bfu P_{\bfx|\bfu}^{\otimes 2}\right.} = I\paren{\bfx_1;\bfx_2|\bfu},\]
where the mutual information is w.r.t. $P_{\bfu,\bfx_1,\bfx_2}$.



Note that even if we could show $\Rcloud\le\RGV$, this does \emph{not} mean cloud codes will never attain a rate larger than the GV bound. It only means that the cloud rate expression we have cannot take values larger than the GV bound. This is because our bounds are only achievable, but we do not have matching upper bounds. Indeed, this is an extremely difficult question even under simple models.

Actually \emph{all} $\cp$ decompositions {meeting the $\cprk$}  of a $\cp$ distribution can be computed.  For a $\cprk$-2 distribution  $\begin{bmatrix}1-w-b&b\\b&w-b\end{bmatrix}\in{\cp_2\paren{w}}\setminus\cK\paren{w,p}  $ where $b\ne w-w^2$, we write its $\cp$ decomposition as
\begin{align*}
    \begin{bmatrix}1-w-b&b\\b&w-b\end{bmatrix}=&\alpha\inputdistr{u}^{\otimes 2} + \beta\inputdistr{v}^{\otimes2}\\
    =&\begin{bmatrix}\alpha(1-u)^2+\beta(1-v)^2&\alpha u(1-u)+\beta v(1-v)\\
    \alpha u(1-u)+\beta v(1-v)&\alpha u^2+\beta v^2\end{bmatrix}.
\end{align*}
Solving the equation in terms of $b$ and $u$, we have
\begin{align*}
    \alpha\coloneqq&\alpha(w,b,u)=\frac{w-b-w^2}{u^2+w-2uw-b},\\
    \beta\coloneqq&\beta(w,b,u)=1-\alpha=\frac{(u-w)^2}{u^2+w-2uw-b},\\
    v\coloneqq&v(w,b,u)=\frac{b-w+uw}{w-u},
\end{align*}
where 
$u\in\sqrbrkt{0,\frac{b}{1-w}}\cup\sqrbrkt{\frac{w-b}{w},1}$.

Any such decomposition gives rise to a joint distribution $P_{\bfu}P_{\bfx|\bfu}^{\otimes2}$ which is a $2\times2\times2$ tensor.
\begin{align*}
    P_{\bfu=0}P_{\bfx|\bfu=0}^{\otimes2}=\begin{bmatrix}
    \alpha (1-u)^2&\alpha u(1-u)\\
    \alpha u(1-u)&\alpha(1-u)^2
    \end{bmatrix},\quad
    P_{\bfu=1}P_{\bfx|\bfu=1}^{\otimes2}=&\begin{bmatrix}
    \beta v^2&\beta v(1-v)\\
    \beta v(1-v)&\beta(1-v)^2
    \end{bmatrix}.
\end{align*}
It also induces a distribution $P_{\bfu,\bfx}$.
\begin{align*}
    P_{\bfu,\bfx}=&\begin{bmatrix}
    \alpha(1-u)&\alpha u\\
    \beta(1-v)&\beta v
    \end{bmatrix}.
\end{align*}

Now for any $\cp$ decomposition $P_{\bfu,\bfx_1,\bfx_2}$ of a  $\cp$ distribution $P_{\bfx_1,\bfx_2}=\begin{bmatrix}w-b&b\\b&1-w-b\end{bmatrix}$, the inner minimization can be written as minimizing a convex function over a polytope.
\begin{align*}
    \begin{array}{rl}
        \min_p & D(p\|P_\bfu P_{\bfx|\bfu}^{\otimes2}) \\
        \text{subject to} & p\in\Kcloud(P_{\bfu,\bfx})
    \end{array}.
\end{align*}
It can be expanded in the following explicit form.
\begin{align*}
    \begin{array}{rl}
        \min_p & p_{0,0,0}\log\frac{p_{0,0,0}}{\alpha (1-u)^2} + p_{0,0,1}\log\frac{p_{0,0,1}}{\alpha u (1-u)} + p_{0,1,0}\log\frac{p_{0,1,0}}{\alpha u(1-u)} + p_{0,1,1}\log\frac{p_{0,1,1}}{\alpha u^2} \\
        & + p_{1,0,0}\log\frac{p_{1,0,0}}{\beta (1-v)^2} + p_{1,0,1}\log\frac{p_{1,0,1}}{\beta v(1-v)} + p_{1,1,0}\log\frac{p_{1,0,0}}{\beta v(1-v)} + p_{1,1,1}\log\frac{p_{1,1,1}}{\beta v^2}\\
        \text{subject to}
        &
        \begin{rcases}
        p_{i,j,k}\ge0,\;\forall i,j,k\\
        \sum_{i,j,k}p_{i,j,k}=1 
        \end{rcases} p\in\Delta\paren{\curbrkt{0,1}^3} \\
        &
        \begin{rcases}
        p_{0,0,0} + p_{0,0,1} = \alpha (1-u)\\
        p_{0,1,0} + p_{0,1,1} = \alpha u\\
        p_{1,0,0}+p_{1,0,1}=\beta (1-v)\\
        p_{1,1,0} + p_{1,1,1}=\beta v 
        \end{rcases} \sqrbrkt{P_{\bfu,\bfx_1,\bfx_2}}_{\bfu,\bfx_1} = P_{\bfu,\bfx} \\
        &
        \begin{rcases}
        p_{0,0,0}+p_{0,1,0}=\alpha (1-u)\\
        p_{0,0,1}+p_{0,1,1}=\alpha u\\
        p_{1,0,0}+p_{1,1,0}=\beta (1-v)\\
        p_{1,0,1}+p_{1,1,1}=\beta v
        \end{rcases} \sqrbrkt{P_{\bfu,\bfx_1,\bfx_2}}_{\bfu,\bfx_2} = P_{\bfu,\bfx} \\
        &
        p_{0,0,1}+p_{0,1,0}+p_{1,0,1}+p_{1,1,0}\le 2p.
    \end{array}
\end{align*}

Note that it is implied by the given constraints that $p_{u,x_1,x_2}=p_{u,x_2,x_1}$. Also, the p.m.f. constraint $\sum_{u,x_1,x_2}p_{u,x_1,x_2}=1$ is actually redundant. Hence the problem can be simplified as follows.
\begin{align*}
    \begin{array}{rl}
        \min_{\vp} & p_{0,0,0}\log\frac{p_{0,0,0}}{\alpha (1-u)^2} + 2p_{0,0,1}\log\frac{p_{0,0,1}}{\alpha u (1-u)} +  p_{0,1,1}\log\frac{p_{0,1,1}}{\alpha u^2} \\
        & + p_{1,0,0}\log\frac{p_{1,0,0}}{\beta (1-v)^2} + 2p_{1,0,1}\log\frac{p_{1,0,1}}{\beta v(1-v)}  + p_{1,1,1}\log\frac{p_{1,1,1}}{\beta v^2}\\
        \text{subject to}
         &-p_{i,j,k}\le0,\;\forall i,j,k\\
         & p_{0,0,0}+p_{0,0,1}=\alpha u \\
         & p_{0,0,1} + p_{0,1,1}=\alpha (1-u)\\
         & p_{1,0,0}+p_{1,0,1}=\beta v\\
         &p_{1,0,1}+p_{1,1,1}=\beta(1-v)\\
         &p_{0,0,1}+p_{1,0,1}\le p.
    \end{array}
\end{align*}
Let $D^*(w,b,u)$ denote the optimal value of the above minimization. The final cloud rate is given by 
\begin{align*}
    \max_{0<w<1}\max_{p<b\le w-w^2}\max_{u\in\sqrbrkt{0,\frac{b}{1-w}}\cup\sqrbrkt{\frac{w-b}{w},1}} &D^*(w,b,u),
\end{align*}
where the first maximization corresponds to finding the optimal input distribution $\inputdistr{w}$, the second maximization corresponds to finding the optimal $\cp$ distribution $\begin{bmatrix}1-w-b&b\\b&w-b\end{bmatrix}$ outside $\cK(w)$, and the third optimization corresponds to finding the optimal $\cp$-decomposition $\alpha\inputdistr{u}^{\otimes2}+\beta\inputdistr{v}^{\otimes2}$ of the optimal $\cp$ distribution.

\section{Concluding remarks and open problems}
\label{sec:rk_and_open}
In this paper, we study list decoding problem  on general adversarial channels for both large and small list sizes. Given any channel, for large (yet constant) list sizes, we prove the list decoding theorem which identifies the fundamental limit of  list decoding. For small (yet arbitrary universal constant) list sizes, we characterize when positive rate list decodable codes are possible.

Many open questions are left after this work is done. We list some of them for future study.
\begin{enumerate}
	\item In this paper, we made no attempt towards understanding channels with arbitrary transition distributions $W_{\bfy|\bfx,\bfs}$ (instead of only those corresponding to deterministic bivariate functions). Pushing our results to such a  general setting remains an intriguing open question.
	\item Other adversarial channels under further assumptions, e.g., online (causal) channels, channels with feedback, channels with bounded memory, etc. are less understood. There are results regarding each of these topics under very restricted models, e.g., bit-flips \cite{chen-et-al,berlekamp1964thesis}, deletions \cite{bukh-guruswami-hastad-2016-adv-deletion-plotkin}, etc.
	\item We do not have any nontrivial \emph{upper} bound on $(L-1)$-list decoding capacity for general adversarial channels. Existing upper bounds for error correction codes seem tricky to generalize. A reasonable starting point might be to extend the classic Elias--Bassalygo bound \cite{bassalygo} whose proof has a similar spirit as the Plotkin bound.
	\item Given any adversarial channel, when we are ``below the Plotkin point'' (i.e., there are non-confusable $\cp$ distributions), can we construct \emph{explicit}  codes of positive rate? We know that random codes is list decodable w.h.p.
\end{enumerate}

\section{Acknowledgement}
\label{sec:ack}
We thank  Andrej Bogdanov who provided elegant reduction from general $L$ to $L=2$ for the proof the asymmetric case of the converse  (Sec. \ref{lem:converse_asymm}) and reconstructed Blinovsky's \cite{blinovsky-1986-ls-lb-binary} characterization of $P_{L-1}$ via conceptually cleaner proof, despite that he generously declined to co-author this paper. We also thank him for inspiring discussions  in the early stage  and helpful comments near the end of this work. 

Part of this work was done  while YZ was visiting the Simons Institute for the Theory of Computing for the Summer Cluster: Error-Correcting Codes and High-Dimensional Expansion.





\appendices
\section{$\cp$ tensors and $\cop$ tensors}
\label{app:cp_cop}

\subsection{Tensor products}
\begin{definition}[Tensor product]
\label{def:tensor_prod}
For two tensors $A\in\ten_n^\tm,B\in\ten_n^{\otimes \ell}$, Their \emph{tensor product} is defined as
\[A\otimes B\coloneqq\sqrbrkt{A\paren{i_1,\cdots,i_m}B\paren{j_1,\cdots,j_\ell}}\in\ten_n^{\otimes\paren{m+\ell}}.\]
\end{definition}
\begin{definition}[Frobenius inner product, Frobenius norm]
\label{def:frob_ip_norm}
For two tensors $A,B\in\ten_n^\tm$, Their \emph{inner product} is defined as
\[\inprod{A}{B}\coloneqq\sum_{i_1,\cdots,i_m\in[n]}A(i_1,\cdots,i_m)B(i_1,\cdots,i_m).\]
The \emph{Frobenius norm} is defined as $\normf{A}\coloneqq\sqrt{\inprod{A}{A}}$.
\end{definition}
\begin{definition}[Hadamard product]
\label{def:had_prod}
For two tensors $A,B\in\ten_n^\tm$, Their \emph{Hadamard product} is defined as
\[A\circ B\coloneqq[A(i_1,\cdots,i_m)B(i_1,\cdots,i_m)]\in\ten_n^\tm.\]
\end{definition}

\subsection{Tensor decomposition}
\begin{definition}[Canonical decomposition]
\label{def:canonical_decomp}
For a tensor $A\in\ten_n^\tm$, its \emph{canonical decomposition} has form
\[A=\sum_{j=1}^r\alpha_j\bigotimes_{i=1}^m\vx_{j,i},\]
where each $\vx_{j,i}\in\bS^{n-1}_2$. The smallest $r$ for $A$ to admit such a decomposition is called the \emph{rank} of $A$. If $A$ is symmetric, then 
\[A=\sum_{j=1}^r\alpha_j\vx_j^\tm\] 
is an analog of the eigendecomposition of symmetric matrices. The smallest $r$ is called the \emph{symmetric rank} of $A$.
\end{definition}
\begin{conjecture}
For $A\in\sym_n^\tm$, $\rk(A)=\symrk(A)$.
\end{conjecture}
\begin{remark}
It is known to be true if $\rk(A)\le m$.
\end{remark}

\begin{definition}[Tucker decomposition]
\label{def:tucker_decomp}
For a tensor $A\in\ten_n^\tm$, the \emph{Tucker decomposition} has form
\[A=\sum_{j_1=1}^{r_1}\cdots\sum_{j_m=1}^{r_m}\alpha_{j_1,\cdots,j_m}\bigotimes_{i=1}^m\vx_{j_i,j}.\]
It is an analogy of the singular value decomposition of matrices.
\end{definition}

A tensor $A\in\ten_n^\tm$ has $n(m-1)^{n-1}$ eigenvalues. $A$ may have non-real eigenvalues even if $A$ is symmetric. If an eigenvector is real, then the corresponding eigenvalue is also real. Such eigenvalues are called \emph{$H$-eigenvalues}. They always exist for even-order tensors.

\subsection{Special tensors}
\begin{definition}[$\nn$ tensors]
\label{def:nn_tensor}
A tensor is said to be \emph{non-negative} if each of its entry is non-negative. The set of order-$m$ dimension-$n$ non-negative tensors is denoted by $\nn_n^{\otimes m}$
\end{definition}

\begin{definition}[$\psd$ tensors, $\pd$ Tensors]
\label{def:psd_pd_tensor}
For even $m$, $A\in\ten_n^\tm$ is \emph{positive semidefinite ($\psd$)} if 
$\inprod{A}{\vx^{\otimes m}}\ge0$
for any $\vx\in\bR^n$. 
$A$ is \emph{positive definite ($\pd$)} if the above inequality is strict for all $\vx\ne0$.

The sets of $\psd$ and $\pd$ tensors is denoted by $\psd_n^{\otimes m}$ and $\pd_n^{\otimes m}$, respectively.
\end{definition}



\begin{definition}[$\cp$ tensors, $\cp$ tensor rank]
\label{def:def_cp_tensor}
A tensor $P\in\ten_n^\tm$ is said to be \emph{completely positive} if for some $r\ge1$, there are component-wise non-negative vectors $\vp_1,\cdots,\vp_r\in\bR_{\ge0}^n$ such that
\[P=\sum_{j=1}^r\vp_j^\tm.\]
The set of $\cp$ tensors is denoted by $\cp_n^\tm$. The least $r$ such that $P$ has a completely positive decomposition is called the \emph{$\cprk$} of $P$. If $\Span\curbrkt{P_1,\cdots,P_r}=\bR^n$ then $P$ is said to be \emph{strongly $\cp$}. 
\end{definition}

\begin{fact}
Verifying if a symmetric non-negative tensor is $\cp$ is $\np$-hard.
\end{fact}

\begin{definition}[$\cop$ tensors]
\label{def:cop_tensor}
$A\in\sym_n^\tm$ is \emph{copositive} if $\inprod{A}{\vx^{\otimes}}\ge0$ for all $\vx\in\bR_{\ge0}^n$. The set of copositive tensors is denoted by $\cop_n^{\otimes m}$.
\end{definition}

\begin{theorem}[Duality]
\label{thm:cp_cop_duality}
$\cp_n^\tm$ and $\cop_n^\tm$ are closed convex pointed cones with nonempty interior in $\sym_n^\tm$. For $m\ge2$, $n\ge1$, they are dual to each other.
\end{theorem}

\begin{definition}[$\dnn$ tensors]
\label{def:dnn_tensor}
For even $m$, $A\in\sym_n^\tm$ is \emph{doubly non-negative ($\dnn$)} if  $A$ is entry-wise non-negative and $\inprod{A}{\vx^{\otimes m}}$ is a sum-of-square as a polynomial in the components of $\vx$. 
\end{definition}

\begin{fact}
The double non-negativity of a tensor can be verified in polynomial time using SDP. 
\end{fact}

\begin{fact}
The following inclusion relations between different sets of special tensors hold.
\begin{enumerate}
	\item $\psd_n^\tm\subseteq\cop_n^\tm$.
	\item $\cp_n^\tm\subseteq\dnn_n^\tm\subseteq\nn_n^\tm\subseteq\cop_n^\tm\subseteq\sym_n^\tm$.
\end{enumerate}
\end{fact}

\section{Hypergraph Ramsey numbers}
\label{app:bound_hypergraph_ramsey_number}
Let $R_k^{(r)}(s_1,\cdots,s_k)$ denote the smallest size of an $r$-uniform hypergraph such that for any $k$-colouring, there must be a monochromatic clique of size $s_i$ for some $i\in[k]$.

Define tower function $t_1(x) = x$ and $t_{i+1}(x) = 2^{t_i(x)}$.

\begin{lemma}[Properties of hypergraph Ramsey numbers]
\label{lem:properties_hypergraph_ramsey_number}
\begin{enumerate}
    \item For any $i\in[k]$, and $s_j\ge r$ ($j\ne i$),
    \begin{align}
        R_k^{(r)}(s_1,\cdots,s_{i-1},r,s_{i+1},\cdots,s_k)=&R_{k-1}^{(r)}(s_1,\cdots,s_{i-1},s_{i+1},\cdots,s_k).\notag
    \end{align}
    \item For any $\sigma\in S_{k}$,
    \begin{align}
        R_k^{(r)}(s_1,\cdots,s_k)=&R_k^{(r)}(s_{\sigma(1)},\cdots,s_{\sigma(k)}).\notag
    \end{align}
\end{enumerate}
\end{lemma}

\begin{lemma}[Finiteness of hypergraph Ramsey numbers]
\label{lem:finiteness_hypergraph_ramsey_number}
For any positive integers $r,k,s_1,\cdots,s_k$, the hypergraph Ramsey number $R_k^{(r)}(s_1,\cdots,s_k)$ is finite. In particular, it satisfies the following recursive inequalities.
\begin{align*}
R_k^{(r)}(s_1,\cdots,s_k)\le&1+R_{k}^{(r-1)}\paren{R^{(r)}_k(s_1-1,s_2,\cdots,s_k),R_k^{(r)}(s_1,s_2-1,\cdots,s_k),\cdots,R_k^{(r)}(s_1,s_2,\cdots,s_k-1)} ,\notag\\
R_k^{(r)}(s_1,\cdots,s_k)\le&1+\sum_{i=1}^kR_k^{(r-1)}\paren{R_k^{(r)}(s_1,\cdots,s_{i-1},s_{i}-1,s_{i+1},\cdots,s_k),\cdots,R_k^{(r)}(s_1,\cdots,s_{i-1},s_{i}-1,s_{i+1},\cdots,s_k)},\notag\\
R_k^{(r)}(s_1,\cdots,s_k)\le&R_{k-1}^{(r)}\paren{s_1,\cdots,s_{k-2},R_2^{(r)}(s_{k-1},s_k)},\notag
\end{align*}
\end{lemma}

\begin{lemma}[Bounds on  hypergraph Ramsey numbers]
\label{lem:bounds_hypergraph_ramsey_number}
\begin{enumerate}
	\item For any $s,t$,
    \begin{align*}
        R_2^{(r)}(s,t)\le&2^{\binom{R_2^{(r-1)}(s-1,t-1)}{r-1}}.
    \end{align*}
	\item For $r\ge3$, there are constants $c,c'>0$ such that \begin{align}
        t_{r-1}(c\cdot s^2)\le&R_2^{(r)}(s,s)\le t_r(c'\cdot s).\notag
    \end{align}
    \item For $s>k\ge2$, there are constants $c,c'>0$ such that
    \[t_r(c\cdot k)<R_k^{(r)}(s,\cdots,s)<t_r(c'\cdot k\log k).\]
\end{enumerate}
\end{lemma}





\section{Expected translation distance of a one-dimensional random walk}
\label{app:rw_dist}
\begin{lemma}
Consider a random walk $\bfx_1,\cdots,\bfx_L$ of length $L$. Each $\bfx_i$ ($1\le i\le L$) is an independent and identically distributed $\curbrkt{-1,1}$-valued random variable satisfying
\[\prob{\bfx_i = 1} = p,\quad\prob{\bfx_i = -1} = 1-p.\]
Without loss of generality, assume $p\ge1/2$. Then the expected translation distance $\expt{\abs{\bfx_1 + \cdots + \bfx_L}}$ of this random walk after $L$ steps
is minimized when $p=1/2$.
\end{lemma}
\begin{proof}
Create another walk $\bfx_1',\cdots,\bfx_L'$ with $p=1/2$ that is coupled with $\bfx_1,\cdots,\bfx_L$ in the following way.
\[ \prob{\bfx_i = 1|\bfx_i' = 1} = 1,\quad\prob{\bfx_i = 1|\bfx_i' = -1} = 2p-1. \]
It is easy to see that the distribution of $\bfx_1,\cdots,\bfx_L$ is preserved under this coupling. 
\begin{align*}
    \prob{\bfx_i = 1} =& \prob{\bfx_i' = 1} \prob{\bfx_i = 1|\bfx_i' = 1} + \prob{\bfx_i' = -1} \prob{\bfx_i = 1|\bfx_i' = -1}\\
    =& \frac{1}{2}\cdot1 + \frac{1}{2}\cdot(2p-1)\\
    =& p.
\end{align*}
Now,
\begin{align*}
    &\expt{\abs{\bfx_1 + \cdots + \bfx_L}} - \expt{\abs{\bfx_1' + \cdots + \bfx_L'}}\\
    = &\sum_{d\in\curbrkt{-L,-L+2,\cdots,L-2,L}}\sum_{\substack{x_1,\cdots,x_L\in\curbrkt{-1,1}\\\sum_{i=1}^Lx_i = d}}\prob{\bfx_1' = x_1,\cdots,\bfx_L' = x_L}\expt{\left.\abs{\sum_{i=1}^L\bfx_i} - \abs{d} \right|\bfx_1' = x_1,\cdots,\bfx_L' = x_L}.
\end{align*}

For each translation distance $d\in\curbrkt{-L,-L+2,\cdots,L-2,L}$ and trajectory $x_1,\cdots,x_L\in\curbrkt{-1,1}$ such that $\sum_{i=1}^Lx_i=d$, 
let $\ell\coloneqq\curbrkt{i\in[L]\colon x_i = -1}$. Note $2(d+\ell)=L$.
We have
\begin{align*}
    \expt{\left.\abs{\sum_{i=1}^L\bfx_i} - \abs{d} \right|\bfx_1' = x_1,\cdots,\bfx_L' = x_L} = & ((2p-1)\cdot1 + (1-(2p-1))\cdot(-1))\ell - (-\ell)\\
    =&2(2p-1)\ell,
\end{align*}
which is non-negative and attains its minima $0$ when $p=1/2$. This finishes the proof.
\end{proof}

\section{Blinovsky \cite{blinovsky-1986-ls-lb-binary} vs. Alon--Bukh--Polyanskiy \cite{alon-bukh-polyanskiy-2018-ld-zero-rate}}
\label{app:blinovsky_abp}
In this section we show that, though differing ostensibly, the formulas of the Plotkin points for $(p,L-1)$-list decoding  given by Blinovsky and Alon--Bukh--Polyanskiy actually agree with each other. The proof is essentially due to the user \texttt{Marko Riedel} on \texttt{Mathematics Stack Exchange} \cite{blinovsky-vs-abp-mathstackexchange}.

For $L = 2k$ or $2k + 1$ for some positive integer $k\in\bZ_{>0}$, Blinovsky's formula is
\[P_{L-1}= \sum_{i = 1}^k\frac{\binom{2(i - 1)}{i - 1}}{i} 2^{-2i};\]
while Alon--Bukh--Polyanskiy wrote it as
\[P_{L - 1} = \frac{1}{2} - 2^{-2k - 1}\binom{2k}{k}.\]
We are going to show that
\begin{lemma}
For any $k\ge1$, 
\[\sum_{i = 1}^k\frac{\binom{2(i - 1)}{i - 1}}{i} 2^{-2i} = \frac{1}{2} - 2^{-2k - 1}\binom{2k}{k}.\]
\end{lemma}
\begin{proof}
To see the above two expressions are always evaluated to the same value,
we first massage the above equation. Multiplying $2^{2k + 2}$ on both sides, shifting the summation index and rearranging terms, we have
\[\sum_{i = 0}^{k - 1}\frac{\binom{2i}{i}}{i + 1}2^{2(k - i)} = 2^{2k+1} - 2\binom{2k}{k}. \]
Adding $\frac{\binom{2k}{k}}{k + 1}$ on both sides, we get
\begin{align}
    \sum_{i = 0}^k\frac{\binom{2i}{i}}{i + 1}2^{2(k - i)} =& 2^{2k + 1} - \paren{2 - \frac{1}{k + 1}}\binom{2k}{k} \notag\\
    =& 2^{2k + 1} - \frac{2k + 1}{k + 1}\binom{2k}{k} \notag\\
    =&2^{2k + 1} - \binom{2k + 1}{k + 1}\label{eqn:comb_id_1}\\
    =&2^{2k + 1} - \binom{2k + 1}{k},  \label{eqn:comb_id_2}
\end{align}
where Eqn. \eqref{eqn:comb_id_1} is by Fact \eqref{eqn:binom_recurse} and Eqn. \eqref{eqn:comb_id_2} is by Fact \eqref{eqn:binom_symm}.

To show 
\begin{equation}
    \sum_{i = 0}^k\frac{\binom{2i}{i}}{i + 1}2^{2(k - i)} = 2^{2k + 1} - \binom{2k + 1}{k},
    \label{eqn:bli_abp_id}
\end{equation}
we conduct induction on $k$.
\begin{enumerate}
    \item When $k = 0$, LHS = $1$ = RHS. 
    \item Assume \eqref{eqn:bli_abp_id} 
    holds for certain $k\ge 1$. We want to show it also holds for $k + 1$.
    \begin{align}
        \sum_{i = 0}^{k + 1}\frac{\binom{2i}{i}}{i + 1}2^{2(k + 1 - i)}
        =&2^2\sum_{i=0}^k\frac{\binom{2i}{i}}{i + 1}2^{2(k -i)} + \frac{\binom{2(k + 1)}{k + 1}}{k + 2}\notag\\
        =&2^2\paren{2^{2k + 1} - \binom{2k + 1}{k}} + \frac{\binom{2k + 2}{k + 1}}{k + 2}\label{eqn:ind_hyp}\\
        =&2^{2(k + 1) + 1}-2\paren{\binom{2k + 1}{k} + \binom{2k + 1}{k + 1}} + \frac{\binom{2k + 2}{k + 1}}{k + 2}\label{eqn:symm_app}\\
        =&2^{2(k + 1) + 1} - \paren{2-\frac{1}{k + 2}}\binom{2k + 2}{k + 1}\label{eqn:pascal_app}\\
        =&2^{2(k + 1) + 1}-\frac{2k + 3}{k + 2}\binom{2k + 2}{k + 1}\notag\\
        =&2^{2(k + 1) + 1} - \binom{2(k + 1) + 1}{(k  + 1) + 1}.\label{eqn:recurse_app}
    \end{align}
Eqn. \eqref{eqn:ind_hyp}, \eqref{eqn:symm_app}, \eqref{eqn:pascal_app} and \eqref{eqn:recurse_app} follow from induction hypothesis, Fact \eqref{eqn:binom_symm}, Fact \eqref{eqn:binom_pascal} and Fact \eqref{eqn:binom_recurse}, respectively. Hence Eqn. \eqref{eqn:bli_abp_id} holds for $k + 1$ as well. 
\end{enumerate}
\end{proof}

\bibliographystyle{alpha}
\bibliography{IEEEabrv,ref} 

\end{document}